\documentclass[aps,prd,10pt,notitlepage,nofootinbib,superscriptaddress,showkeys,showpacs]{revtex4-1}

\usepackage{amsmath,amssymb,amsthm,latexsym,bbm}
\usepackage[english]{babel}
\usepackage{color}
\usepackage{xspace}
\usepackage{graphicx}
\usepackage{pifont,dsfont}
\usepackage{marvosym}
\usepackage{slashed}
\usepackage{enumitem}
\usepackage{multirow}
\usepackage{subcaption}
\usepackage{caption}
\usepackage{hyperref}

\theoremstyle{definition}
\newtheorem{lemma}{Lemma}

\newtheorem{theorem}{Theorem}

\newtheorem{proposition}{Proposition}

\newtheorem*{theorem*}{Theorem}

\DeclareMathOperator{\tr}{tr}

\makeatletter
\renewcommand\paragraph{\@startsection{paragraph}{4}{\z@}%
                                     {3ex\@plus 1ex \@minus -.2ex}%
                                     {-1ex \@plus .2ex}%
                                     {\normalfont\normalsize\bfseries}}
\makeatother

\begin{document}

\title{\Large \bf Tensor models with generalized melonic interactions}

\author{{\bf Valentin Bonzom}}\email{bonzom@lipn.univ-paris13.fr}
\affiliation{LIPN, UMR CNRS 7030, Institut Galil\'ee, Universit\'e Paris 13,
99, avenue Jean-Baptiste Cl\'ement, 93430 Villetaneuse, France, EU}

\date{\small\today}

\begin{abstract}
\noindent Tensor models are natural generalizations of matrix models. The interactions and observables in the case of unitary invariant models are generalizations of matrix traces. Some notable interactions in the literature include the melonic ones, the tetrahedral one as well as the planar ones in rank three, or necklaces in even ranks. Here we introduce generalized melonic interactions which generalize the melonic and necklace interactions. We characterize them as tree-like gluings of quartic interactions. We also completely characterize the Feynman graphs which contribute to the large $N$ limit. For a subclass of generalized melonic interactions called totally unbalanced interactions, we prove that the large $N$ limit is Gaussian and therefore the Feynman graphs are in bijection with trees. This result further extends the class of tensor models which fall into the Gaussian universality class. Another key aspect of tensor models with generalized melonic interactions is that they can be written as matrix models without increasing the number of degrees of freedom of the original tensor models. In the case of totally unbalanced interactions, this new matrix model formulation in fact decreases the number of degrees of freedom, meaning that some of the original degrees of freedom are effectively integrated. We then show how the large $N$ Gaussian behavior can be reproduced using a saddle point analysis on those matrix models.
\end{abstract}

\medskip

\keywords{Tensor models, Matrix models, Bubbles, Large $N$ limit, Intermediate field}

\maketitle

\section*{Introduction}

Random tensor models \cite{GurauBook} are generalization of random matrix models \cite{MatrixReview}. Here we consider the complex case, which for matrices typically corresponds to a joint distribution of the form $\exp -\tr V(MM^\dagger)$ where $V$ is a polynomial \cite{ColoredMatrixModels}. Such a distribution has a $U(N)\times U(N)$ invariance, $M\mapsto UMV^\dagger$. The natural generalization to tensors, introduced in \cite{Uncoloring}, is to consider a joint distribution of the form $\exp - V(T, \bar{T})$ where $V$ is invariant under $U(N)^d$ transformations. Here $T$ is a complex tensor with $d$ indices ranging from 1 to $N$. In particular, it has $N^d$ degrees of freedom. We consider polynomials of the form
\begin{equation} \label{Action}
V(T, \bar{T}) = \sum_{a_1, \dotsc, a_d} T_{a_1 \dotsb a_d} \bar{T}_{a_1 \dotsb a_d} + N^s\,t\, B(T, \bar{T})
\end{equation}
where $B(T, \bar{T})$ is called a bubble polynomial. Bubble polynomials are in bijection with $d$-regular, bipartite, edge-colored, connected graphs called {\bf bubbles} \cite{Uncoloring} which simply describe the pattern of index contractions of the polynomials.

The typical first question when trying to solve a tensor model is the large $N$ limit: evaluate expectations of $U(N)^d$-invariant polynomials at large $N$. To do so, one cannot recourse to typical methods of random matrices such as eigenvalues. Instead, a natural method for physicists is the Feynman expansion. While it is not a rigorous tool to study the convergence of random tensors at large $N$ in a probabilistic setting, it is expected to give correct predictions (like in random matrices).

Furthermore, the Feynman expansion connects random tensors to random triangulations of $d$-dimensional (pseudo-)manifolds \cite{GurauBook, Uncoloring, SigmaReview, 3D}, just like it does between random matrices and random combinatorial maps. It means that results obtained in this setting for random triangulations have a straight up interpretation in random tensors. For instance, the generating functions of rooted planar maps correspond to the calculation of 2-point functions in matrix models \cite{MatrixModelsCombinatorics}; similarly, 2-point functions of tensor models are this way expressed as generating functions of some connected, face-colored triangulations. For combinatorial purposes, it is easier and sufficient to consider the dual 1-skeletons of those triangulations, which are edge-colored, $d$-regular graphs \cite{ItalianSurvey,LinsMandel,TopologyTensor-Casali-Cristofori-Dartois-Grasselli}. Another way to think about it is as the Feynman expansion being a definition of our objects of interest, which we could call ``combinatorial tensor models'' to emphasize their definition in terms of Feynman graphs.

In contrast with matrix models, it is actually a difficult task to find a non-trivial large $N$ limit in tensor models, the first non-trivial one being due to Gurau \cite{1/NExpansion}. Indeed, if $s$ in \eqref{Action} is ``sufficiently'' small, a large $N$ limit exists but is trivial, in the sense that only a finite number of Feynman graphs contribute at large $N$. A non-trivial large $N$ limit is thus a limit where an infinite number of graphs contribute to the expectations of observables. It was shown in \cite{PhDLionni} by Lionni that if a large $N$ limit for a given bubble interaction, then there is a unique value of $s$ which makes it non-trivial.

Such tensor models have been studied and understood at large $N$ in a few of cases. At $d=3$, the most general result \cite{3D} is that when $B(T, \bar{T})$ is a planar bubble, i.e. dual to a colored triangulation of the 2-sphere, the large $N$ limit is Gaussian: expectations of $U(N)^3$-invariant polynomials at large $N$ are those of a Gaussian model whose covariance is the large $N$, 2-point function.

In terms of Feynman graphs, a non-trivial large $N$ limit which is Gaussian (such as in \cite{3D}) generates graphs in bijection with trees. In contrast with such a behavior, we recall that single-trace matrix models with polynomial interactions are typically non-Gaussian at large $N$. In terms of Feynman graphs, this is due to the fact that planar maps (i.e. ribbon graphs) lie in a different universality class than trees, whose $n$-point functions are not factorizations of 2-point functions.

The first theorem showing that a large class of distributions for random tensors converge universally to a Gaussian is Gurau's universality theorem \cite{Universality}. In the context of distributions given as $\exp -V(T, \bar{T})$ as above, Gurau's universality theorem only applies to a class of bubbles known as melonic bubbles (otherwise the large $N$ limit is trivial). Melonic bubbles are special bubbles which are series-parallel and in the large $N$ limit, they give rise to melonic Feynman graphs which are themselves series-parallel and are responsible for the Gaussian large $N$ limit.

The theorem of \cite{3D} is thus an extension of \cite{Universality}, although only in $d=3$. Another extension, to tensors with indices having different ranges, is in \cite{New1/N}. In even dimensions, non-Gaussian universality classes can be reached, in particular all those of ordinary matrix models, i.e. planar maps with matter. So far they always seem to rely on planar maps and multi-matrix models. In particular, one finds phase transitions between the large $N$ Gaussian behavior and the large $N$, planar limit of matrix models \cite{SigmaReview, MelonoPlanar, LucaJohannes}. Other tensor models have been solved as large $N$ on a case by case basis, i.e. for some given interactions like the one whose interaction bubble is $K_{3,3}$, \cite{StuffedWalshMaps, PhDLionni}.

The world of tensor models with different symmetries than $U(N)^d$ has also started to be explored. An obvious generalization is to consider $O(N)^d$-invariant models. The first step in that direction was Tanasa's multiorientable which is equipped with a $U(N)\times O(N)\times U(N)$-invariance (at $d=3$), \cite{MO-Review, FusyTanasaMO}. Carrozza and Tanasa then introduced an $O(N)^3$-invariant model with tetrahedral interaction \cite{CarrozzaTanasa} (the latter is not bipartite, so does not appear in $U(N)^d$-invariant models), and its expansion beyond the large $N$ limit can be found in \cite{NadorCTKT}. Further generalizations to $O(N)\times O(D)\times O(N)$-symmetry have been obtained in \cite{FerrariRivasseauValette}. In particular, the concept of mirror melons as a  generalization of melonic graphs was introduced. Those models lead to interesting connections with SYK physics \cite{KlebanovTarnopolsky, FerrariLargeD}. Other tensor models inspired by the SYK model have been shown to have non-trivial large $N$ limits: two symmetric tensors \cite{TwoTensors-Gurau}, and tensors in irreducible representations of the symmetric group of the three indices (at $d=3$) \cite{SymTracelessTensors, Irreducible-Carrozza}.

In this article, we stay in the realm of tensor models with $U(N)^d$-invariance. We define tensor models, with non-trivial large $N$ limits, for polynomials $B(T, \bar{T})$ associated to {\bf generalized melonic} (GM) bubbles and show that they extend some universality theorems mentioned above. 

\paragraph{Characterization of GM bubbles --} In Section \ref{sec:TensorialInteractions} we define GM bubbles as a generalization of melonic bubbles. The latter can be constructed using dipole moves. As for GM bubbles, they are constructed from local moves called $C$-bidipole insertions, for $C\subset \{1, \dotsc, d\}$ and $|C|\leq d/2$.

In addition to this local structure, we can characterize their global structure, for which once again they generalize melonic bubbles: GM bubbles correspond to trees of quartic bubbles. There is indeed a canonical gluing of bubbles, which in terms of polynomials consists in removing a $T$ from one polynomial and a $\bar{T}$ from another and contracting the free indices in a canonical way. There is a general theorem stating that all bubbles can be obtained from the gluing of quartic bubbles, where quartic bubble polynomials are those quadratic in $T$ and in $\bar{T}$. This theorem was proved in \cite{StuffedWalshMaps}. In the case of GM bubbles, we can thus show the following additional property.

\begin{theorem} \label{thm:QuarticTree}
GM bubbles are exactly the gluings of quartic bubbles which form trees.
\end{theorem}

Indeed, a graph can be associated to the gluing by representing each quartic bubble as a vertex and each gluing as an edge between two vertices. A bubble is GM if and only if this graph is a tree.

\paragraph{Tensor models with GM interactions --} In Section \ref{sec:SPTensors}, we study tensor models whose interactions are GM. We start by recalling that as proved in \cite{StuffedWalshMaps}, the Feynman graphs of any tensor model form a subset of those of tensor models with quartic interactions, see Proposition \ref{thm:Surjection}. We thus study quartic models and describe their large $N$ limit in Theorem \ref{thm:Quartic}. Somewhat surprisingly this had never appeared before in the literature, although it only requires a slight generalization of the rank 4 case detailed in \cite{MelonoPlanar}. By combining those results, we find the following theorem.

\begin{theorem} \label{thm:Scaling}
There is a non-trivial large $N$ limit if and only if the scaling coefficients are
\begin{equation} \label{ScalingCoefficients}
s = \sum_{C} |C| b_C - \frac{d(V-2)}{2}
\end{equation}
where the sum is over the $C$-bidipole insertions of the bubble and $V$ is its number of vertices. The graphs contributing to the large $N$ limit are a subset of those of the large $N$ limit of the quartic models.
\end{theorem}

\paragraph{Gaussian universality at large $N$ --} The analysis at large $N$ reveals that tensor models with bubbles which have $C$-bidipoles with $|C|=d/2$ behave differently from those which have none. If $B$ has a $C$-bidipole insertion with $|C|=d/2$, it can create planar objects like in matrix models at large $N$. Otherwise, we say that $B$ is {\bf totally unbalanced} and we prove the following theorem in Section \ref{sec:TotallyUnbalanced}.

\begin{theorem} \label{thm:Gaussian}
Tensor models with totally unbalanced GM interactions are Gaussian at large $N$. The covariance satisfies an explicit polynomial equation.
\end{theorem}

This concludes the analysis at large $N$, with a new extension of the universality theorem of Gurau.

\paragraph{Matrix models --} Tensor models with GM interactions have another interesting property which is the object of Section \ref{sec:MatrixModel}. They can be recast as matrix models by a series of Hubbard-Stratonovich transformations. This gives a new intermediate field theory. Compared to the one of \cite{StuffedWalshMaps}, it only applies to GM interactions and not arbitrary interactions. However, the key point is that it does not increase the number of degrees of freedom, like \cite{StuffedWalshMaps} does. It is a longstanding discussion for tensor models which is rooted in the unitary invariance. Indeed, in unitary-invariant matrix models, although the random matrices have $N^2$ degrees of freedom, unitary invariance is such that matrix models only have $N$ degrees of freedom, the eigenvalues, after integrating over the unitary group.

However in tensor models, no such reduction of the number of degrees of freedom is known in general. It is possible to explicitly integrate some angular degrees of freedom to evaluate expectations in the Gaussian distribution as effective matrix correlations \cite{UnitaryGaussianExpectations}, but this method applies case by case only.

For models with melonic cycles for instance (such as quartic bubbles), there is an intermediate field theory which reduces the number of degrees of freedom to $N$ instead of $N^d$, see e.g. \cite{ConstructiveQuartic} where it is used to prove Borel summability of the perturbative series. The most general intermediate field theory of \cite{StuffedWalshMaps} can however, typically, increase the number of degrees of freedom up to $N^{2d}$. Although it can still be useful to study the large $N$ limit of many models, see e.g. \cite{Octahedra} for the large $N$ limit of the model with the cube as bubble (dual to the octahedron), it goes against the idea of integrating some degrees of freedom to find intermediate field matrix models.

In fact for totally unbalanced GM interactions we show that our new matrix models do reduce the number of degrees of freedom. It is multi-matrix model with matrices of sizes $N^{2|C|}$ for each $C$-bidipole insertion. The interaction can be represented as a tree $\mathcal{T}$ with edges colored by sets $C_e\subset \{1, \dotsc, d\}$ satisfying $|C|\leq d/2$. The matrix model associates a pair of matrix $(X_e, X_e^\dagger)$ to each edge of $\mathcal{T}$, or equivalently a complex matrix $X_h$ to each half-edge $h$ of the tree with the contraint $X_{h_2} = X_{h_1}^\dagger$ for $e=\{h_1 h_2\}$. We prove the following theorem.

\begin{theorem} \label{thm:MatrixModel}
The partition of a tensor model with a GM interaction writes as
\begin{equation}
Z_N(t) = \int \prod_{e\in\mathcal{T}} dX_e dX_e^\dagger\ \exp -\sum_{e\in\mathcal{T}} \tr_{V_{C_e}} (X_e X^\dagger_e) - \tr_{E_d}\ln \Biggl(\mathbbm{1}_{E_d} - \sum_{v\in\mathcal{T}} \prod_{h_v}^{\substack{\text{counter-}\\\text{-clockwise}}} (N^{s_B} t_B)^{\frac{1}{V-2}} \epsilon_h\tilde{X}_{h_v} \Biggr) 
\end{equation}
where the sum in the logarithm is over the vertices $v\in\mathcal{T}$, and the product over the half-edges $h_v$ incident to $v$, and $\epsilon_h$ are signs whose product is $-1$.
\end{theorem}

We also show that a simple ansatz allows to reproduce the Gaussian behavior at large $N$ in terms of a saddle point. As expected in a Gaussian model, the Vandermonde determinant does not contribute and all eigenvalues fall into the potential well (also why we expect our simple ansatz to be correct), see Proposition \ref{thm:SaddlePoint}.

After that, the next step would be to study the fluctuations around the saddle point, which at leading order will by definition be Gaussian, as in the case of quartic melonic models \cite{Quartic-Nguyen-Dartois-Eynard}. Most interestingly, this procedure of using the intermediate field, finding the saddle point and looking at the fluctuations around it, is the only method known so far to exhibit some version \cite{BlobbedTR-Borot, BlobbedTR-Borot-Shadrin} of the topological recursion (see \cite{QuarticTR} in the quartic melonic case) and potential connections to integrability \cite{GiventalTensor-Dartois}. This method thus requires to not work with the original tensor models but with the intermediate field model because it has fewer degrees of freedom and it is easier to understand the latter through their Schwinger-Dyson equations rather than the original ones\footnote{In comparison, Schwinger-Dyson equations of tensor models are much harder to analyze \cite{SchwingerDysonTensor, Revisiting, DoubleScaling}.}.

We can thus hope that this new intermediate field theory can help understand tensor models beyond the large $N$ limit and also shed some light on how or whether it is possible to integrate some degrees of freedom in tensor models.

%

\section{GM tensorial interactions} \label{sec:TensorialInteractions}

\subsection{Tensors, unitary invariance and bubbles}

\paragraph*{Unitary invariance --} Let $V\simeq \mathbbm{C}^N$ and denote $V_1, \dotsc, V_d$, $d>2$ copies of $V$ where $V_c$ is said to be of color $c\in\{1, \dotsc, d\}$. If $C\subset \{1, \dotsc, d\}$ is a subset of colors, we denote
\begin{equation}
V_C = \bigotimes_{c\in C} V_c \qquad \text{and} \qquad
E_d = \bigotimes_{c=1}^d V_c.
\end{equation}
We consider a tensor $T$ as an element of $E_d$ and its complex conjugate $\bar{T}$ in the dual space $E_d^*$. The components are denoted $T_{a_1 \dotsb a_d}$ and $\bar{T}_{a_1 \dotsb a_d}$, for $d>2$ and $a_c =1, \dotsc, N$ for $c=1, \dotsc, d$. 

We focus on interactions which are invariant under the natural action of $U(V_1)\otimes \dotsb \otimes U(V_d) \simeq U(N)^d$, i.e.
\begin{equation}
T'_{a_1 \dotsb a_d} = \sum_{b_1, \dotsc, b_d} U^{(1)}_{a_1 b_1} \dotsm U^{(d)}_{a_d b_d}\ T_{b_1\dotsb b_d}.
\end{equation}
where the unitary matrices $U^{(c)}$ are independent.

\paragraph*{Bubbles --} The ring of invariant polynomials is generated by the bubble polynomials. A {\bf bubble} is a bipartite, connected graph: $i)$ whose vertices have degree $d$, $ii)$ each edge carries a color from the set $\{1, \dotsc, d\}$, $iii)$ each color is incident exactly once on each vertex. If $B$ is such a bubble, its associated polynomial $B(\{T_v\})$ is obtained by assigning a tensor $T_v$ to each white vertex $v$ of $B$, a $\bar{T}_v$ to each black vertex, and contracting their indices as follows. For an edge of color $c$ between two vertices $v_1$ and $v_2$, we identify and sum the indices of the corresponding $T_{v_1}$ and $\bar{T}_{v_2}$ which are in position $c$,
\begin{equation}
\sum_{a_c=1}^N (T_{v_1})_{\dotsb a_c \dotsb} (\bar{T}_{v_2})_{\dotsb a_c \dotsb} = \begin{array}{c} \includegraphics[scale=.4]{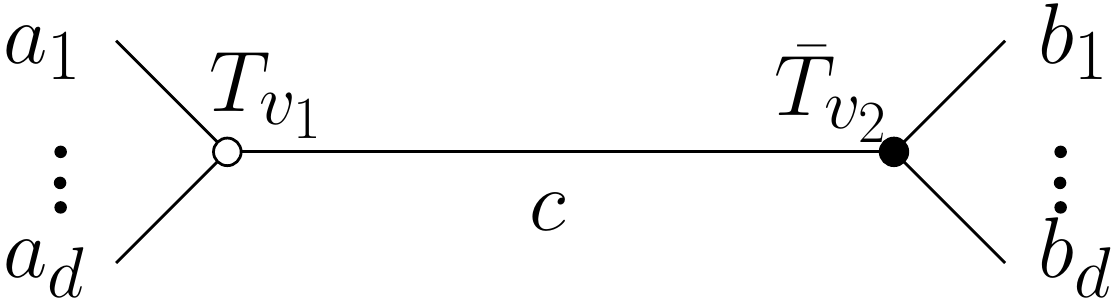}\end{array}.
\end{equation}
This way all indices are summed between $T$s and $\bar{T}$s, which ensures invariance. 

For two tensors $T_{v_1}, T_{v_2}$ and a color set $C\subset \{1, \dotsc, d\}$ we denote the set of complementary colors $\widehat{C} = \{1, \dotsc, d\}\setminus C$ and 
\begin{equation} \label{TwoTensorContraction}
(T_{v_2} | T_{v_1})_{\widehat{C}} = \begin{array}{c} \includegraphics[scale=.4]{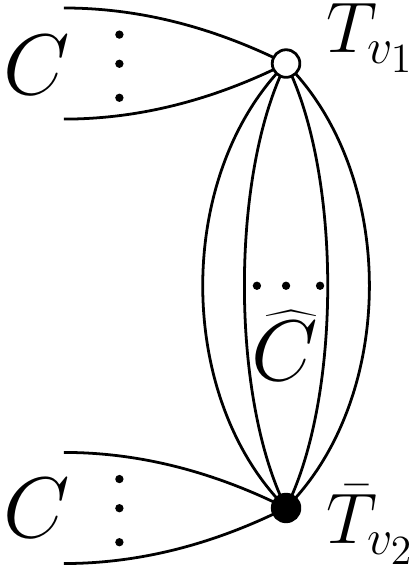} \end{array} \quad \in \quad \bigotimes_{c\in C} V^*_c\otimes V_c
\end{equation}
the matrix obtained by contracting the indices of $T_{v_1}$ and $\bar{T}_{v_2}$ whose colors are in $\widehat{C}$. It is a $N^{|C|}\times N^{|C|}$-square matrix. In case $C=\emptyset$, we simply omit $\widehat{C} = \{1, \dotsc, d\}$ in the notation, i.e. $(T_{v_2}|T_{v_1})$ which is a scalar.

\paragraph*{The 2-vertex bubble --} There is a single bubble with two vertices, corresponding to the {\bf only invariant of degree 2},
\begin{equation}
\begin{array}{c} \includegraphics[scale=.35]{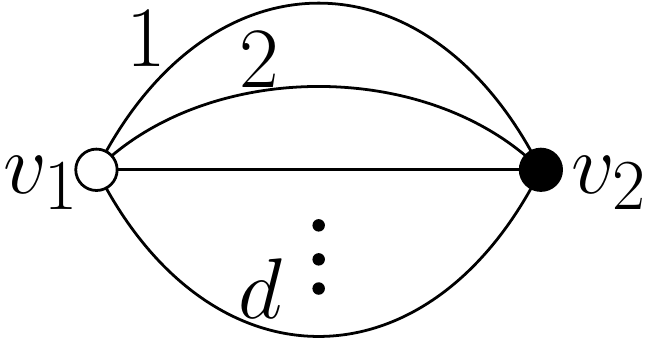}\end{array} = \sum_{a_1, \dotsc, a_d} (T_{v_1})_{a_1 \dotsb a_d} (\bar{T}_{v_2})_{a_1 \dotsb a_d} = (T_{v_2}|T_{v_1})
\end{equation}
and its complex conjugate (except when $T_{v_1} = T_{v_2}$ since $(T|T)$ is real).

\paragraph*{Quartic bubbles --} At degree 4, the bubbles are called quartic bubbles and can be obtained this way. Consider a white vertex: there are only two black vertices it can connect to. Therefore, it can have edges with colors in $C\subset \{1, \dotsc, d\}$ connected to one black vertex and $\widehat{C}=\{1, \dotsc, d\}\setminus C$ connected to the other. This determines the bubble $Q_C$
\begin{equation}
Q_{C} = \begin{array}{c} \includegraphics[scale=.35]{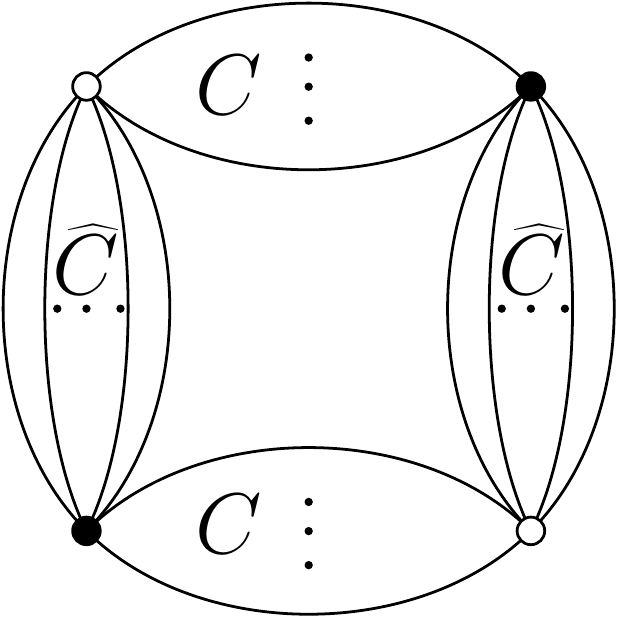}\end{array} = Q_{\widehat{C}}
\end{equation}
As for the associated polynomial, we write
\begin{equation} \label{CutQuartic}
\begin{array}{c} \includegraphics[scale=.35]{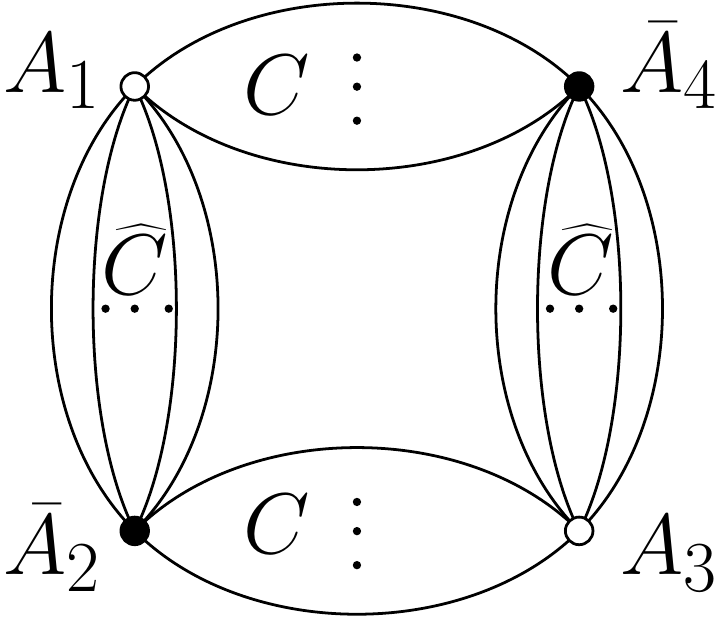} \end{array} = Q_{C}(A_1, \bar{A}_2; A_3, \bar{A}_4) =  Q_{\widehat{C}}(A_1, \bar{A}_4 ; A_3, \bar{A}_2)
\end{equation}
or equivalently
\begin{equation} \label{QuarticMatrix}
\begin{array}{c} \includegraphics[scale=.35]{4VertexBubblePolynomial.pdf} \end{array}
 = \tr_{V_C} \Bigl((A_2|A_1)_{\widehat{C}} (A_4|A_3)_{\widehat{C}}\Bigr) = \tr_{V_{\widehat{C}}} \Bigl((A_4|A_1)_{{C}} (A_2|A_3)_{{C}}\Bigr)
\end{equation}
where $\tr_V$ denotes the trace of an endomorphism of $V$. In the first equality, the trace is therefore on square matrices of size $N^{|C|}\times N^{|C|}$. In the second equality, the exchange symmetry $C\leftrightarrow \widehat{C}$ is used.

If $A_1=A_3=T$ and $\bar{A}_2 = \bar{A}_4 = \bar{T}$, we simply write $Q_C(T, \bar{T})$.


\paragraph*{Admissible color sets --} $C$ and $\widehat{C}$ play equivalent roles since $Q_C = Q_{\widehat{C}}$. However, we will see that in order to reduce the number of degrees of freedom of tensor models, it will be important to distinguish between the subset with more than $d/2$ colors and the one with less than $d/2$ colors. As a convention, we thus choose $C$ such that
\begin{itemize}
\item $|C|\leq d/2$,
\item if $|C|=d/2$ (then $|\widehat{C}|=d/2$ too), then we choose $1\in C$ as a convention.
\end{itemize}
We denote $\mathcal{C}$ the set of such admissible color sets.

\subsection{Generalized melonic bubbles} \label{sec:GMB}

\paragraph*{Melonic bubbles --} A {\bf melonic dipole} is a pair of vertices connected by exactly $d-1$ edges. A melonic dipole insertion of color $c$ is the move inserting a melonic dipole on an edge of color $c$,
\begin{equation} \label{DipoleInsertion}
\begin{array}{c} \includegraphics[scale=.5]{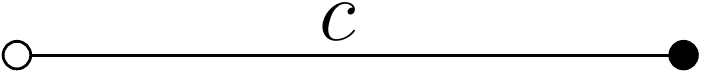} \end{array} \quad \to \quad \begin{array}{c} \includegraphics[scale=.5]{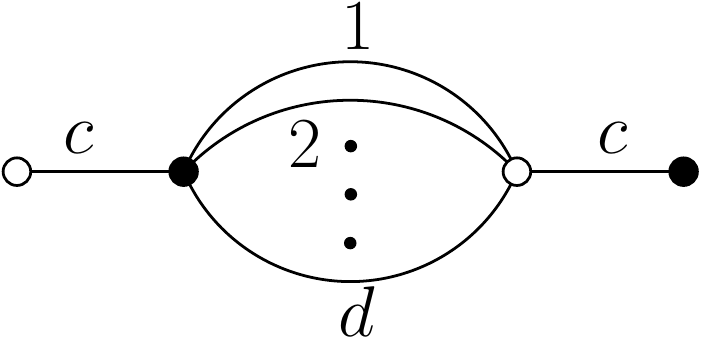} \end{array}
\end{equation}
It is a sequence of two series and $(d-1)$ parallel extensions.

A melonic bubble is any bubble obtained by repeatedly inserting melonic dipoles on arbitrary edges, starting from the bubble with two vertices. The first step starting from the bubble with two vertices produces the quartic bubble $Q_{\{c\}}$ introduced above. Then another edge is selected, in $Q_{\{c\}}$, to perform a new melonic dipole insertion, creating a bubble of degree 6, and so on. There are obviously several sequences of melonic moves to construct a given melonic bubble.

\paragraph*{Bidipoles --} If $C\subset \{1, \dotsc, d\}$, a $C$-bidipole $\mathcal{V}=\{v, \bar{v},w\}$ is a set of three vertices $v$, $\bar{v}$, $w$, such that $\bar{v}$ is connected to $w$ by the edges with colors in $C$, and to $v$ by the edges with colors in $\widehat{C}= \{1, \dotsc, d\}\setminus C$. A $C$-bidipole {\bf insertion} is the removal of a vertex and in its place the insertion of a $C$-bidipole,
\begin{equation} \label{qMove}
\begin{array}{c} \includegraphics[scale=.5]{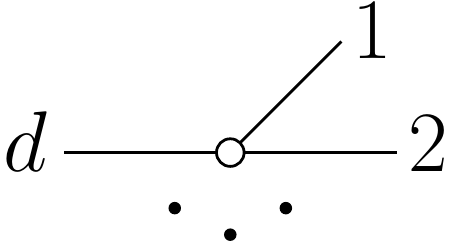} \end{array} \quad \to \quad \begin{array}{c} \includegraphics[scale=.5]{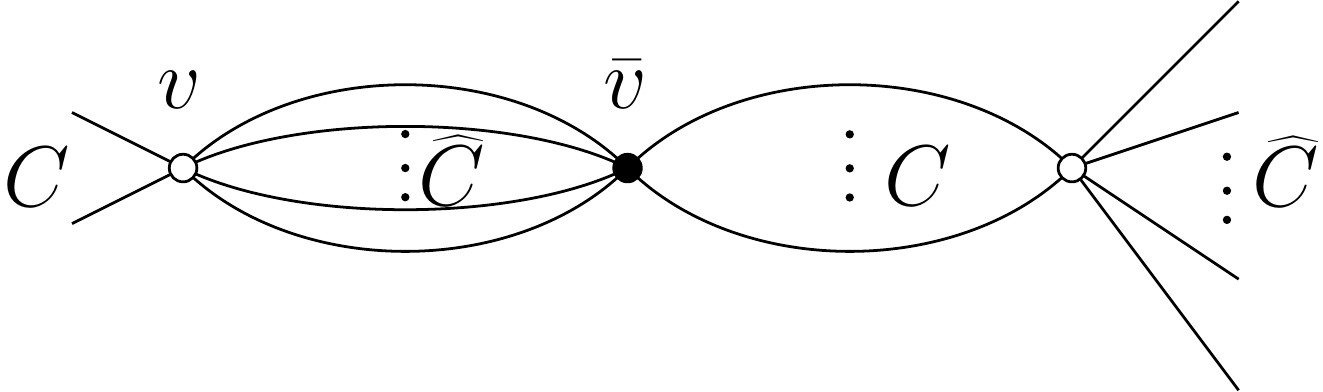} \end{array}
\end{equation}
and we call bidipole {\bf removal} the inverse operation.

There is no difference between a $C$-bidipole and a $\widehat{C}$-bidipole. The move also exists with white and black vertices exchanged. The melonic move corresponds to $C$ (or $\widehat{C}$) of cardinality 1.

\paragraph*{Generalized melonic bubbles --} We call bubbles built from the two-vertex bubble by arbitrary sequences of bidipole insertions {\it generalized melonic} (GM) bubbles.

\paragraph*{Uniqueness of the insertion set --} By definition, for every GM bubble $B$ there exists a (non-unique) sequence of GM bubbles 
\begin{equation} \label{BubbleSequence}
B = B^{(0)}\ \underset{C_1, v_1}{\leftarrow}\ \ B^{(1)}\ \leftarrow \dotsb \underset{C_{V/2-2}, v_{V/2-2}}{\leftarrow} \ B^{(V/2-2)} = Q_{C_{V/2-2}}
\end{equation}
where $(B^{(i)})$ has $V-2i$ vertices, and a sequence of sets $(C_i\subset \{1, \dotsc, d\})$ and a sequence of vertices $(v_i \in B^{(i)})$, such that the bubble $B^{(i)}$ is obtained from the bubble $B^{(i+1)}$ by a $C_i$-bidipole insertion at the vertex $v_i$.

Here is an example of a construction of a GM bubble using bidipole moves. At each step we have circled the vertex $v_i$ on which the next bidipole insertion is performed and indicated the corresponding set $C_i$.
\begin{multline} \label{ExampleSPBubble}
\begin{array}{c} \includegraphics[scale=.4]{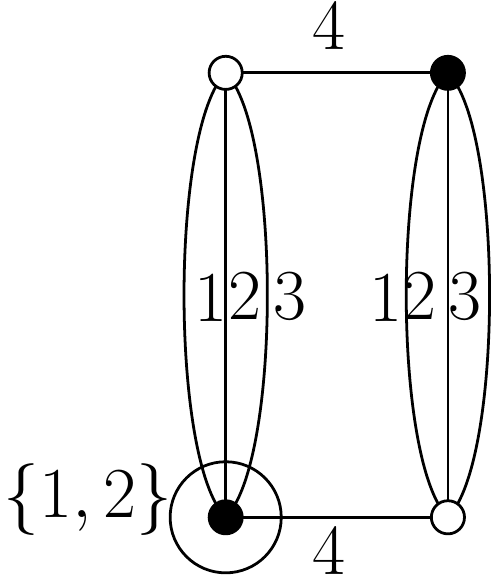} \end{array} \to \begin{array}{c} \includegraphics[scale=.4]{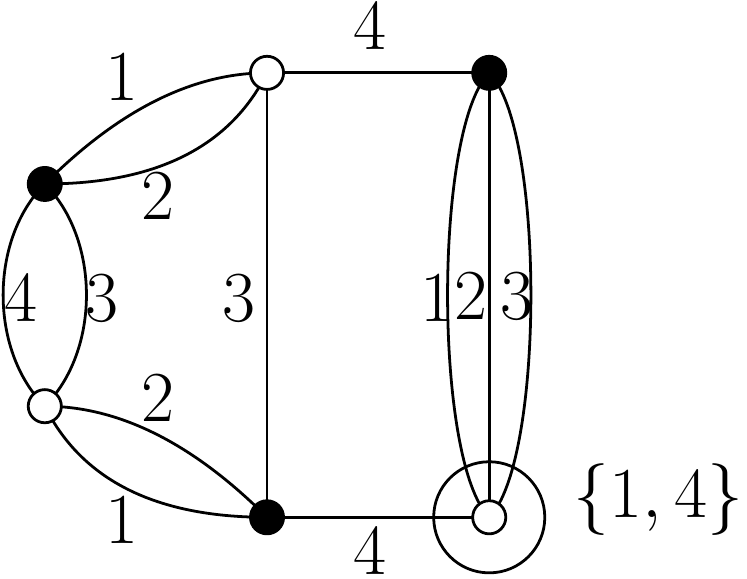} \end{array} \to \begin{array}{c} \includegraphics[scale=.4]{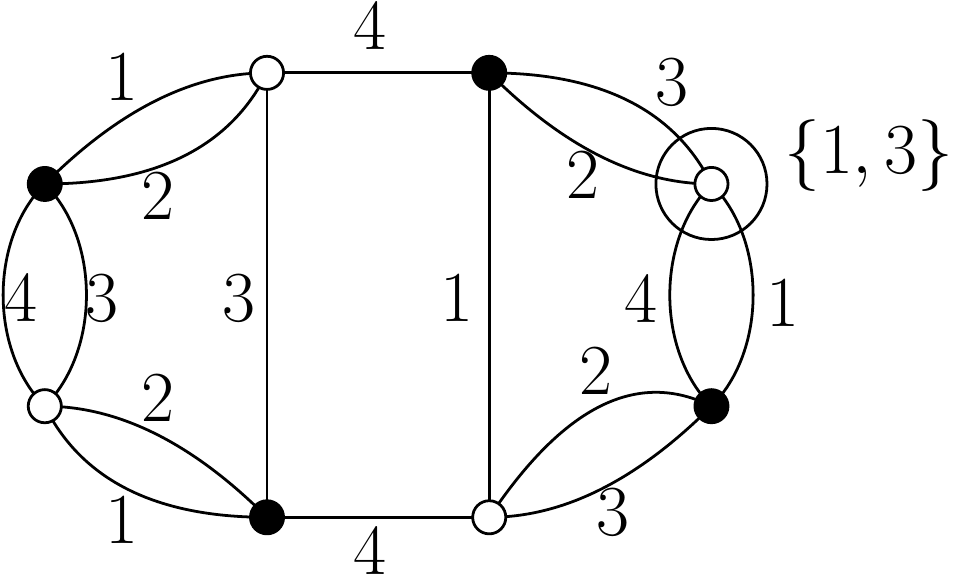} \end{array} \to \begin{array}{c} \includegraphics[scale=.4]{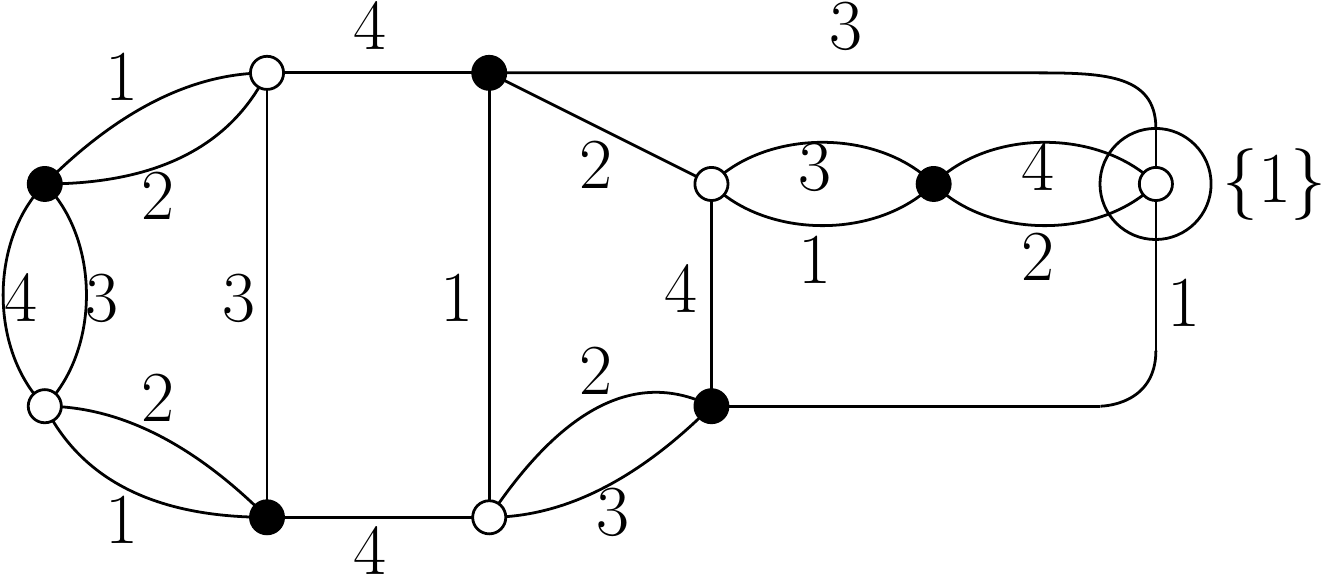} \end{array} \\
\to \begin{array}{c} \includegraphics[scale=.4]{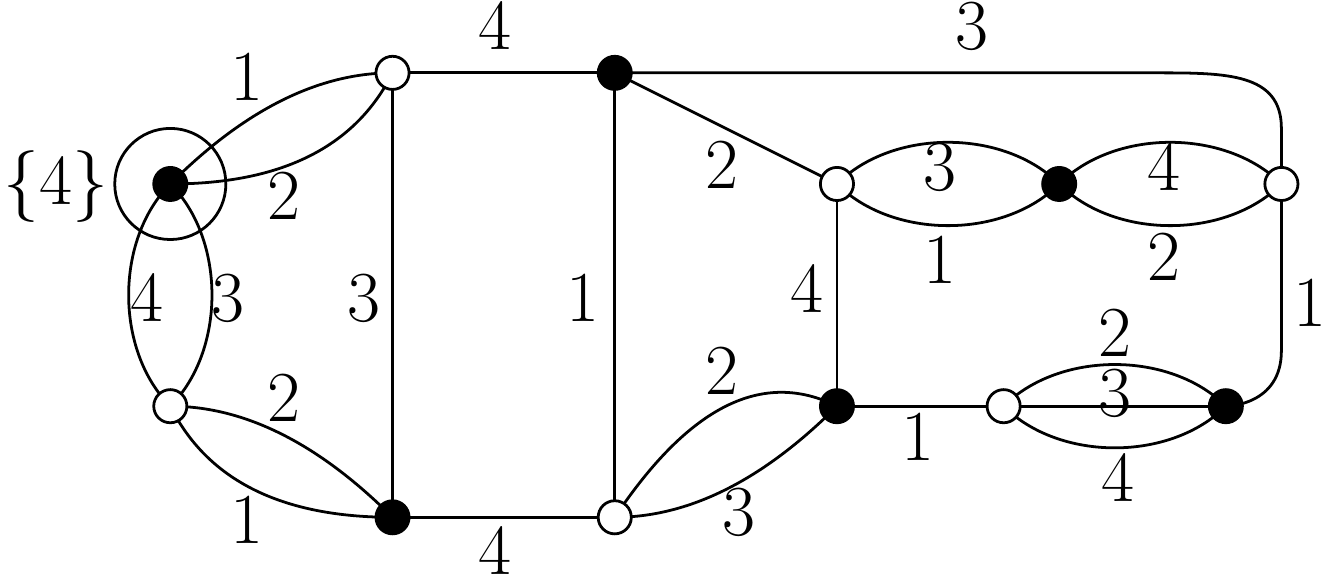} \end{array} \to \begin{array}{c} \includegraphics[scale=.4]{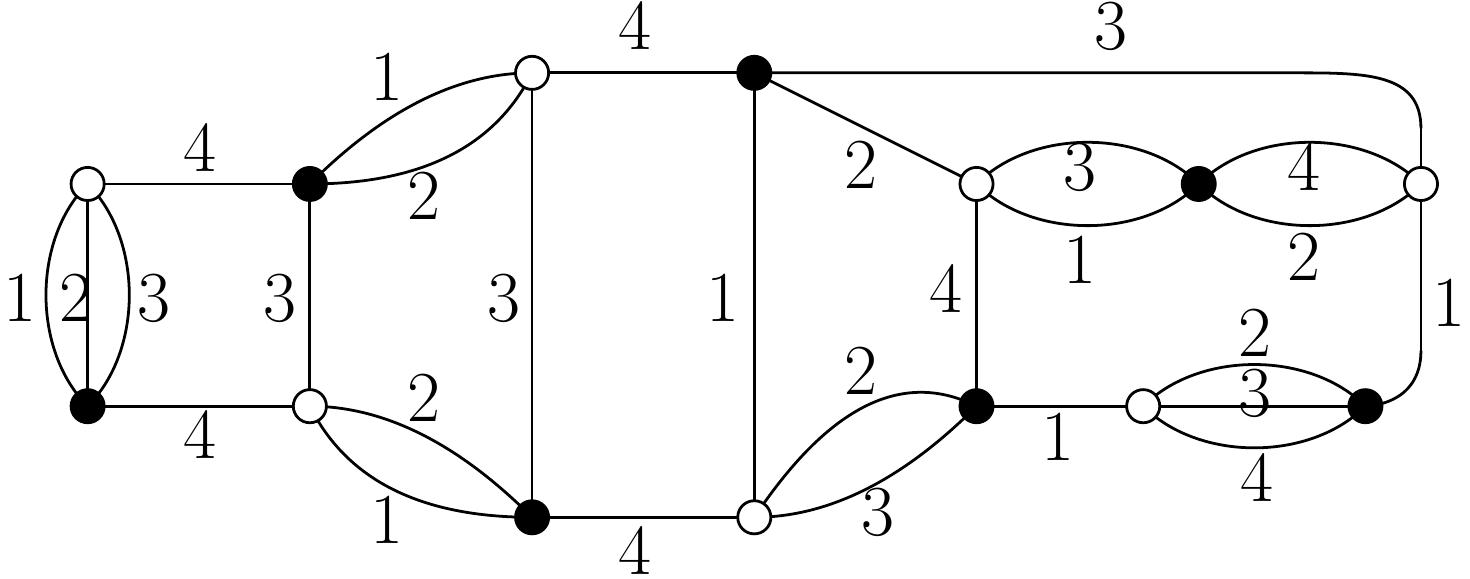} \end{array}
\end{multline}

\begin{proposition} \label{thm:SetsUniqueness}
The multiset of bidipole insertion color sets, $\mathcal{C}_B = \{C_1, \dotsc, C_{V/2-2}\}_{i=1,\dotsc,V/2-2}$, is uniquely determined by $B$.
\end{proposition}

Equivalently $\mathcal{C}_B$ is independent of the choice of sequence \eqref{BubbleSequence}. The bubbles $(B^{(i)})$ and the vertices $(v_i)$ are not canonical but the color sets $C_i$ are.

\begin{proof}
We proceed by induction on the number of vertices of $B$. If it is quartic, i.e. $B=Q_C$ for some $C$, there is obviously a single set $C$. Similarly, it can be checked that if $B$ has six vertices, it is fully characterized (up to exchanging black with white vertices) by two sets $C_1, C_2$. Indeed, one can perform a $C_2$-bidipole insertion on a vertex of $Q_{C_1}$ or a $C_1$-bidipole insertion on a vertex of $Q_{C_2}$ with the same result.

Assume that the proposition holds for bubbles with $V\geq 6$ vertices and consider $B$ GM with $V+2$ vertices. We find a $C_1$-bidipole for some $C_1$, $\mathcal{V}_1=\{u_1,v_1,w_1\}$. Removing the bidipole we get another GM bubble $B^{(1)}$ with $V$ vertices. From the induction hypothesis, it has a unique set $\mathcal{C}_{B^{(1)}}$ and
\begin{equation}
\mathcal{C}_B=\{C_1\} \cup \mathcal{C}_{B^{(1)}}.
\end{equation}

Since $\mathcal{C}_{B^{(1)}}$ is unique, the only way to possibly get a different $\mathcal{C}_B$ is to start by removing from $B$ another bidipole, say a $C_1'$-bidipole $\mathcal{V}_1'=\{u_1', v_1', w_1'\}\neq \mathcal{V}_1$. If $\mathcal{V}_1\cap\mathcal{V}_1' = \emptyset$, then both bidipoles are non-overlapping and we can either remove $\mathcal{V}_1$ then $\mathcal{V}_1'$ or the other way around and get the same GM bubble $B^{(2)}$. It comes
\begin{equation}
\mathcal{C}_B=\{C_1,C_1'\} \cup \mathcal{C}_{B^{(2)}}.
\end{equation}
where $\mathcal{C}_{B^{(2)}}$ is unique.

If $|\mathcal{V}_1\cap \mathcal{V}_1'|=2$, it implies $C_1'=C_1$ and removing $\mathcal{V}_1$ or $\mathcal{V}_1'$ gives the same bubble $B^{(1)}$
\begin{equation}
\begin{array}{c} \includegraphics[scale=.4]{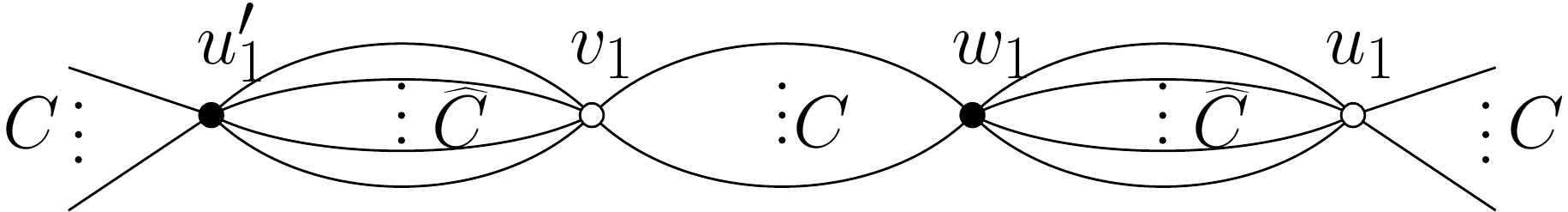} \end{array} \qquad \to \qquad \begin{array}{c} \includegraphics[scale=.4]{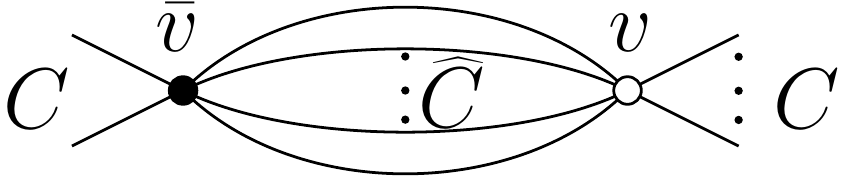} \end{array}
\end{equation}
Again, $\mathcal{C}_B = \{C_1\}\cup \mathcal{C}_{B^{(1)}}$.

If $|\mathcal{V}_1\cap \mathcal{V}_1'|=1$ then we have a situation as follows, where we can remove both $\mathcal{V}_1$ and $\mathcal{V}'_1$ in any order and get a GM bubble $B^{(2)}$,
\begin{multline}
\begin{array}{c} \includegraphics[scale=.4]{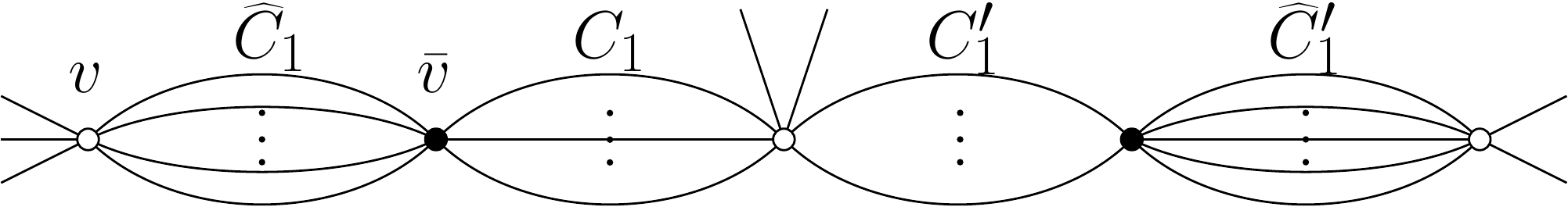} \end{array} \subset B\\
\to \qquad \begin{array}{c} \includegraphics[scale=.4]{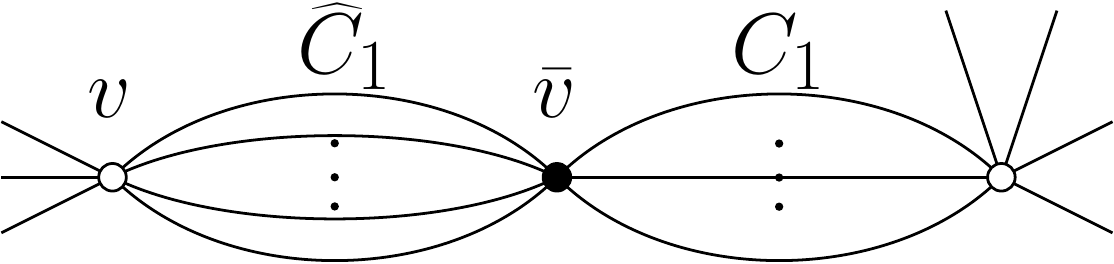} \end{array} \subset B^{(1)'} \qquad \to \qquad \begin{array}{c} \includegraphics[scale=.4]{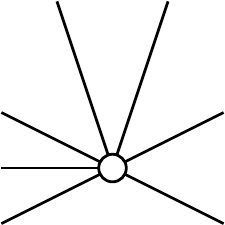} \end{array} \subset B^{(2)}
\end{multline}
Therefore $\mathcal{C}_B=\{C_1,C_1'\} \cup \mathcal{C}_{B^{(2)}}$.
\end{proof}

\paragraph*{Topology --} For a bubble $B$, the $C$-bubbles are the connected subgraphs formed by the edges with colors in $C$. A topological $C$-dipole is a subgraph made of two vertices connected by all the colors in $\widehat{C}$ such that the $C$-bubbles incident to the two vertices are distinct.

Topological $C$-dipoles can be created or removed without changing the topology \cite{DipoleMoves},
\begin{equation}
\begin{array}{c} \includegraphics[scale=.5]{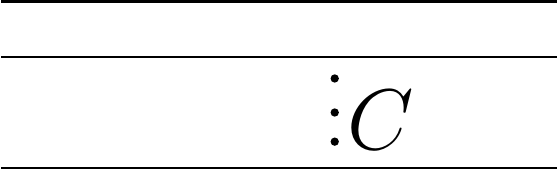} \end{array} \quad \leftrightarrow \quad \begin{array}{c} \includegraphics[scale=.5]{CDipoleInsertion.pdf} \end{array}
\end{equation}
As is clear from \eqref{qMove}, bidipole insertions are topological. Since the 2-vertex bubble is a $d$-ball, we conclude that GM bubbles are $d$-balls.

\subsection{GM bubbles as trees of quartic bubbles}

In this section we prove Theorem \ref{thm:QuarticTree}.

\paragraph*{The set $\mathbbm{G}_V$ --} Consider graphs made of quartic bubbles which are connected along some additional edges which we call {\bf dashed lines}. A vertex can be incident to one or no dashed line. We denote $\mathbbm{G}_V$ the set of connected graphs made of quartic bubbles and dashed lines, with $V$ vertices having no incident dashed lines, and such that all dashed lines are edge-cuts (meaning that cutting any dashed line disconnects the graph).

\begin{proposition} \label{thm:G_V}
The following statements are equivalent.
\begin{itemize}
\item[a)] $G\in \mathbbm{G}_V$
\item[b)] For all dashed lines in $G$, there exist $G_1 \in \mathbbm{G}_{V_1}$, $G_2\in\mathbbm{G}_{V_2}$ with $V_1+V_2-2=V$ such that
\begin{equation}
G = \begin{array}{c} \includegraphics[scale=.4]{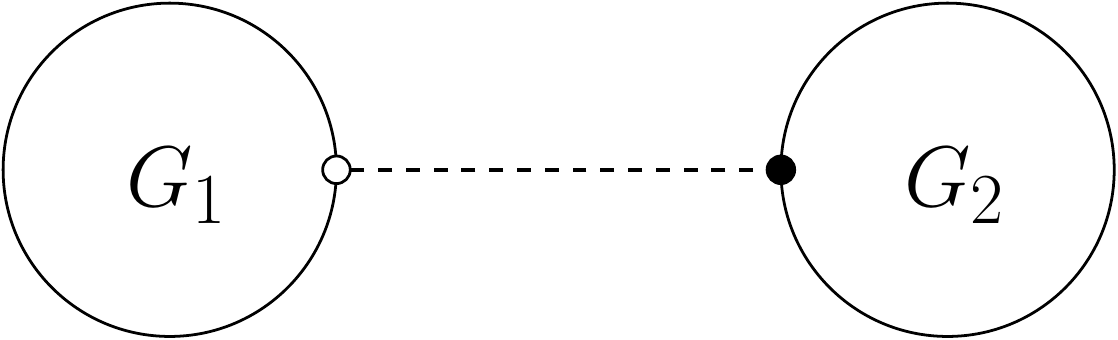} \end{array}
\end{equation}
\item[c)] There exists $G'\in\mathbbm{G}_{V-2}$ and a color set $C$ such that
\begin{equation} \label{BubbleInduction}
G = \begin{array}{c} \includegraphics[scale=.4]{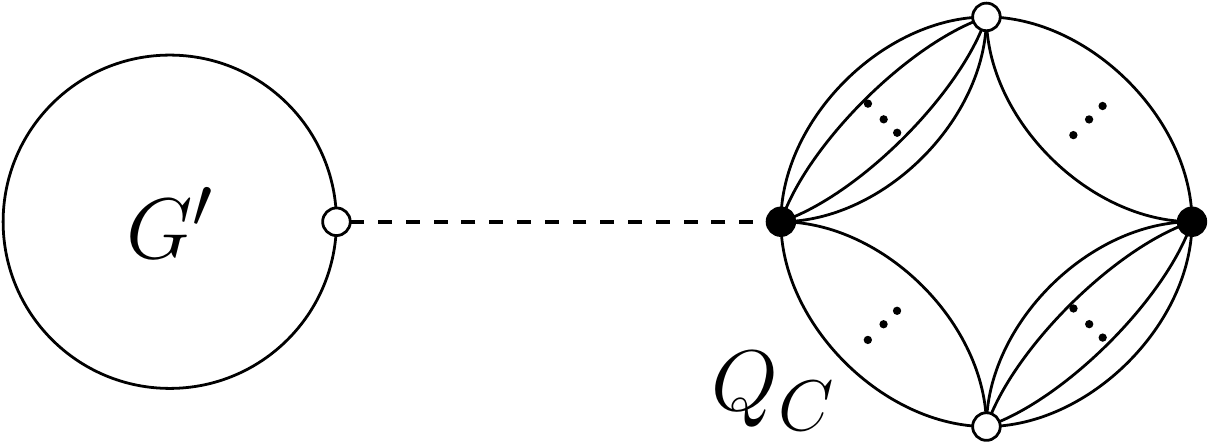} \end{array}
\end{equation}
\end{itemize}
\end{proposition}

\begin{proof}
This works exactly like for trees, by mapping quartic bubbles to vertices of degree at most four and dashed lines to edges.
\end{proof}

It is easily found that $G\in\mathbbm{G}_V$ has $V/2-1$ quartic bubbles and $V/2-2$ dashed lines.

\paragraph*{Dashed line contraction --} The contraction of a dashed line removes the dashed line and the two vertices it is incident to, and reconnects the edges of color $\{1, \dotsc, d\}$ respecting their colors,
\begin{equation}
\begin{array}{c} \includegraphics[scale=.4]{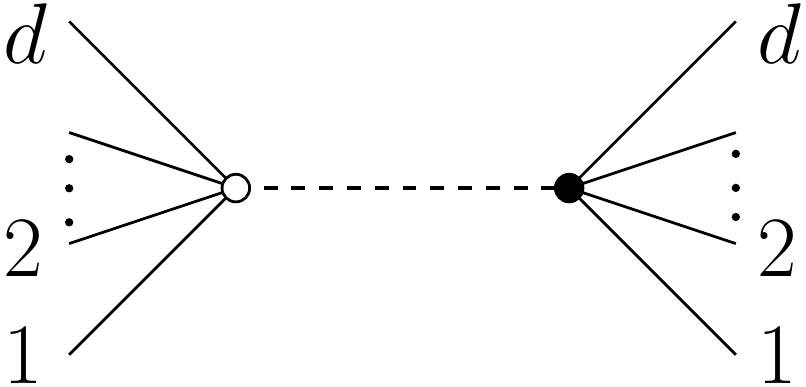} \end{array} \quad \to \quad \begin{array}{c} \includegraphics[scale=.4]{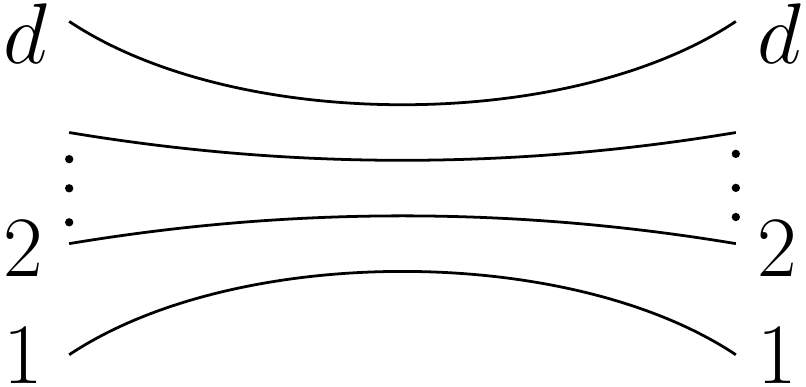} \end{array}
\end{equation}
The {\bf boundary operator} $\partial$ is a defined as the contraction of all dashed lines (it is independent of the order of the contractions). Note that it is defined not only on $\mathbbm{G}_V$ but for any collection of bubbles connected by dashed lines.

\begin{proposition} \label{thm:BoundaryBubbles}
For all $G\in \mathbbm{G}_V$, the boundary bubble $\partial G$ is a GM bubble with $V$ vertices.
\end{proposition}

\begin{proof}
One can use the characterization of $G\in \mathbbm{G}_V$ in \eqref{BubbleInduction} to perform an induction. Graphs in $\mathbbm{G}_4$ have no dashed lines so the action of $\partial$ is trivial and they are directly GM bubbles. Now we use \eqref{BubbleInduction} and from the induction hypothesis we get that $\partial G'$ is a GM bubble on $V-2$ vertices.

Then it is enough to check that the contraction of the dashed line in \eqref{BubbleInduction} is equivalent to the insertion of a $C$-bidipole,
\begin{equation} \label{BoundaryBidipole}
\partial \begin{array}{c} \includegraphics[scale=.4]{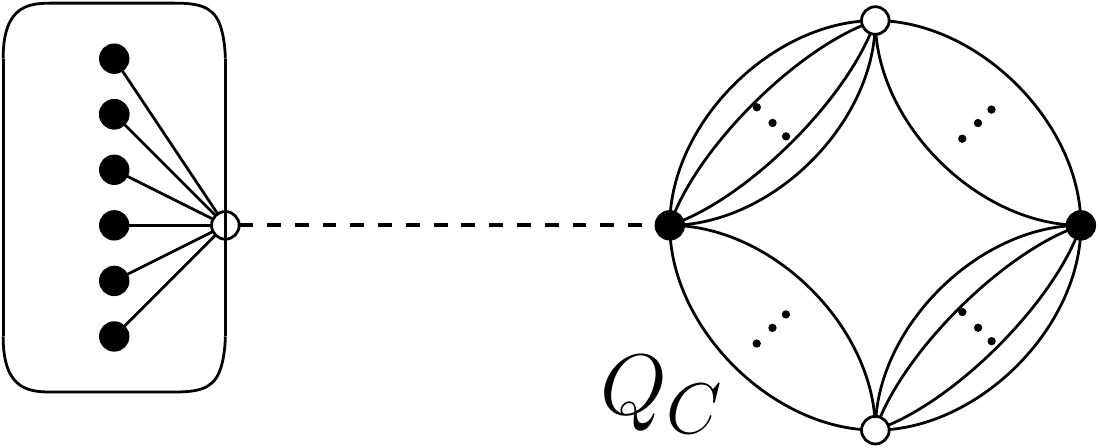} \end{array} \quad = \quad \begin{array}{c} \includegraphics[scale=.4]{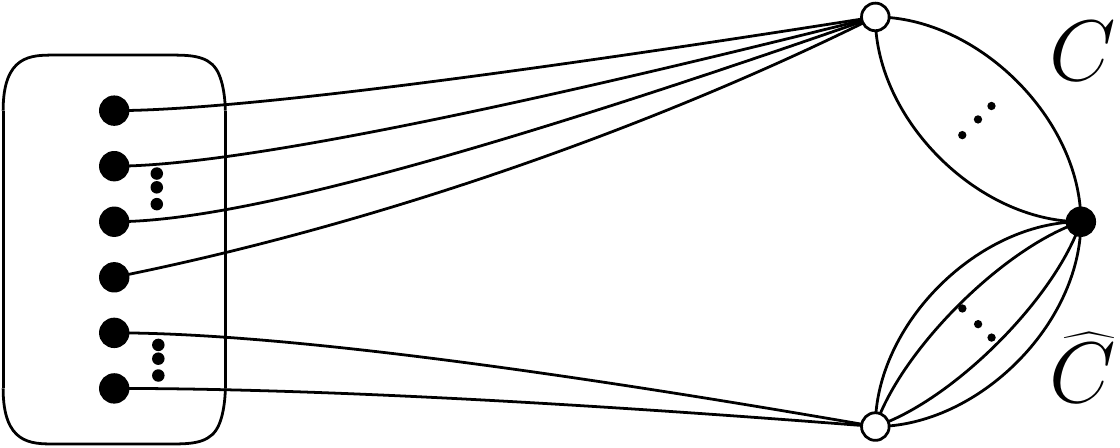} \end{array}
\end{equation}
\end{proof}

Theorem \ref{thm:QuarticTree} follows as a corollary of Propositions \ref{thm:G_V} and \ref{thm:BoundaryBubbles}.

\paragraph*{Construction --} Any sequence \eqref{BubbleSequence} gives rise to a graph $G\in\mathbbm{G}_V$ such that $\partial G=B$ which is built as follows. Consider $v_1$ in $B^{(1)}$ and connect it with a dashed line to the quartic bubble $Q_{C_1}$. We are assured by \eqref{BoundaryBidipole} that the contraction of that dashed line produces $B$. Then proceed iteratively with $B^{(i-1)}$ until the dashed line which is attached to $v_{V/2-2}$ on the quartic bubble $Q_{C_{V/2-2}}$. Due to Proposition \ref{thm:SetsUniqueness}, the list of quartic bubbles required in $G$ is unique, although $G$ itself is not. As an exercize, it can be checked that
\begin{equation}
\partial \left(\begin{array}{c} \includegraphics[scale=.3]{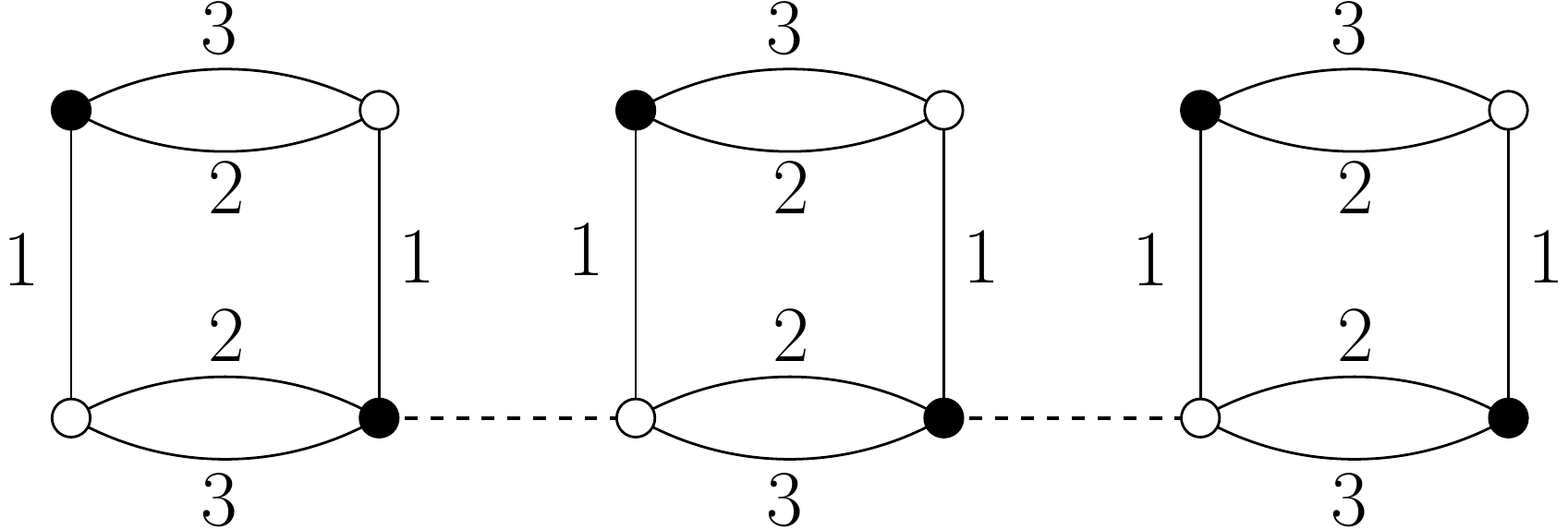} \end{array}\right) = \partial \left(\begin{array}{c} \includegraphics[scale=.3]{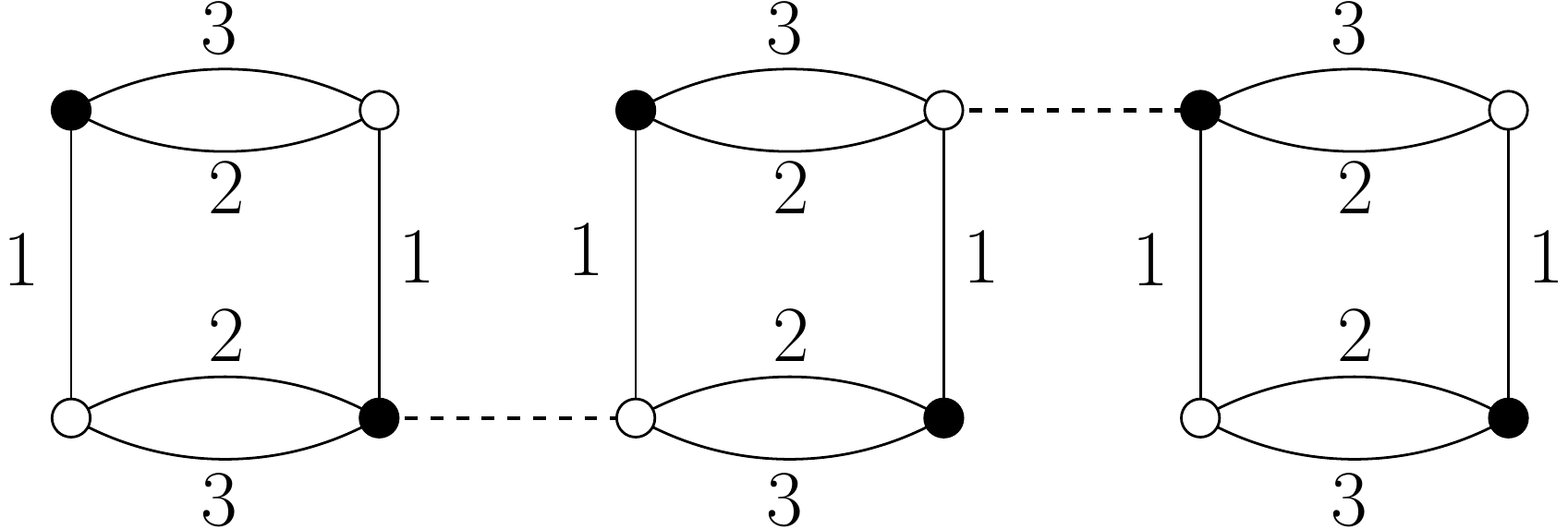} \end{array}\right) =  \begin{array}{c} \includegraphics[scale=.3]{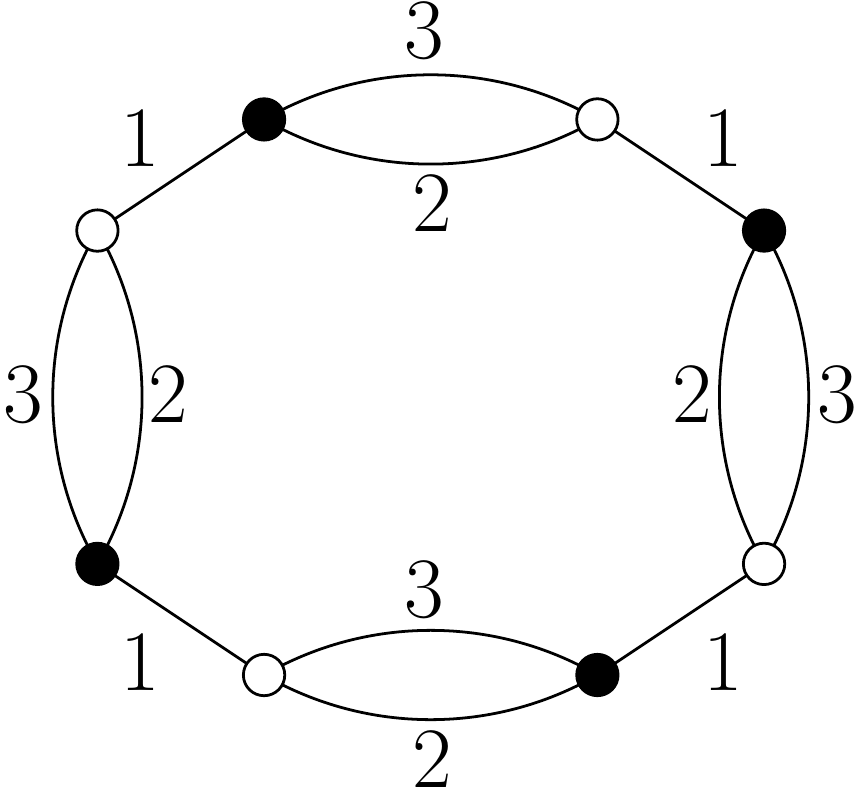} \end{array}
\end{equation}

\paragraph*{Example --} Consider the sequence of GM bubbles in \eqref{ExampleSPBubble}. We can read $B$ and $B^{(1)}$ from the last step,
\begin{equation}
B = \begin{array}{c} \includegraphics[scale=.4]{SPGraph.pdf} \end{array} = \partial \left(\begin{array}{c} \includegraphics[scale=.4]{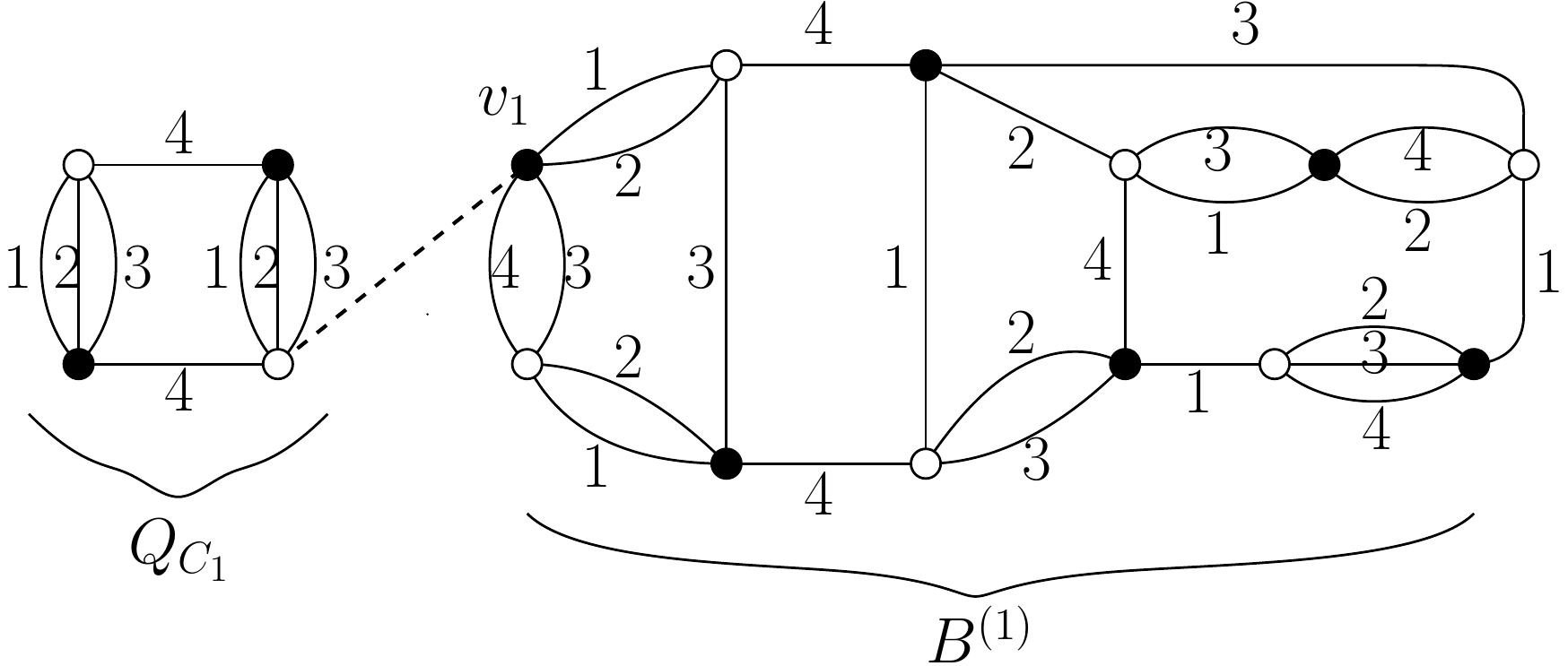} \end{array}\right)
\end{equation}
Continuing the process following the sequence of \eqref{ExampleSPBubble}, we get
\begin{multline}
B = \partial \begin{array}{c} \includegraphics[scale=.38]{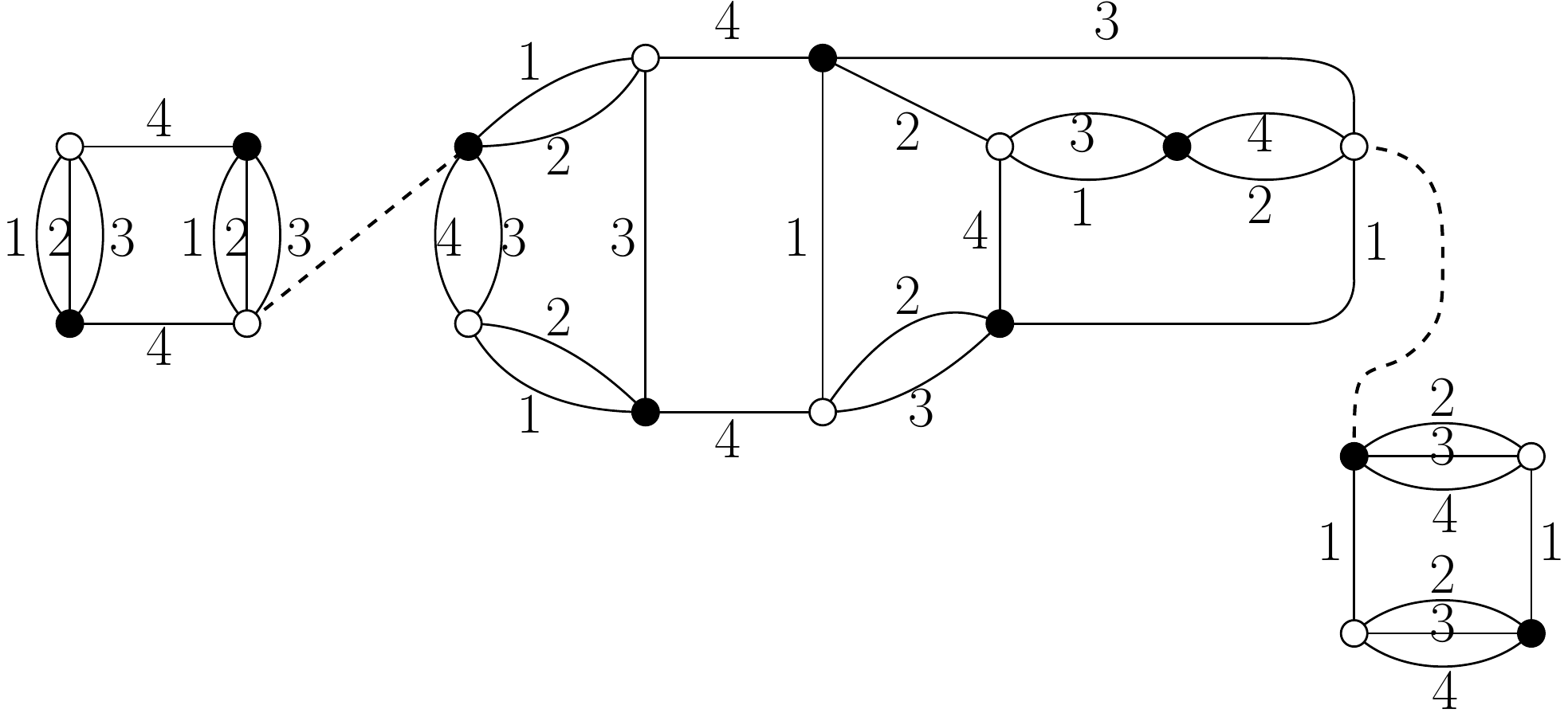} \end{array} 
= \partial \begin{array}{c} \includegraphics[scale=.38]{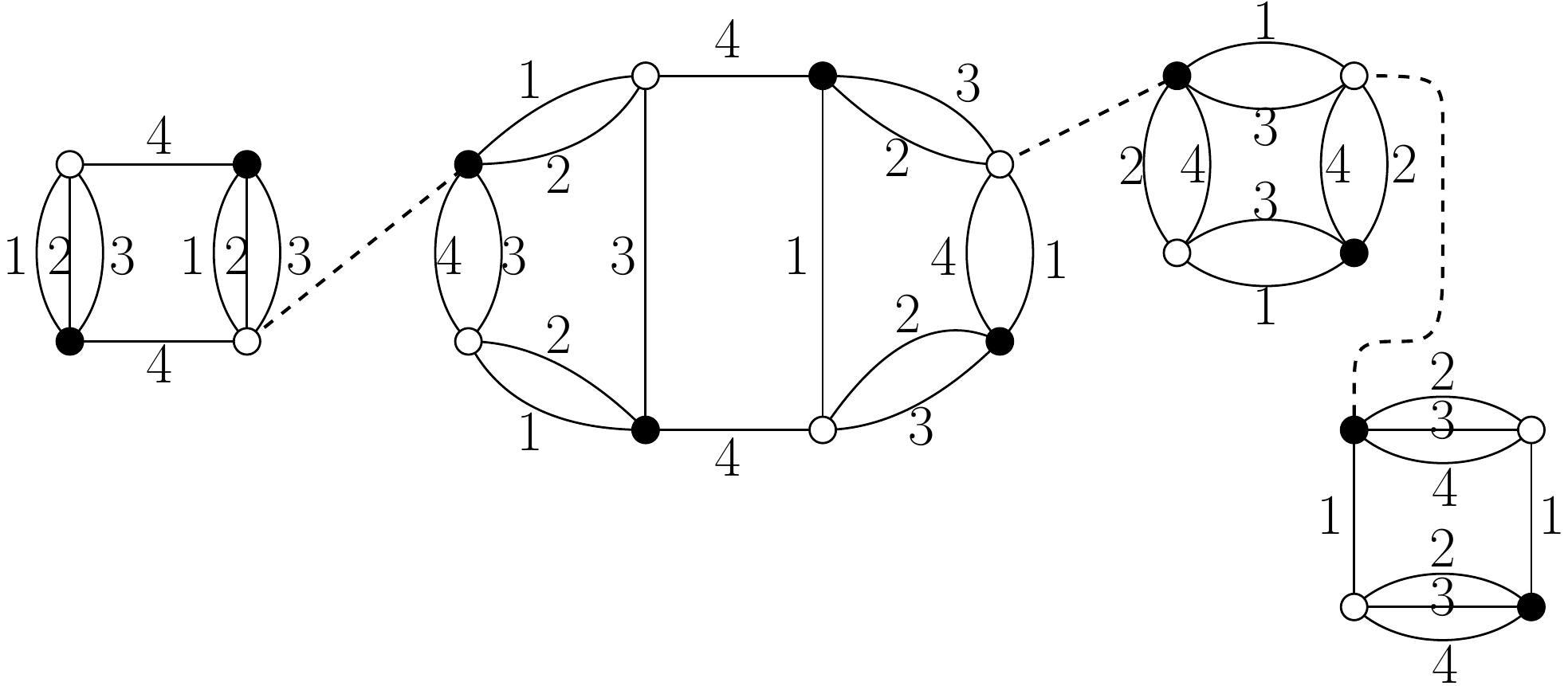} \end{array} \\
= \partial \begin{array}{c} \includegraphics[scale=.38]{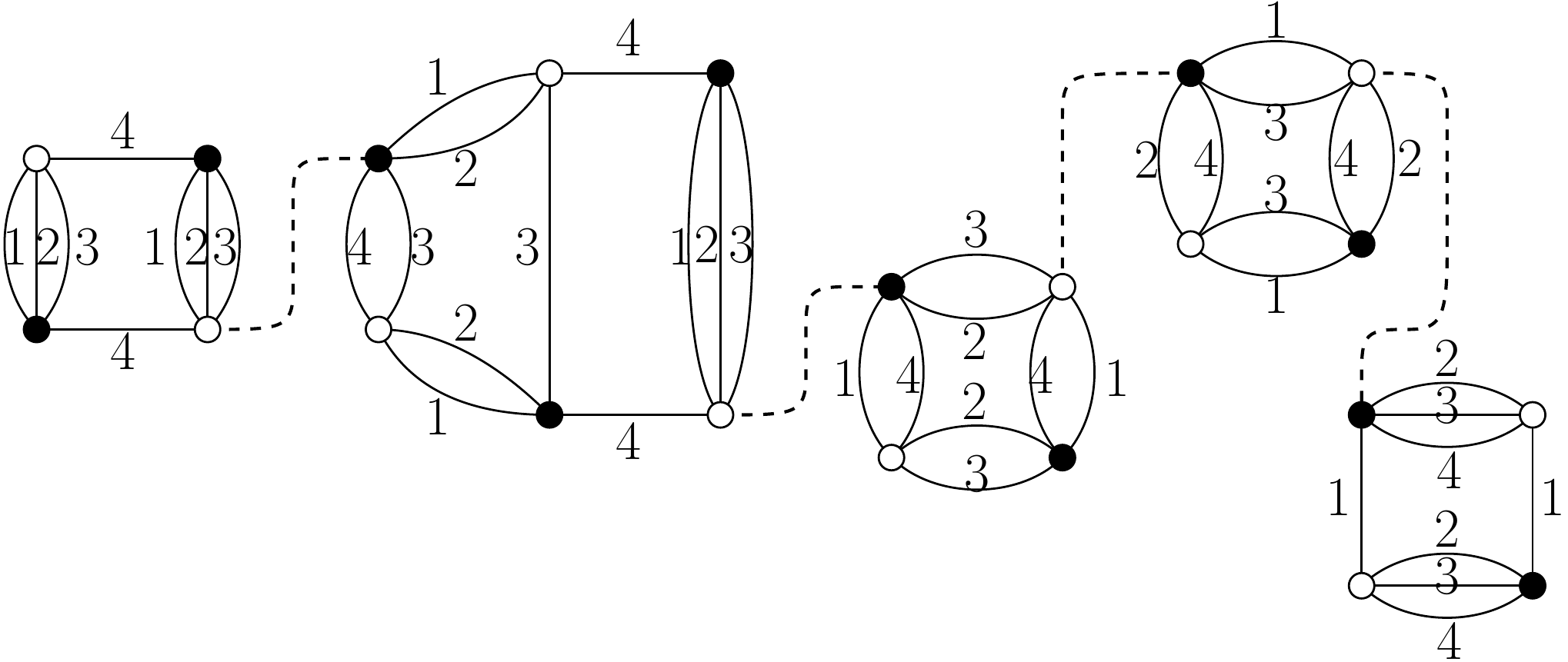} \end{array} 
= \partial \begin{array}{c} \includegraphics[scale=.38]{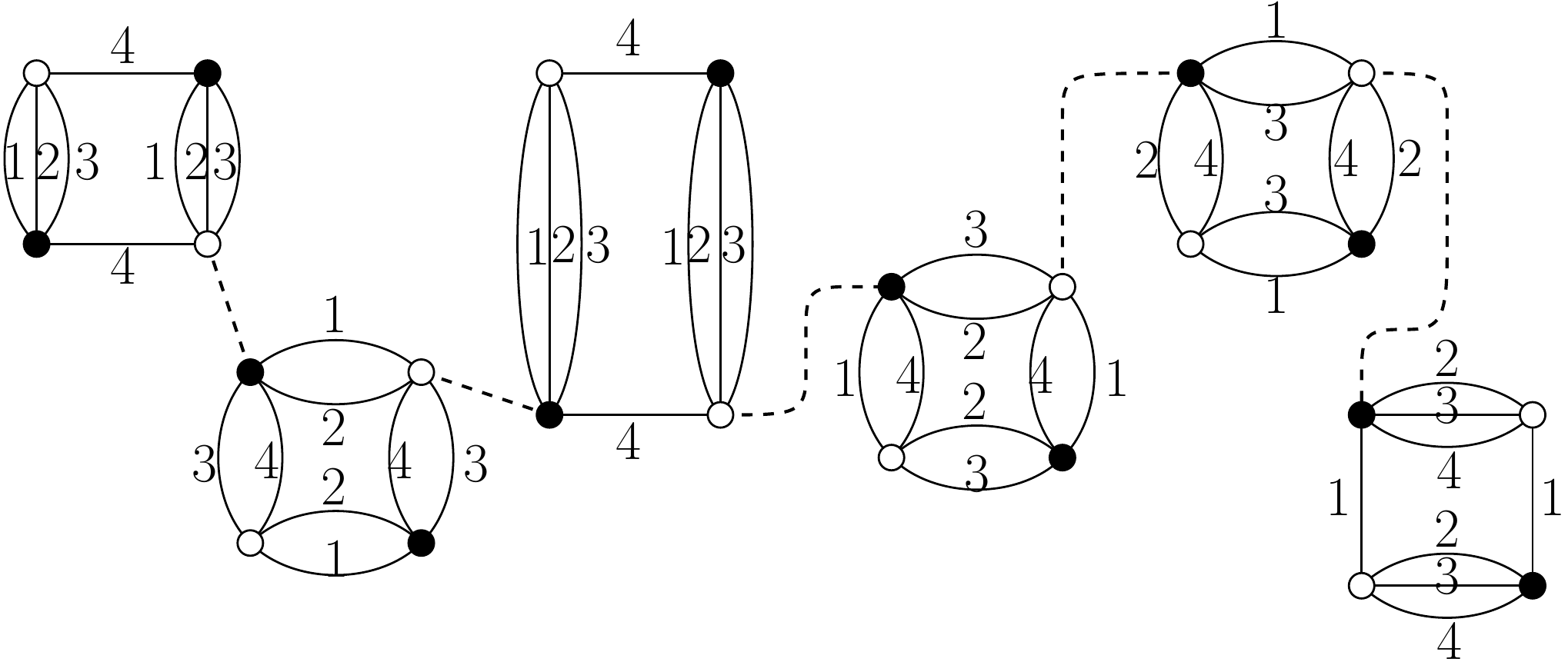} \end{array}
\end{multline}
where we can identify $\mathcal{T}$ such that $\partial \mathcal{T}=B$ in the last step. In particular, we can read the list of color sets $\mathcal{C}_B = \{\{4\}, \{1, 2\}, \{4\}, \{1, 4\}, \{1, 3\}, \{1\}\}$.

\paragraph*{The pairing $\pi_B$ --} Each bidipole move $B^{(i+1)} \to B^{(i)}$ introduces two new vertices and there is an almost canonical way to use this to define a unique pairing $\pi_B$. We recall that a pairing is a partition of the vertices into pairs of black and white vertices. Looking at \eqref{qMove}, there seems to be an ambiguity to decide which one is the new white vertex. Since we have defined $|C_i|\leq d/2$ however, there is a canonical difference between $C_i$ and $\widehat{C}_i$ when $|C_i|<d/2$. Therefore, we say that $\{v, \bar{v}\}$ forms a canonical pair if they are connected by more than $d/2$ edges, i.e. $v$ and $\bar{v}$ are connected by the colors of $\widehat{C}_i$.

In the case $|C_i|=|\widehat{C}_i|=d/2$ however, there really is an intrinsic ambiguity. We can nevertheless use a convention: a pair $\{v, \bar{v}\}$ is formed if $v$ and $\bar{v}$ are not connected by the color 1. In all cases we have therefore $v$ and $\bar{v}$ connected by the colors of $\widehat{C}_i$.

Then we set $\pi_B(v) = \bar{v}$ and $\pi_B(\bar{v}) = v$. We can proceed this way at each step $B^{(i+1)} \to B^{(i)}$ and this way find a natural pairing $\pi_B$ of the vertices $B$. It can also be seen as an involution without fixed points between black and white vertices.

\begin{proposition}
$\pi_B$ is unique.
\end{proposition}

\begin{proof}
We give two sketches of proof. A constructive one relies on the fact that there is no overlap between any pairs of vertices connected by more than $d/2$ colors or exactly $d/2$ colors not including 1. No overlap means that the edges incident to $v$ and $\bar{v}$ do not connect the vertices $v'$ and $\bar{v}'$ of another pair. Therefore, it is well-defined to consider all pairs $\{v_i, \bar{v}_i\}$ in $B$, set $\pi_B(v_i) =\bar{v}_i$ and then remove the $C_i$-bidipoles. Because there is no overlap, the order of the removals does not matter and we get a bubble $B'$, to which the same reasoning can be applied.

There is a simpler proof by induction on the number of vertices of $B$. The pairing is obviously unique on a quartic bubble. From $B^{(1)}$ to $B$, the pair $\{v, \bar{v}\}$ is defined as explained above, thus $\pi_B(v)=\bar{v}$. Then we define the restriction of $\pi_B$ to the other vertices by $\pi_{B^{(1)}}$ which is unique.
\end{proof}

\section{Tensor models with GM interactions} \label{sec:SPTensors}

\paragraph*{The partition function --} The tensor models under consideration are those which use GM bubbles as interactions. Let $R$ be a finite set and $\{B_r\}_{r\in R}$ a finite set of GM bubbles. The partition function is
\begin{equation} \label{PartitionFunction1}
Z_N(\{t_r\}) = \int dT d\bar{T}\ \exp -\Bigl((T|T) + \sum_{r\in R} N^{s_r}\, t_r\, B_r(T, \bar{T}) \Bigl)
\end{equation}
The quantities $t_r$ are known as coupling constants (or counting parameters in combinatorics) and $s_r$ are the scaling coefficients. They are {\it a priori} unknown at this stage and will be determined later in Theorem \ref{thm:Scaling}.

We have chosen the normalization so that the bare propagator scales like $N^0$, in contrast with another, standard choice in the literature where it scales with $N^{-(d-1)}$. The relation between both choices is a rescaling of $T$ and $\bar{T}$ by $N^{(d-1)/2}$.

In matrix models the scaling coefficients are $s_V = 1-V/2$ for the interactions $\tr (MM^\dagger)^{V/2}$. {\it A priori}, $s_r$ may depend on the whole set of bubbles $\{B_r\}_{r\in R}$, and should be denoted $s_r(R)$. However, as we will see in Theorem \ref{thm:Scaling}, there is a unique choice which leads to a non-trivial large $N$ limit of the model and it is such that each coefficient $s_r$ only depends on the bubble $B_r$, as in matrix models.

\paragraph*{Feynman graphs of tensor models: the set $\mathbbm{G}(\{b_r,B_r\})$ --} We focus on the Feynman graphs which contribute to the free energy $F_N(\{t_r\}) = \ln Z_N(\{t_r\})$. Such a Feynman graph is obtained by taking a collection of $b_1$ copies of $B_1$, $b_2$ copies of $B_2$ and so on ($b_r$ copies of $B_r$) and performing Wick contractions so that the graph is connected. A Wick contraction must connect a black vertex of a bubble to a white vertex of a bubble. To distinguish the edges of the Wick contractions from the edges of the bubbles, we assign the former the {\bf fictitious color 0} (while the edges of the bubbles have the colors $\{1, \dotsc, d\}$) which will be drawn as dashed lines.

We denote $\mathbbm{G}(\{b_r,B_r\})$ the connected graphs with $b_r$ bubbles $B_r$, obtained this way. They are regular, bipartite graphs of degree $d+1$ such that each vertex is incident to each color from $\{0, 1, \dotsc, d\}$ exactly once, and such that the connected subgraphs obtained by removing the edges of color 0 are $b_r$ copies of $B_r$ for all $r\in R$.

\paragraph*{Feynman amplitudes --}
The scaling with $N$ of a Feynman graph $G\in \mathbbm{G}(\{b_r,B_r\})$ goes as follows. 
\begin{itemize}
\item It receives a factor $N^{s_r}$ for each copy of $B_r\subset G$. 
\item As is standard in tensor models, each {\bf bicolored cycle} with colors $\{0, c\}$, for all $c\in \{1, \dotsc, d\}$, contributes to one power of $N$.
\end{itemize}
The amplitude of $G$ thus goes like $N^{\delta(G)}$ with
\begin{equation} \label{delta}
\delta(G) = \sum_{c=1}^d L_{0c}(G) + \sum_{r\in R} s_r b_r
\end{equation}
where $L_{0c}(G)$ the number of bicolored cycles with colors $\{0,c\}$. For given numbers of bubbles $\{b_r\}$ there is a maximal value of the number of bicolored cycles
\begin{equation}
L_{\max}(\{b_r, B_r\}) = \max_{G\in\mathbbm{G}(\{b_r, B_r\})} \sum_{c=1}^d L_{0c}(G) \quad \text{and} \quad \delta_{\max}(\{b_r, B_r\}) = \max_{G\in\mathbbm{G}(\{b_r, B_r\})} \delta(G)
\end{equation}
obviously satisfying $\delta_{\max}(\{b_r\}) = L_{\max}(\{b_r, B_r\}) + \sum_{r\in R} s_r b_r$
and further denote the set of graphs which maximize the number of bicolored cycles at fixed numbers of bubbles
\begin{equation}
\mathbbm{G}_{\max}(\{b_r, B_r\}) = \left\{ G\in \mathbbm{G}(\{b_r, B_r\}) \quad \text{s.t.}\quad \delta(G) = \delta_{\max}(\{b_r, B_r\})\right\}.
\end{equation}
We can use equivalently $\delta(G)$ or $\sum_{c=1}^d L_{0c}(G)$ since there is only a shift between them at fixed $\{b_r\}$.

\paragraph*{Large $N$ limit: existence --} We say that {\bf a large $N$ limit exists} if there exists a rescaling $q$ of the free energy $F_N(\{t_r\}) = \ln Z_N(\{t_r\})$ such that $F_N(\{t_r\})/N^q$ has a finite limit as $N\to\infty$. Since $F_N$ has an expansion onto $\mathbbm{G}(\{b_r, B_r\})$ where each graph is weighted by $N^{\delta(G)}$, the existence of a (perturbative) large $N$ limit is equivalent to the existence of a bound on $\delta(G)$ independently on the numbers of bubbles,
\begin{equation}
q = \delta_{\max}(\{B_r\}) = \max_{\{b_r\}} \delta_{\max}(\{b_r, B_r\}) <\infty
\end{equation}
Then the free energy has a $1/N$ expansion which starts like
\begin{equation}
F_N(\{t_r\}) = N^{\delta_{\max}(\{B_r\})} \Bigl(\sum_{\{b_r\}}\ \sum_{G\in\mathbbm{G}_{\max}(\{b_r, B_r\})} s(G) \prod_{r\in R} (-t_r)^{b_r} + \mathcal{O}(1/N)\Bigr),
\end{equation}
where $s(G)$ is a combinatorial factor coming from the symmetries of the bubbles $\{B_r\}$ and of $G$.

\paragraph*{Large $N$ limit: non-triviality --} Furthermore we say that the large $N$ limit is {\bf non-trivial} if there is an infinite number of graphs in the above sum. In fact, we want a stronger condition so that each $B_r$ for $r\in R$ contributes non-trivially: 
\begin{equation}
\forall\,r\in R, \qquad \bigcup_{\{b_{r'}\}} \mathbbm{G}_{\max}(\{b_{r'}, B_{r'}\}) \quad \text{contains graphs with $b_r \to\infty$}.
\end{equation}

\paragraph*{Large $N$ limit of tensor models with GM interactions --} Although it has not appeared yet in the literature, it is basically a matter of assembling pieces already present in the literature. It is obtained in two steps. 
\begin{itemize}
\item The first step is a now classical result in tensor models: Feynman graphs of a given model form a subset of those of the quartic model. It was proved in \cite{StuffedWalshMaps} using the intermediate field, but it was used prior in specific models \cite{DoubleScaling}. We explain how it works below in Section \ref{sec:QuarticToSP}.
\item The second step is to find the large $N$ limit of tensor models with quartic interactions. This is a slight generalization of \cite{MelonoPlanar}. It is done in Section \ref{sec:QuarticLargeN}.
\end{itemize}

\subsection{From quartic to arbitrary GM bubbles} \label{sec:QuarticToSP}

We start the analysis with a single type of GM bubble $B$, with $V$ vertices. Let $H\in \mathbbm{G}_V$ with $\partial H = B$, where we recall that the boundary operator $\partial$ performs the contractions of the all dashed lines in $H$. The vertices of $H$ with no incident dashed line are the vertices of $B$ (after the contractions of all dashed lines). So we can consider the set $\mathbbm{G}(b, H)$ of connected graphs 
\begin{itemize}
\item with $b$ copies of $H$
\item which are connected by additional dashed lines so that each vertex is incident to exactly one dashed line.
\end{itemize}
Therefore $G\in\mathbbm{G}(b, H)$ is made of quartic bubbles (those of all copies of $H$ in $G$), which are connected along dashed lines, such that each vertex is incident to all colors in $\{1, \dotsc, d\}$ plus a dashed line. Interpreting the latter as edges of color 0, we find that $\mathbbm{G}(b, H)$ is a subset of the graphs from the quartic model. Say $H$ has $b_C(H)$ bubbles of type $Q_C$, then 
\begin{equation}
\mathbbm{G}(b, B) \leftarrow \mathbbm{G}(b, H) \subset \mathbbm{G}(\{b_C(H)\, b, Q_C\}),
\end{equation}
where the left arrow corresponds to replacing $H$ with $B=\partial H$ in all graphs (see Proposition \ref{thm:Surjection} below).

In the more general case of a set of bubbles $\{B_r\}$, we introduce a graph $H_r\in \mathbbm{G}_{V_r}$ such that $\partial H_r = B_r$, made of $b_C^{(r)}$ copies of the quartic bubbles $Q_C$ for all $C$. We further define $\mathbbm{G}(\{b_r, H_r\})$ as the set of graphs obtained by connecting $b_r$ copies of $H_r$ along dashed lines as above. Then
\begin{equation} \label{Inclusion}
\mathbbm{G}(\{b_r, H_r\}) \subset \mathbbm{G}\Bigl(\Bigl\{ \sum_r b_C^{(r)} b_r, Q_C\Bigr\}\Bigr).
\end{equation}
Notice that although the choice of $H_r$ is not canonical, the numbers $\{b_C^{(r)}\}_{C\in\mathcal{C}}$ are uniquely determined by $B$, as shown in Proposition \ref{thm:SetsUniqueness}.

\begin{proposition} \label{thm:Surjection}
There is a surjection
\begin{equation}
S:\ \mathbbm{G}(\{b_r, H_r\})\ \to\ \mathbbm{G}(\{b_r, B_r\})
\end{equation}
which preserves the number of bicolored cycles with colors $\{0,c\}$ for all $c\in\{1, \dotsc, d\}$, where the dashed lines of each $H_r$ are interpreted as edges of color 0.
\end{proposition}

In other words, the set of graphs $\mathbbm{G}(\{b_r, B_r\})$ of our model can be thought of as just a {\bf subset of the graphs from the quartic model}, with the same number of bicolored cycles.

\begin{proof} 
This is a corollary of the results of \cite{StuffedWalshMaps} which in our case can be proven directly in a few lines. If $G\in\mathbbm{G}(\{b_r, H_r\})$, $S(G)$ is obtained by applying the boundary operator $\partial$ to each copy of $H_r$. $G$ thus encodes a gluing of $b_r$ copies of $B_r = \partial H_r$ along dashed lines, as expected. Since $\partial$ transforms a bicolored path in $H_r$ with colors $\{0,c\}$ between two vertices with no incident dashed line into a single edge of color $c$ in $B_r$, then a bicolored cycle of colors $\{0,c\}$ in $G$ remains a bicolored cycle with colors $\{0,c\}$ in $S(G)$. There are no additional bicolored cycles in $S(G)$ which shows that the number of bicolored cycles is preserved.

Surjectivity is obtained by considering $G\in \mathbbm{G}(\{b_r, B_r\})$ and replacing each copy of $B_r$ with $H_r$. Remember that one can identify the vertices of $B$ with those of $H$ which do not have an incident dashed line -- this is necessary to know how to place $H_r$ instead of $B_r$.
\end{proof}

\subsection{Large \texorpdfstring{$N$}{N} limit for tensor models with quartic interactions} \label{sec:QuarticLargeN}

Recall that the quartic bubbles are identified by their color sets $C\subset \{1, \dotsc, d\}$ where $C$ is admissible if $|C|\leq d/2$, and if $|C|=d/2$ then $1\in C$. The finite set $R$ is thus $R=\mathcal{C}$ the set of admissible color sets. The most generic quartic model reads
\begin{equation} \label{QuarticModel}
Z_N(\{t_C\}) = \int dT d\bar{T}\ \exp -\Bigl((T|T) + \sum_{C\in \mathcal{C}} N^{s_C}\,t_C\, Q_C(T, \bar{T}) \Bigr)
\end{equation}
and $\mathbbm{G}(\{b_C, Q_C\})$ is the set of connected Feynman graphs with $b_C$ bubbles of type $Q_C$, and $\mathbbm{G}_{\max}(\{b_C, Q_C\})$ the subset which dominates at large $N$.

To study the large $N$ limit, we will use a convenient representation called the {\bf intermediate field} which provides a bijection between $\mathbbm{G}(\{b_C, Q_C\})$ and a set $\mathbbm{M}(\{b_C\})$ of decorated combinatorial maps, see Theorem \ref{thm:QuarticIntField} below.

\smallskip

\paragraph*{Intermediate field representation --} The intermediate field for quartic models (not just tensors) is simply the Hubbard-Stratonovich transformation applied to the partition function in order to integrate the original variables. We can of course apply it to \eqref{QuarticModel} which gives a matrix model whose large $N$ limit can then be studied \cite{ConstructiveQuartic}.

Since our goal is to characterize $\mathbbm{G}_{\max}(\{b_C, Q_C\})$, we will instead view the intermediate field as a bijection between the set of Feynman graphs $\mathbbm{G}(\{b_C, Q_C\})$ and the one of the matrix model after the Hubbard-Stratonovich transformation. In terms of Feynman graphs, the intermediate field is a {\bf bijection} which maps some cycles into vertices equipped a cyclic order. We refer to \cite{DoubleScaling-Dartois} for its first use in the context of tensor models and to \cite{StuffedWalshMaps} for details.

In the intermediate field representation, one first has to choose for each quartic bubble a {\bf pairing} of its vertices. For later purposes, it is convenient that the vertices of each pair are connected with as many colors as possible. We thus choose the canonical pairing $\pi_{Q_C}$ but any other can do in principle. To shorten the notation, we use the bar to denote the partner of a vertex, i.e. $\pi_{Q_C}(v) = \bar{v}$.

In the intermediate field representation, the quartic bubble $Q_C$ is represented as an edge carrying the color set $C$,
\begin{equation} \label{IFEdge}
Q_C = \begin{array}{c} \includegraphics[scale=.4]{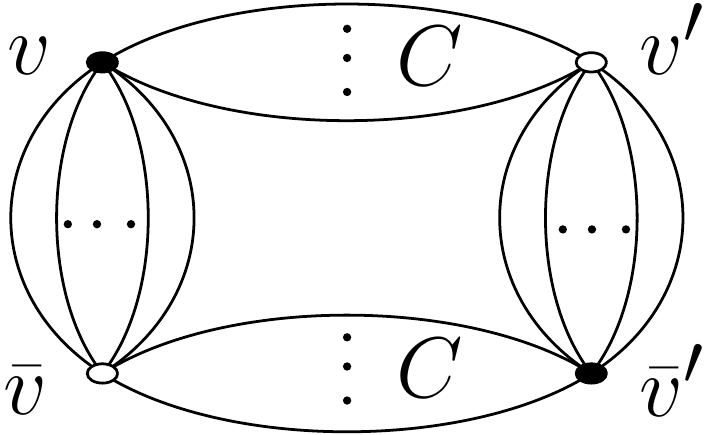} \end{array} \qquad \to \qquad \begin{array}{c} \includegraphics[scale=.4]{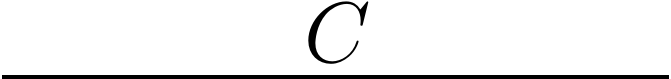} \end{array}
\end{equation}
which should be thought of as two half-edges where each half-edge represents a canonical pair of vertices of $Q_C$, i.e. $\{v, \bar{v}\}$ and $\{v', \bar{v}'\}$. If we had used the other pairing of $Q_C$, the edge would instead carry the color set $\widehat{C}$ as label.

Next we need to encode the connectivity between quartic bubbles. To do this, we look at the paths which alternate canonical pairs of vertices and dashed lines. More precisely, we start from a vertex $v_1$ then jump to its partner $\bar{v}_1$. Then follow the incident dashed line to a vertex $v_2$, then to its partner $\bar{v}_2$ and so on, until one comes back to $v_1$. This produces a cycle $(v_1, \bar{v}_1, v_2, \bar{v}_2, \dotsc)$. In the intermediate field representation, it is represented as a vertex whose incident half-edges correspond to the pairs $\{v_i, \bar{v}_i\}$. Since the latter are ordered along the cycle, it means that the half-edges incident to the vertex in the intermediate field representation are also ordered, similarly to the case of combinatorial maps,
\begin{equation} \label{IFClosedVertex}
\begin{array}{c} \includegraphics[scale=.4]{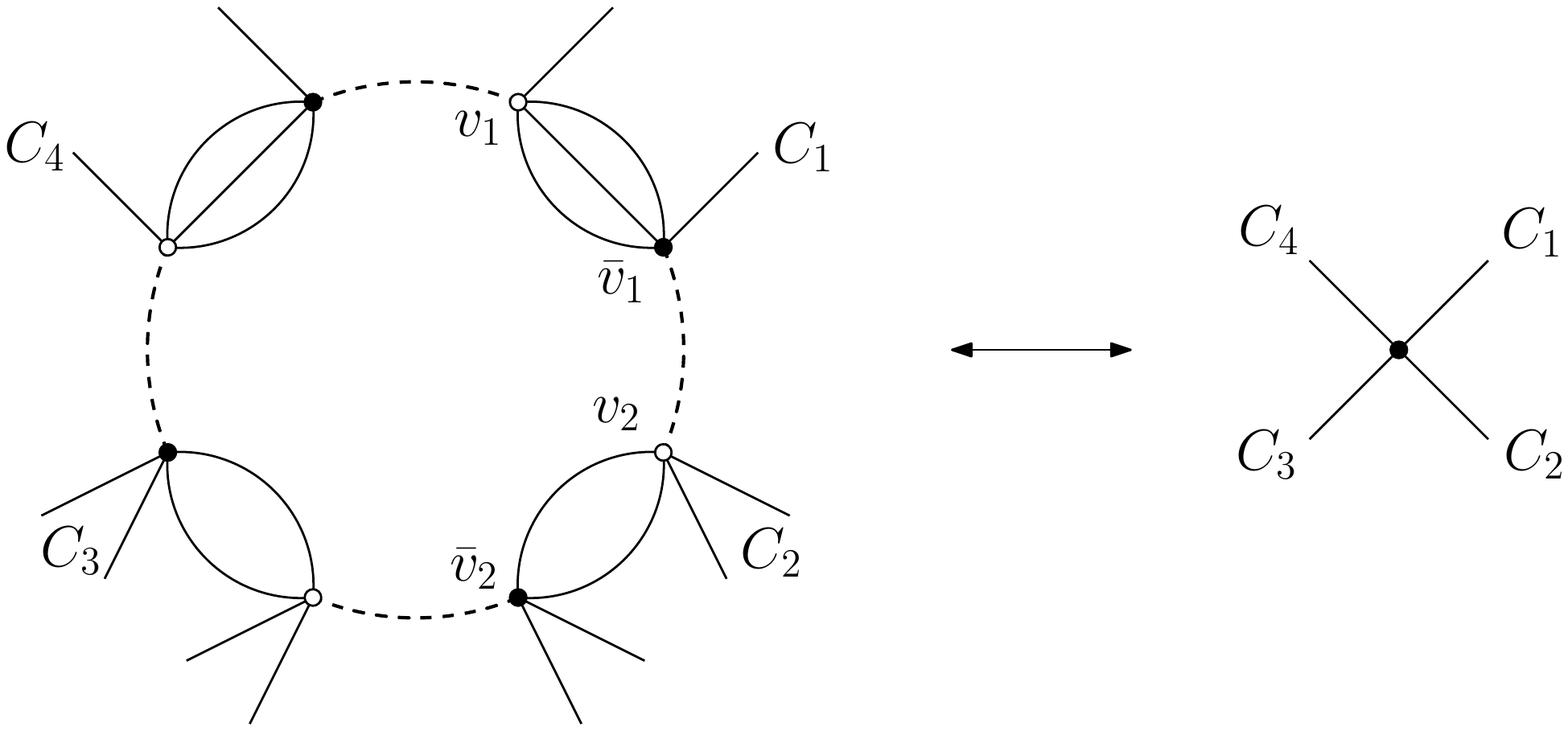} \end{array}
\end{equation}

\paragraph*{The set $\mathbbm{M}(\{b_C\})$ --} This is the set of connected combinatorial maps $M$ with edges decorated by color sets $\{C_e\}_{e\in M}$ and such that the color sets $C$ appears on $b_C$ edges. For $M\in \mathbbm{M}(\{b_C\})$ we define
\begin{itemize}
\item For $C\subset \{1, \dotsc, d\}$, the map $M_C$ is the submap of $M$ obtained by deleting all edges whose color sets $C'$ are not $C$.
\item For $c\in\{1, \dotsc, d\}$, the map $M_c$ is the submap obtained by removing all edges $e$ where $c\not\in C_e$.
\end{itemize}

Recall that the {\bf face} of a combinatorial map is a closed path obtained by following alternatingly edges and corners counter-clockwise. The faces of color $c\in\{1, \dotsc, d\}$ of $M$ are defined as the faces of $M_c$. Notice that if there is a vertex in $M$ whose incident edges do not have the color $c$, then $M_c$ has an isolated vertex and it contributes to one face (and one connected component).


\begin{theorem} \label{thm:QuarticIntField}
There is a bijection $J:\ \mathbbm{G}(\{b_C, Q_C\})\,\to\,\mathbbm{M}(\{b_C\})$ where the bicolored cycles of color $\{0, c\}$ are mapped to the faces of color $c$. In particular
\begin{equation}
L_{0c}(G) = F_{c}(J(G))
\end{equation}
where $F_c(M)$ denotes the number of faces of $M_c$.
\end{theorem}

\begin{proof}
This is again a corollary of \cite{StuffedWalshMaps}, in our case simple enough to be proved in a few lines. The bijection is quite obvious, as detailed before. The main interest is the relation between bicolored cycles and faces. A bicolored cycle with colors $\{0,c\}$ alternates edges of color 0, which we can orient, say, from black to white vertices, and edges of color $c$ which we orient from white to black vertices so that the cycle itself is oriented.

In the intermediate field representation, an edge of color 0 from a black to a white vertex corresponds to a corner, clockwise-oriented from say $C_1$ to $C_2$, as can be seen in \eqref{IFClosedVertex}. Then for the edge of color $c$, there are two possibilities. Either $c\in C_2$, then one follows the edge with $C_2$ in the intermediate field representation. Or $c\not\in C_2$, then the edge of color $c$ corresponds in the intermediate field representation to the clockwise crossing of the edge with $C_2$.

Therefore, one follows corners, and edges which have the color $c\in C_e$. This is the definition of a face of color $c$.
\end{proof}

We can thus ``identify'' $G$ and $M=J(G)$. In particular we use the short hand
\begin{equation}
\delta(M) = \delta(G) \qquad \text{if $M=J(G)$}.
\end{equation}
From the above theorem, we can define
\begin{equation}
\mathbbm{M}_{\max}(\{b_C\}) = J\left( \mathbbm{G}_{\max}(\{b_C, Q_C\})\right)
\end{equation}
which coincides with the set of maps $M\in\mathbbm{M}(\{b_C\})$ which maximize the number of faces, i.e. $\sum_{c=1}^d F_{c}(M)$.

\begin{lemma} \label{lemma:Unhooking}
If $e$ is an edge in $M\in\mathbbm{M}(\{b_C\})$ which is not a bridge, then unhooking $e$ from one of its end, i.e.
\begin{equation} \label{Unhooking}
M = \begin{array}{c} \includegraphics[scale=.55]{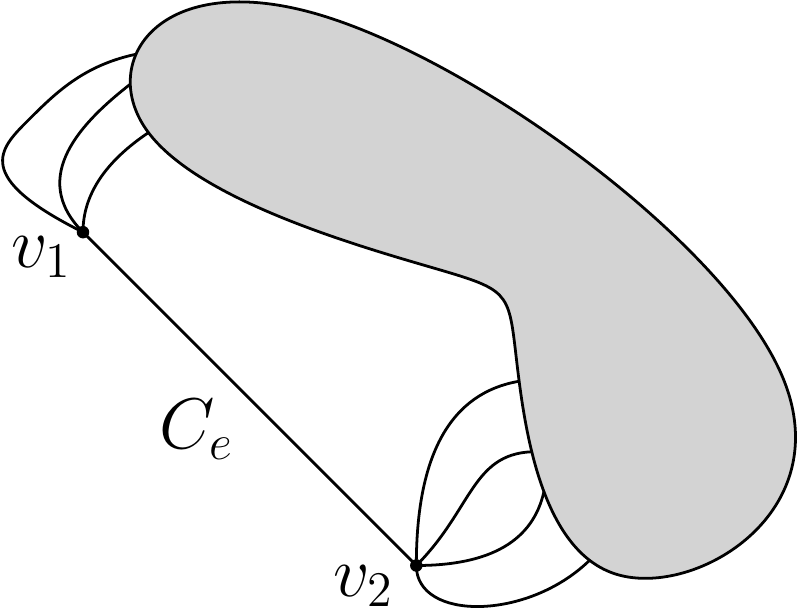} \end{array} \qquad \to \qquad M' = \begin{array}{c} \includegraphics[scale=.55]{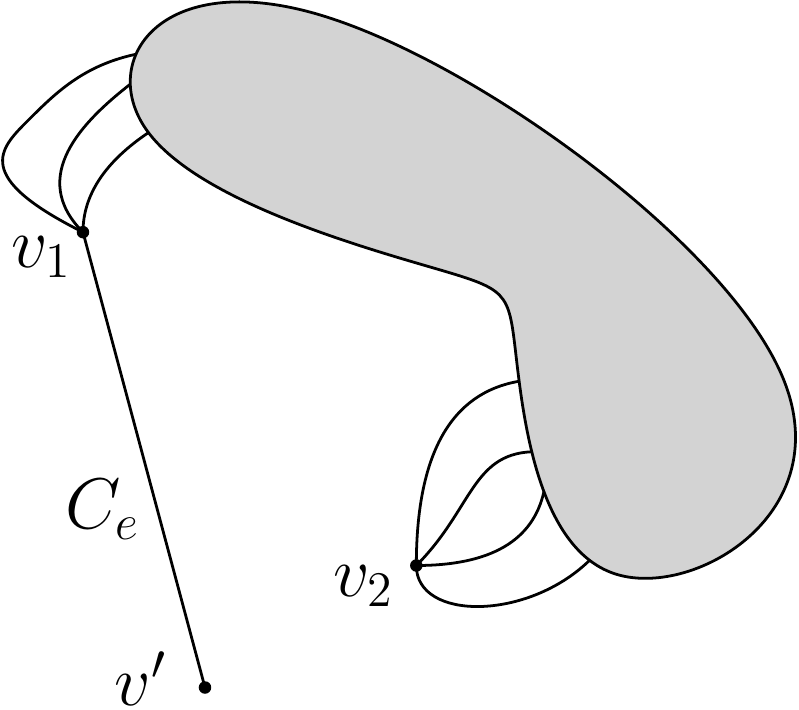} \end{array}
\end{equation}
gives $M'\in\mathbbm{M}(\{b_C\})$ such that
\begin{equation}
\delta(M) \leq \delta(M').
\end{equation}
Moreover
\begin{equation}
\delta(M) < \delta(M') \qquad \text{if $|C_e| <d/2$.}
\end{equation}
\end{lemma}

\begin{proof}
Since $e$ is not a bridge, its unhooking leaves the map connected, hence $M'\in\mathbbm{M}(\{b_C\})$. We study the variations of the number of faces for each color $c\in\{1, \dotsc, d\}$. 
\begin{itemize}
\item For $c\in C_e$, there may be one or two faces going along $e$ in $M_c$. In $M'_c$ the edge $e$ has become a bridge attached to a leaf $v'$ and there is now a single face. Thus $F_c(M) \leq F_c(M') + 1$.
\item For $c\not\in C_e$, the map $M'_c$ has an additional face, which is due to the vertex $v'$ alone being a new connected component. Thus $F_c(M) = F_c(M') -1$.
\end{itemize}
Summing over $c\in\{1, \dotsc, d\}$ we get
\begin{equation}
\sum_{c=1}^d F_c(M) \leq \sum_{c\in C_e} \bigl(F_c(M') + 1\bigr) + \sum_{c\not\in C_e} \bigl(F_c(M') -1\bigr) = \sum_{c=1}^d F_c(M') - d + 2|C_e|.
\end{equation}
Recall that $\delta(M)$ is just obtained from shifts of $\sum_{c=1}^d F_c(M)$ by quantities proportional to $b_C$s which are here constant. With $|C_e| \leq d/2$, we directly get $\delta(M) \leq \delta(M')$ with the strong inequality if $|C_e|<d/2$.
\end{proof}

\begin{lemma} \label{thm:Submaps}
If $s_C \leq |C|-d$, and $\tilde{M}\subset M$ is a connected submap of $M\in \mathbbm{M}(\{b_C\})$, then
\begin{equation}
\delta(M) \leq \delta(\tilde{M}).
\end{equation}
\end{lemma}

This lemma shows that with a crucial bound on $s_C$, the function $\delta$ is monotonous for the inclusion.

\begin{proof}
We fix a map $\tilde{M}\in\mathbbm{M}(\{b_C\})$ and proceed by induction on $M\supset \tilde{M}$, i.e. the number of edges $E(M\setminus \tilde{M})$ of $M$ which are not in $\tilde{M}$. If $E(M\setminus \tilde{M})=0$, then $M = \tilde{M}$, there is nothing to prove.

For $E\geq 0$, assume the proposition holds for all $M\supset \tilde{M}$ such that $E(M\setminus \tilde{M}) = E$, and let $M_0\supset \tilde{M}$ such that $E(M_0\setminus \tilde{M})=E+1$. We can always find in $M_0\setminus \tilde{M}$ an edge $e$ which is
\begin{itemize}
\item either not a bridge in $M_0$,
\item or connected to a leaf.
\end{itemize}
We can reduce the first case to the second case as follows. Since $e$ is not a bridge, it can be unhooked as in Lemma \ref{lemma:Unhooking} which gives $M'$ such that $\delta(M_0) \leq \delta(M')$. Then we have an edge connected to a leaf as in the second case. 

Having $e$ connected to a leaf, we can remove the edge and the leaf to get a new map $M'$, still connected and with one less edge of type $C_e$. $M'$ loses a face of color $c$ for all $c\in \widehat{C}_e$ compared to $M_0$. Therefore
\begin{equation} \label{LeafRemoval}
\begin{aligned}
\delta(M') &= \sum_{c\in C_e} F_c(M_0) + \sum_{c\not\in C_e} \bigl(F_c(M_0) - 1) + \sum_{C} s_C b_C  - s_{C_e} \\
& = \delta(M_0) + \bigl(|C_e| - d - s_{C_e}\bigr) \geq \delta(M_0).
\end{aligned}
\end{equation}
according to our hypothesis on $s_C$. Moreover, the induction hypothesis is $\delta(\tilde{M})\geq \delta(M')$, which yields $\delta(\tilde{M}) \geq \delta(M_0)$. This completes the induction.
\end{proof}

\begin{theorem} \label{thm:Quartic}
There is a non-trivial large $N$ limit if and only if 
\begin{equation}
s_C = |C| - d.
\end{equation}
It is such that for all $G\in \bigcup_{\{b_r\}} \mathbbm{G}_{\max}(\{b_C,B_C\})$
\begin{equation}
\delta(G) = \delta_{\max}(\{B_C\}) = d.
\end{equation}
Moreover the set $\mathbbm{M}_{\max}(\{b_C\})$ is defined as maps $M\in\mathbbm{M}(\{b_C\})$ such that
\begin{itemize}
\item $M$ is planar
\item edges with $|C_e| <d/2$ are bridges
\item $M_C$ is planar for all admissible color sets $C$
\item $M_C$ and $M_{C'}$ for $C\neq C'$ can only meet at cut-vertices.
\end{itemize}
\end{theorem}

\begin{proof}
\begin{description}[wide=0pt]
\item[Edges with $|C_e| <d/2$ are bridges] Consider $e$ in $M$ which is not a bridge. $M'$ is obtained by unhooking $e$ as in Lemma \ref{lemma:Unhooking}. 
\begin{itemize}
\item If $|C_e|<d/2$, then $\delta(M')>\delta(M)$. Therefore if $M\in \mathbbm{M}_{\max}(\{b_C\})$, then all edges with $|C_e|<d/2$ are bridges.
\item If $|C_e|=d/2$, then the weaker inequality $\delta(M')\geq \delta(M)$ is obtained. This implies that if we let $T\subset M$ be a spanning tree in $M$, we can unhook all edges $e\in M\setminus T$ as above and find this way
\begin{equation}
\delta(M) \leq \delta(T).
\end{equation}
\end{itemize}
This shows that plane trees belong in $\mathbbm{M}_{\max}(\{b_c\})$ for all values of $\{b_C\}$. In particular, if there is a large $N$ limit, then it is non-trivial as it would contain all plane trees with all possible values of $C$ its edges.

\item[Large $N$ limit and the value of $s_C$] To make sure that there is indeed a large $N$ limit, we must still find the value of $s_C$. It is found by counting the number of faces of plane trees as a function of the $b_C$s. Notice that the trivial map, reduced to a vertex, has $d$ faces. Assume $T$ is a tree with $b_C$ edges having color set $C$. Let $e$ be an edge connected to a leaf and $T'$ be $T$ with $e$ and its leaf removed. Clearly $T'$ has one face less than $T$ for each color in $\widehat{C}_e = \{1, \dotsc, d\}\setminus C_e$. Thus the same calculation as in \eqref{LeafRemoval} shows that
\begin{equation}
\sum_{c=1}^d F_c(T) = \sum_{c=1}^d F_c(T') + d - |C_e| \qquad \text{and} \qquad \delta(T) = \delta(T') + \bigl(d-|C_e| + s_{C_e}\bigr).
\end{equation}
A trivial induction thus shows that
\begin{equation}
\sum_{c=1}^d F_c(T) = d + \sum_C (d-|C|) b_C  \qquad \text{and} \qquad \delta(T) = d + \sum_C \bigl(d-|C|+s_C)\bigr) b_C.
\end{equation}
Therefore the existence of a large $N$ limit requires $s_C \leq |C|-d$ or else trees can have arbitrarily large values of $\delta(T)$ as the numbers of edges $\{b_C\}$ grows. 

In particular, this shows that for $s_C = |C|-d$ there is a non-trivial large $N$ limit where all trees contribute and for which $\delta(T) = \delta_{\max}(\{B_C\}) = d$.

Let us show that it is the {\bf only possible value} of $s_C$. Denote $\sigma_C = |C|-d - s_C \geq 0$ the difference between $s_C$ and the above value. Then, from $\delta(M)\leq \delta(T)$ we find
\begin{equation}
\delta(M) \leq \sum_{c=1}^d F_c(T) + \sum_C s_C b_C = \sum_{c=1}^d F_c(T) + \sum_C (|C|-d) b_C - \sum_C \sigma_C\,b_C = d - \sum_C \sigma_C\,b_C.
\end{equation}
This gives the bound $\sum_C \sigma_C b_C \leq d - \delta(M)$. If a $\sigma_C>0$ is non-vanishing, then graphs at fixed $\delta$ can only have a finite number $b_C$ of bubbles $B_C$, and thus the large $N$ limit with respect to $Q_C$ is trivial.

\item[Consequence of Lemma \ref{thm:Submaps} and $\delta(M)\leq d$] Assume that $M$ maximizes the number of faces, i.e. $\delta(M)=d$. Then, from Lemma \ref{thm:Submaps}, we know that all submaps $\tilde{M}$ of $M$ must satisfy $\delta(\tilde{M}) = d$ too since it is the upper bound.

\item[$M_C$ is planar for all colors sets] Consider a color set $C$. If $|C|<d/2$, then all edges of $M_C$ are bridges, as seen above, and $M_C$ is a forest, which is always planar. We thus focus on the case $|C|=d/2$, and consider $\tilde{M}$ a connected component of $M_C$. As an ordinary map, $\tilde{M}$ has $F(\tilde{M})$ faces, $E(\tilde{M})$ edges and $V(\tilde{M})$ vertices. Its faces coincide with the faces of colors $c\in C$ of $\tilde{M}$, while its vertices coincide with the faces of colors $c\not\in C$, and there are $d/2$ such colors in both cases. We get
\begin{equation}
\delta(\tilde{M}) = \sum_{c=1}^d F_c(\tilde{M}) - (d-|C|) E(\tilde{M}) = \frac{d}{2} F(\tilde{M}) + \frac{d}{2} V(\tilde{M}) - (d-|C|) E(\tilde{M}) = \frac{d}{2} \bigl(2-2g(\tilde{M})\bigr)
\end{equation}
where $g(\tilde{M})$ is the genus of $\tilde{M}$ seen as an ordinary map. Therefore $\delta(\tilde{M}) = d$ if and only if $\tilde{M}$ is planar. We conclude that all connected components of $M_C$ must be planar for $M\in\mathbbm{M}_{\max}(\{b_C\})$.

\item[$M_C$ and $M_{C'}$ can only meet at cut-vertices] Suppose that $M$ has a cycle made of edges with at least two different color sets. Denote $\tilde{M}$ the submap consisting of this cycle, and $v$ a vertex where edges $e, e'$ with two different color sets $C, C'$ meet. Let us study the number of faces of $\tilde{M}$ by unhooking $e'$,
\begin{equation}
\tilde{M} = \begin{array}{c} \includegraphics[scale=.4]{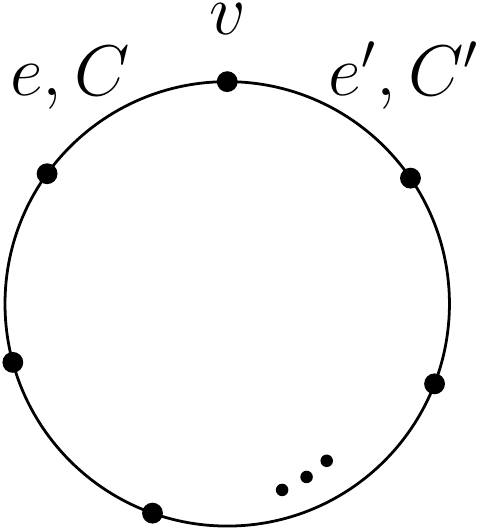} \end{array} \qquad \to \qquad \tilde{M}' = \begin{array}{c} \includegraphics[scale=.4]{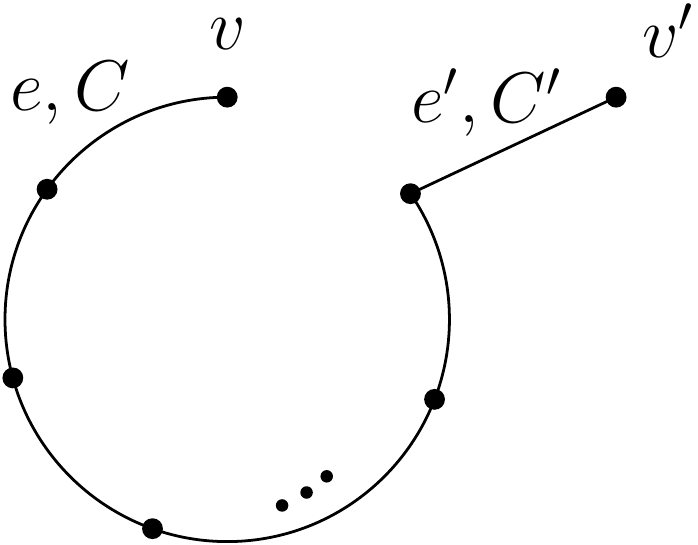} \end{array}
\end{equation}
Denote the color sets along the cycle $C_i$, for $i\in I$ some finite set. First, we know that $|C_i|=d/2$ otherwise all those edges would be cut-edges. For all colors which are present all along the cycle, i.e. $c\in \bigcap_i C_i$, $\tilde{M}$ has two faces and $\tilde{M}'$ only one, so at least $|\bigcap_i C_i|\leq d/2-1$ are lost, where the inequality comes from the fact that there are different color sets along the cycle. However, for all the other colors, and there are more than $d/2$ of them, $\tilde{M}'$ creates one additional face. It is easy to check by investigating the various cases: $c$ belongs in both $C$ and $C'$ but not in the color set of another edge of the cycle, $c$ belongs in $C$ but not in $C'$ or the other way around, or $c$ belongs neither to $C$ nor $C'$. Therefore we find 
\begin{equation}
\sum_{c=1}^d F_{c}(\tilde{M}') > \sum_{c=1}^d F_{c}(\tilde{M})
\end{equation}
and since they have the same number of edges of each type, it comes $\delta(\tilde{M})<d$. We conclude with Lemma \ref{thm:Submaps} that $\delta(M)<d$ if it contains a cycle with multiple color sets. This is equivalent to $M\in\mathbbm{M}_{\max}(\{b_C\})$ has its submaps $M_C, M_{C'}$ meeting at cut-vertices only.

\item[$M$ is a planar map] It is enough to consider the case where all edges have color sets $|C_e|=d/2$. Let us use an induction on the number of cut-vertices. If there is none, i.e. $M$ is non-separable, then all edges must have the same color set $C$ and $M=M_C$ is planar as proved above.

Let us consider $M$ with at least one cut-vertex $v$ where the incident edges have color sets $C_1, C_2, \dotsc$ We split $M$ into maps $M^{(1)}, M^{(2)}, \dotsc$ as follows. For $M^{(1)}$, remove all edges incident to $v$ except those of type $C_1$, and erase all the other connected components created by those edge removals. For $M^{(2)}$, we do the same with $C_2$ and so on. The maps $M^{(i)}$ are connected and $M$ is formed by gluing them back at $v$, with the appropriate cyclic order of the incident edges. 

Since $M^{(i)}$ touches the other maps at $v$ only, if we remove $v$ from $M^{(i)}$ we get a map with less cut-vertices than $M$. By induction, we know that all the connected components obtained by removing $v$ are planar. To check that $M^{(i)}$ itself is planar, we must look at the way those connected components are glued around $v$ in $M^{(i)}$. Remember that all the edges incident to $v$ in $M^{(i)}$ have the same color type $C_i$ and that $M_{C_i}$ is planar. We conclude that $M^{(i)}$ is planar.

Therefore the non-planarity of $M$ can only come from the way the maps $(M^{(i)})$ are glued together at $v$. Clearly, if $M$ is not planar, it is sufficient to permute the order of some edges incident on $v$ to make it planar, without modifying any of the maps $M^{(i)}$. When permuting the order of two edges, one following the other in the cyclic order around $v$, with different colors sets $C_1, C_2$, the faces of color $c$ for $c\not\in C_1$ and $c\not\in C_2$ are unaffected. Faces are affected only if their colors are in $C_1\cap C_2$ which is not empty since both $C_1$ and $C_2$ contain the color 1. Permuting the order of two such edges at $v$ does not change the number of connected components of the maps $M_c$ for all $c\in\{1, \dotsc, d\}$. Therefore, when $M$ is made planar using those permutations, the number of faces increases. We conclude that $M$ has to be planar to maximize the number of faces.
\end{description}
\end{proof}

%
%

\subsection{Large \texorpdfstring{$N$}{N} limit for tensor models with GM interactions} \label{sec:GMBLargeN}


In this section, we prove the following explicit formulation of Theorem \ref{thm:Scaling}.

\begin{theorem*}[\ref{thm:Scaling}]
There is a non-trivial large $N$ limit if and only if the scaling coefficients are
\begin{equation}
s_r = \sum_{C} |C| b_C^{(r)} - \frac{d(V_r-2)}{2}.
\end{equation}
in particular $s_r>(d-1)(V_r-2)/2$ except for ordinary melonic bubbles for which $s_r=(d-1)(V_r-2)/2$.

Let $\mathbbm{G}_{\max}(\{b_r, H_r\})\subset \mathbbm{G}(\{b_r, H_r\})$ be the subset of graphs for which $\delta(G) =d$. Then
\begin{equation}
\mathbbm{G}_{\max}(\{b_r, B_r\}) = S\bigl(\mathbbm{G}_{\max}(\{b_r, H_r\})\bigr),
\end{equation}
and $\mathbbm{G}_{\max}(\{b_r, H_r\}) \subset \mathbbm{G}_{\max}\Bigl(\Bigl\{ \sum_r b_C^{(r)} b_r, Q_C\Bigr\}\Bigr)$.
\end{theorem*}

The subset $\mathbbm{G}_{\max}\Bigl(\Bigl\{ \sum_r b_C^{(r)} b_r, Q_C\Bigr\}\Bigr)$ is described in Theorem \ref{thm:Quartic}. The graphs which dominate the large $N$ limit of tensor models with GM interactions are thus a subset of those from the quartic case.

The main ingredients of the proof are obviously Proposition \ref{thm:Surjection} and Theorem \ref{thm:Quartic}. In addition, to prove that the large $N$ limit is non-trivial and unique, we simply reproduce for completeness the proof of Lionni from \cite{PhDLionni}. For that purpose, it will be useful to know the number of bicolored cycles for a family of graphs with arbitrary number of bubbles.

\begin{lemma} \label{thm:Trees}
There exists a set of graphs in $\mathbbm{G}(b,B)$ for all $b\geq 1$ whose numbers of bicolored cycles with colors $\{0,c\}$ for all $c\in\{1, \dotsc, d\}$ is
\begin{equation}
\sum_{c=1}^d L_{0c}(G) = \biggl(\frac{d(V-2)}{2} - \sum_C |C| b_C\biggr) b + d.
\end{equation}
where $b_C$ is the number of $C$-bidipole insertions performed to construct $B$, or equivalently the number of quartic bubbles $Q_C$ in $H\in\mathbbm{G}_V$ with $\partial H=B$.
\end{lemma}

\begin{proof}
We first construct such a graph with a single bubble, $b=1$. Consider $B$ with its canonical pairing $\pi_B$ and $G_B\in\mathbbm{G}(1,B)$ obtained by connecting the vertices of every pair $\{v, \pi_B(v)\}$ with an edge of color 0. By induction on the number of vertices of $B$, it can be checked that 
\begin{equation}
\sum_{c=1}^d L_{0c}(G_B) = d + \sum_C (d-|C|) b_C = \frac{dV}{2} - \sum_C |C| b_C
\end{equation}
since every new $C$-bidipole insertion with the two vertices connected by $\widehat{C}$ also connected by an edge of color 0 adds exactly one bicolored cycle for each color in $\widehat{C}$ and does not change the number of bicolored cycles for $c\in C$.

Then consider two copies of $G_B$ and cut an edge of color 0 in both copies and glue the half-edges of color 0 so as to obtain a new connected graph $G^{(1)}_B$. The number of bicolored cycles is $2\sum_{c=1}^d L_{0c}(G_B) - d$ since the gluing merges $d$ bicolored cycles. By repeating this operation with $G^{(1)}_B$ and $G_B$, we get $G^{(2)}_B$ and so on. When $G$ built this way has $b$ copies of $B$, the number of bicolored cycles is thus
\begin{equation}
b \sum_{c=1}^d L_{0c}(G_B) - (b-1)d = d + \frac{b\,d}{2}(V-2) - b\sum_C |C| b_C.
\end{equation}
\end{proof}

\begin{proof}[Proof of Theorem \ref{thm:Scaling}]
There are two notions of $\delta$ on the set $\mathbbm{G}(\{b_r, H_r\})$. One is due to the inclusion \eqref{Inclusion} as a subset of the quartic model
\begin{equation}
\delta_Q(G) = \sum_{c=1}^d L_{0c}(G) + \sum_C \Bigl(\sum_r b_C^{(r)} b_r\Bigr) s_C
\end{equation}
with a scaling coefficient $s_C$ for each quartic bubble inherited from the quartic model. The other is due to the surjection $S$ and associates the scaling coefficient $s_r$ to each copy of $H_r$,
\begin{equation}
\delta(G) = \sum_{c=1}^d L_{0c}(G) + \sum_r b_r\,s_r.
\end{equation}
This is the $\delta$ coming from the original model, i.e. on $\mathbbm{G}(\{b_r,B_r\})$.

From Theorem \ref{thm:Quartic}, there is a unique way to make $\delta_Q$ bounded for all graphs in $\mathbbm{G}\Bigl(\Bigl\{ \sum_r b_C^{(r)} b_r, Q_C\Bigr\}\Bigr)$ and with non empty sets of graphs satisfying $\delta(G)=d$ when $\sum_r b_C^{(r)} b_r \to \infty$ for any $C$. This is $s^*_C = |C|-d$. In addition, equating $\delta_Q = \delta$ on $\mathbbm{G}(\{b_r, H_r\})$ gives
\begin{equation}
s_r = \sum_C b_C^{(r)}\ s_C
\end{equation}
i.e. the scaling coefficient of $B_r$ is just the sum of the scaling coefficients of its quartic decomposition $H_r$. Using the unique values $s^*_C = |C|-d$, one finds 
\begin{equation}
s_r^* = \sum_{C} |C| b_C^{(r)} - \frac{d(V_r-2)}{2}
\end{equation}
as given in the Theorem.

With those values, we clearly have $\delta(G)\leq d$ for all graphs, meaning that a large $N$ limit exists. We define $\mathbbm{G}_{\max}(\{b_r,B_r\})$ as the image via $S$ of 
\begin{equation}
\mathbbm{G}_{\max}(\{b_r, H_r\}) = \mathbbm{G}(\{b_r, H_r\}) \cap \mathbbm{G}_{\max}\Bigl(\Bigl\{ \sum_r b_C^{(r)} b_r, Q_C\Bigr\}\Bigr),
\end{equation}
i.e. graphs built from the objects $H_r$ and satisfying $\delta(G)=d$. 

The rest of the proof simply follows Lionni's proof of the uniqueness of a non-trivial large $N$ limit, \cite{PhDLionni}. First, to show that the large $N$ limit is non-trivial, it is enough to set $b_{r'}=0$ for all $r'\in R$ except one $b_r\neq 0$. Then the graphs of Lemma \ref{thm:Trees} satisfy $\delta(G) = d$ and they exist for all values of $b_r\geq 1$.

To prove uniqueness, we first notice that if $s_r>s_r^*$ for some $r\in R$, then there is no large $N$ limit anymore because $\delta(G)$ is unbounded for $G\in\mathbbm{G}(b_r,B_r)$. Then consider a graph in $\mathbbm{G}_{\max}(\{b_r,B_r\})$ therefore satisfying
\begin{equation}
\sum_{c=1}^d L_{0c}(G) + \sum_r b_r s_r^* = d,
\end{equation}
and introduce $\sigma_r = s_r^* - s_r$. It comes that
\begin{equation}
\delta(G) = \sum_{c=1}^d L_{0c}(G) + \sum_r b_r\, s_r = d - \sum_r \sigma_r\,b_r.
\end{equation}
which shows that at fixed $\delta(G)$ there is an upper bound on the number of bubbles $b_r$ as soon as $\sigma_r>0$.
\end{proof}

%
%

\section{Totally unbalanced GM bubbles} \label{sec:TotallyUnbalanced}

\paragraph*{Definition --} We say that a GM bubble $B$ is balanced if it is constructed from $C$-bidipole insertions satisfying $|C|=d/2$, i.e. the multiset $\mathcal{C}_B=\{C_i\}$ does not contain elements such that $|C_i|=d/2$, and we say it is unbalanced otherwise. Among the latter, we say that $B$ is {\bf totally unbalanced} if all bidipole insertions satisfy $|C_i|<d/2$.


\subsection{Expectations of totally unbalanced GM bubble polynomials in arbitrary tensor models}

In this section we show that the expectation of any totally unbalanced GM bubble polynomial:
\begin{itemize}
\item for the Gaussian distribution at large $N$, is given by a single Wick contraction, 
\item in arbitrary tensor models, behaves like in the Gaussian model where the covariance is the full 2-point function at large $N$. 
\end{itemize}

\paragraph*{Pairings on a bubble --} $\mathbbm{G}(B)$ is the set of Feynman graphs with a single bubble, whose vertices are thus paired by edges of color 0. We denote $\pi \in \mathbbm{G}(B)$ with $\{v, \pi(v)\}$ being a pair and $\pi^2 = \mathbbm{1}$ being the identity on the set of vertices of $B$.

$\mathbbm{G}_{\max}(B)$ consists in the pairings which maximize the number of bicolored cycles $\sum_{c=1}^d L_{0c}(G)$. In field theory language, it is the set of dominant Wick contractions at large $N$ for the expectation of $B(T, \bar{T})$ in the Gaussian distribution.

\begin{proposition}
If $B$ is a totally unbalanced GM bubble, then $\mathbbm{G}_{\max}(B)$ has a single element which is the canonical pairing $\pi_B$ defined in Section \ref{sec:GMB}.
\end{proposition}

Notice that the corresponding number of bicolored cycles is given by $b=1$ in Lemma \ref{thm:Trees}, $\sum_{c=1}^d L_{0c}(\pi_B) = d + \frac{d(V-2)}{2} - \sum_C |C|b_C$.

\begin{proof}
It is easy to prove by induction on the number of vertices. Indeed, $B$ has a $C$-bidipole with $|C|<d/2$, meaning that there are two vertices $v, \bar{v}$ connected by $|\widehat{C}| = d-|C|>d/2$ edges. If $v$ and $\bar{v}$ are not connected by an edge of color 0, we flip the two edges as follows,
\begin{equation} \label{flip}
\begin{array}{c} \includegraphics[scale=.4]{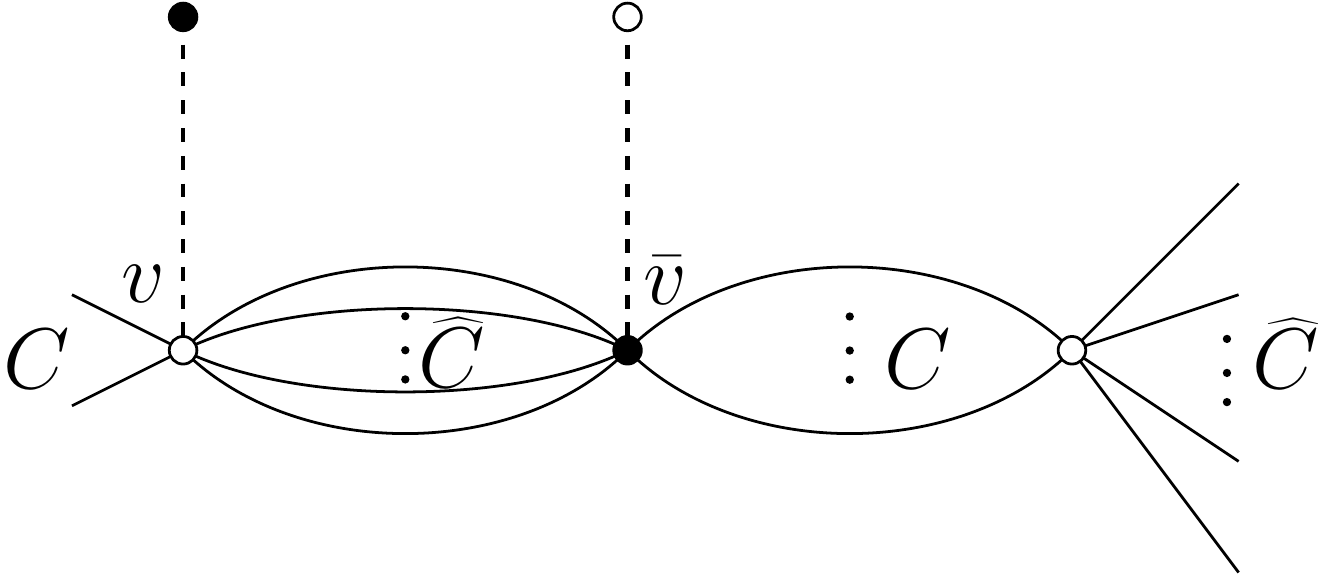} \end{array} \qquad \to \qquad \begin{array}{c} \includegraphics[scale=.4]{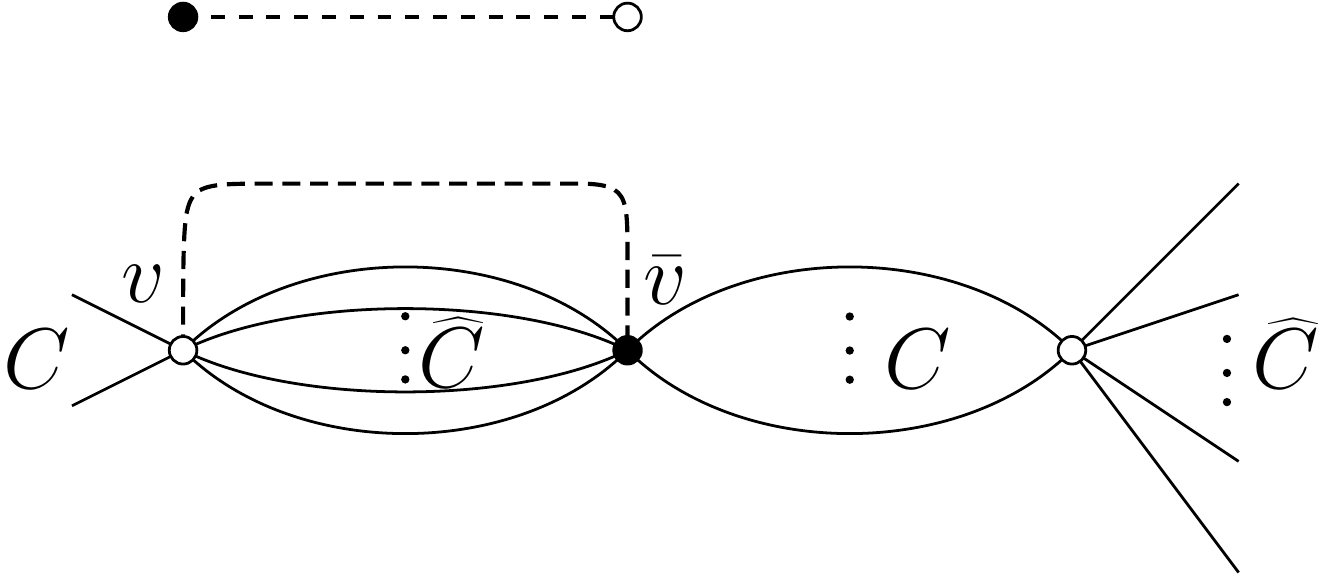} \end{array}
\end{equation}
This creates one new bicolored cycle of colors $\{0,c\}$ for each $c\in\widehat{C}$ so even if a bicolored cycle is lost for each $c\in C$ (and at most one can be lost), the total number $\sum_{c=1}^d L_{0c}$ increases.
\end{proof}

\paragraph*{The maximal 2-cut property --} For $G\in \mathbbm{G}(\{b_r, B_r\})$, we say that a bubble $B\subset G$ satisfies the {\bf 2-cut property} if there exists $\pi \in\mathbbm{G}(B)$ such that the vertices of a pair $\{v, \pi(v)\}$ are connected by an edge of color 0, or by a pair of edges of color 0 forming a 2-edge-cut,
\begin{equation}
\begin{array}{c} \includegraphics[scale=.4]{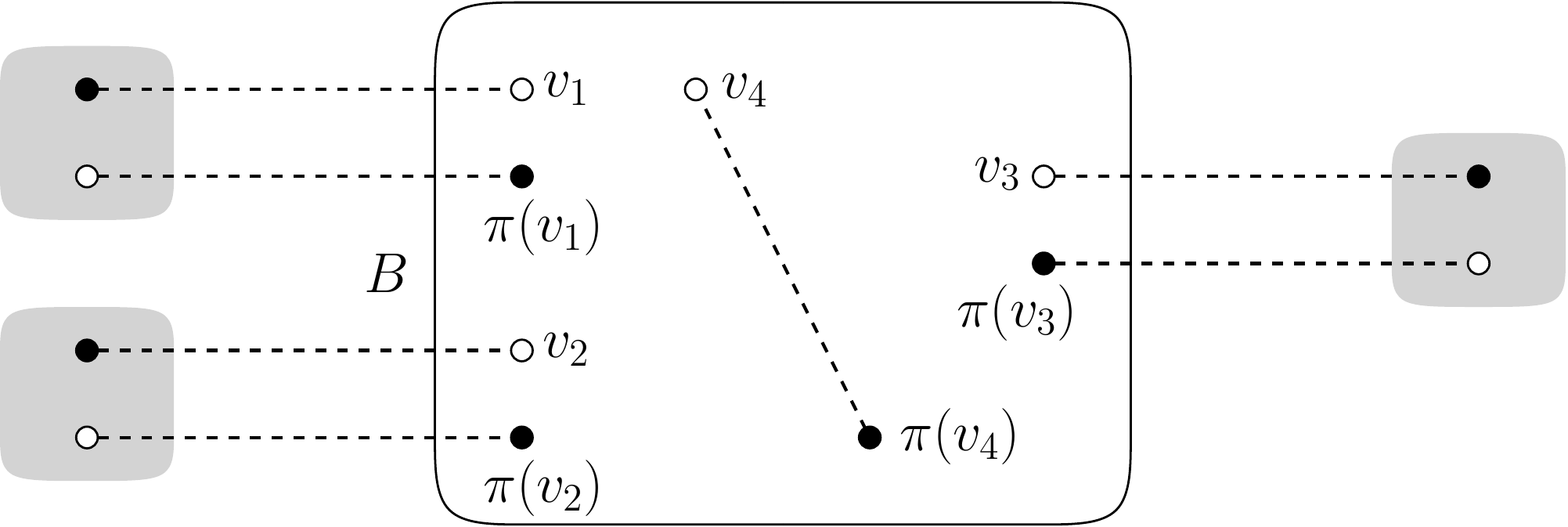} \end{array}
\end{equation}
In field theory language, it means that $B$ only has 2-point functions (including the bare propagator) incident to it.

We say that $B\subset G$ satisfies the {\bf maximal} 2-cut property if $\pi\in\mathbbm{G}_{\max}(B)$. In field theory language, it means that the 2-point functions are attached to $B$ as in the large $N$ limit of the Gaussian distribution.

\paragraph*{The set $\mathbbm{G}(\{b_r, B_r\}; B)$ --} Consider a set $\{B_r\}_{r\in R}$ of arbitrary bubbles and $B$ an additional bubble (which may be chosen among the set $\{B_r\}$). We denote $\mathbbm{G}(\{b_r, B_r\}; B)$ the set of connected graphs made of $b_r$ copies of $B_r$ for all $r\in R$ and a marked copy of $B$. They are the graphs which enter the Feynman expansion of the expectation of the bubble polynomial $B(T, \bar{T})$.

We denote $\mathbbm{G}_{\max}(\{b_r, B_r\}; B)$ the subset which maximizes the number of bicolored cycles $\sum_{c=1}^d L_{0c}(G)$.

\begin{theorem} \label{thm:2Cut}
If $B$ is a {\bf GM bubble}, then for all $G\in\mathbbm{G}_{\max}(\{b_r, B_r\}; B)$, the marked bubble $B$ satisfies the maximal 2-cut property.
\end{theorem}

In this theorem, the bubbles $B_r$ are arbitrary. Notice that we only refer to the number of bicolored cycles and not to any scaling coefficients. This formulation is in fact not tied to the existence of a large $N$ limit since it fixes the number of copies of each bubble.

\begin{proof}
The reasoning is the same as in the equivalent theorem of \cite{3D} for planar bubbles (instead of GM) in 3 dimensions. We therefore only sketch the proof and refer to \cite{3D} for details. It relies on an induction on the number of vertices of $B$. The theorem is obvious if $B$ has two vertices only.

Now consider $B$ with $V+2$ vertices and $G$ a Feynman graph containing $B$. If $B$ does not satisfy the maximal 2-cut property, we distinguish two cases.
\begin{itemize}
\item $B$ satisfies the 2-cut property with a pairing $\pi\not\in\mathbbm{G}_{\max}(B)$.
\item There are some edges of color 0 incident to $B$ which form a $k$-bond, for $k>2$ (a bond is a set of edges forming an edge-cut with no subset of edges forming an edge-cut).
\end{itemize}

If $v, \bar{v}$ are vertices of a $C$-bidipole connected by more than $|\widehat{C}| \geq d/2$ edges in $B$, either there is a single edge of color 0 between them and we can then contract the $C$-dipole by removing $v,\bar{v}$, the edges with colors in $\widehat{C}$ and the edge of color 0 between them and reconnecting the half-edges with colors in $C$. This produces a bubble $B'$ with $V$ vertices, which is still GM and a graph $G'\in\mathbbm{G}(\{b_r,B_r\}; B')$ in which $B'$ does not satisfy the maximal 2-cut property. We can then apply the induction hypothesis. The same reasoning applies if the edges of color 0 incident to $v$ and $\bar{v}$ form a 2-edge-cut.

If the edges of color 0 incident to $v$ and $\bar{v}$ are part of a $k$-bond for $k>2$, then we can flip them as in \eqref{flip} and observe that the number $\sum_{c=1}^d L_{0c}(G)$ increases, meaning that $G\not\in\mathbbm{G}_{\max}(\{b_r, B_r\};B)$.
\end{proof}

\subsection{Gaussian large \texorpdfstring{$N$}{N} limit for tensor models with totally unbalanced GM interactions} \label{sec:Gaussian}

\paragraph*{Large $N$ Gaussian model --} We say that the model with interactions $\{B_r(T, \bar{T})\}$ is {\bf Gaussian at large $N$} if for all bubbles $B$, for all graphs $G\in \mathbbm{G}_{\max}(\{b_r, B_r\};B)$, the marked bubble $B \subset G$ satisfies the maximal 2-cut property.

\begin{theorem*}[\ref{thm:Gaussian}]
If the bubbles $\{B_r\}_{r\in R}$ are GM, then the model is Gaussian at large $N$ and the large $N$ covariance $C(\{t_r\})$ satisfies
\begin{equation} \label{LargeN2point}
C(\{t_r\}) = 1 + \sum_{r'\in R} \frac{1}{2} t_{r'}\,V_{r'}\,C(\{t_r\})^{V_r/2}.
\end{equation}
In particular, the expectation of the polynomial associated to any bubble $B$ is
\begin{equation} \label{GaussianExpectation}
\langle B(T, \bar{T})\rangle = N^{\Omega(B)} \Bigl( \lvert \mathbbm{G}_{\max}(B)\rvert\, C(\{t_r\})^{V(B)/2} + \mathcal{O}(1/N)\Bigr)
\end{equation}
for some unknown $\Omega(B)$, while $V(B)$ is the number of vertices of $B$ and $\lvert \mathbbm{G}_{\max}(B)\rvert$ is the number of dominant Wick pairings at large $N$ in the Gaussian distribution.
\end{theorem*}

\begin{proof}
By definition, it means showing that for all $G\in\mathbbm{G}_{\max}(\{b_r, B_r\}; B)$ with GM bubbles $\{B_r\}$ and an arbitrary $B$, the marked bubble $B$ satisfies the maximal 2-cut property.

The reasoning is again the same as in the equivalent theorem of \cite{3D} with planar bubbles instead of GM bubbles. It relies on the fact that for $G\in\mathbbm{G}_{\max}(\{b_r, B_r\}; B)$, all copies of every $B_r \subset G$ satisfy the maximal 2-cut property, and this forces $B$, even if not GM, to satisfy the maximal 2-cut property too. We refer to \cite{3D} for details.

The equation \eqref{LargeN2point} can be obtained either directly from the combinatorics, or from any Schwinger-Dyson equation, which can be turned into an equation on $C(\{t_r\})$ using the expectations \eqref{GaussianExpectation}.
\end{proof}

\section{Matrix models for GM tensorial interactions} \label{sec:MatrixModel}

\subsection{Plane tree representation of GM bubbles}

The set $\mathbbm{G}_V$ clearly consists in trees, since dashed lines connecting quartic bubbles are 1-edge-cuts. To make it more explicit, we use the intermediate field representation as in Section \ref{sec:QuarticLargeN}. It will provide a bijection between $\mathbbm{G}_V$ and a set $\mathbbm{T}_V$ of plane trees with $V/2$ vertices, Proposition \ref{thm:IntField} below.

\paragraph*{Intermediate field representation of GM bubbles --} Since $\mathbbm{G}_V$ consists in quartic bubbles connected by dashed lines, we can try and apply the intermediate field representation. The pairing of the vertices is as before, $\{v, \bar{v}\}$ forming a pair when they are connected by the colors of $\widehat{C}$. Each bubble $Q_{C}$ is represented as an edge decorated by the color set $C$, see Equation \eqref{IFEdge}.

The vertices of the intermediate field representation were presented in Section \ref{sec:QuarticLargeN} as the cycles alternating canonical pairs of vertices of quartic bubbles and dashed lines, see \eqref{IFClosedVertex}. However, since dashed lines are edge-cuts in $\mathbbm{G}_V$, there are no such cycles. Instead, there are paths of the same type: Start from a vertex $v_1$ with no incident dashed line. Jump to its partner $\bar{v}_1$. If it is incident to a dashed line, follow the latter to a vertex $v_2$, then to $\bar{v}_2$ and so on, until one arrives at a vertex with no incident dashed line. Such a path thus goes $(v_1, \bar{v}_1, v_2, \bar{v}_2, \dotsc)$. In the intermediate field representation, it is represented as a vertex whose incident half-edges correspond to the pairs $\{v_i, \bar{v}_i\}$. Since the latter are ordered along the path, the half-edges incident to the vertex in the intermediate field representation are also ordered, similarly to the case of combinatorial maps.

This is however not quite a cyclic order, since the path has a start and an end. This information is stored by marking the appropriate corner, here in gray,
\begin{equation} \label{IFVertex}
\begin{array}{c} \includegraphics[scale=.5]{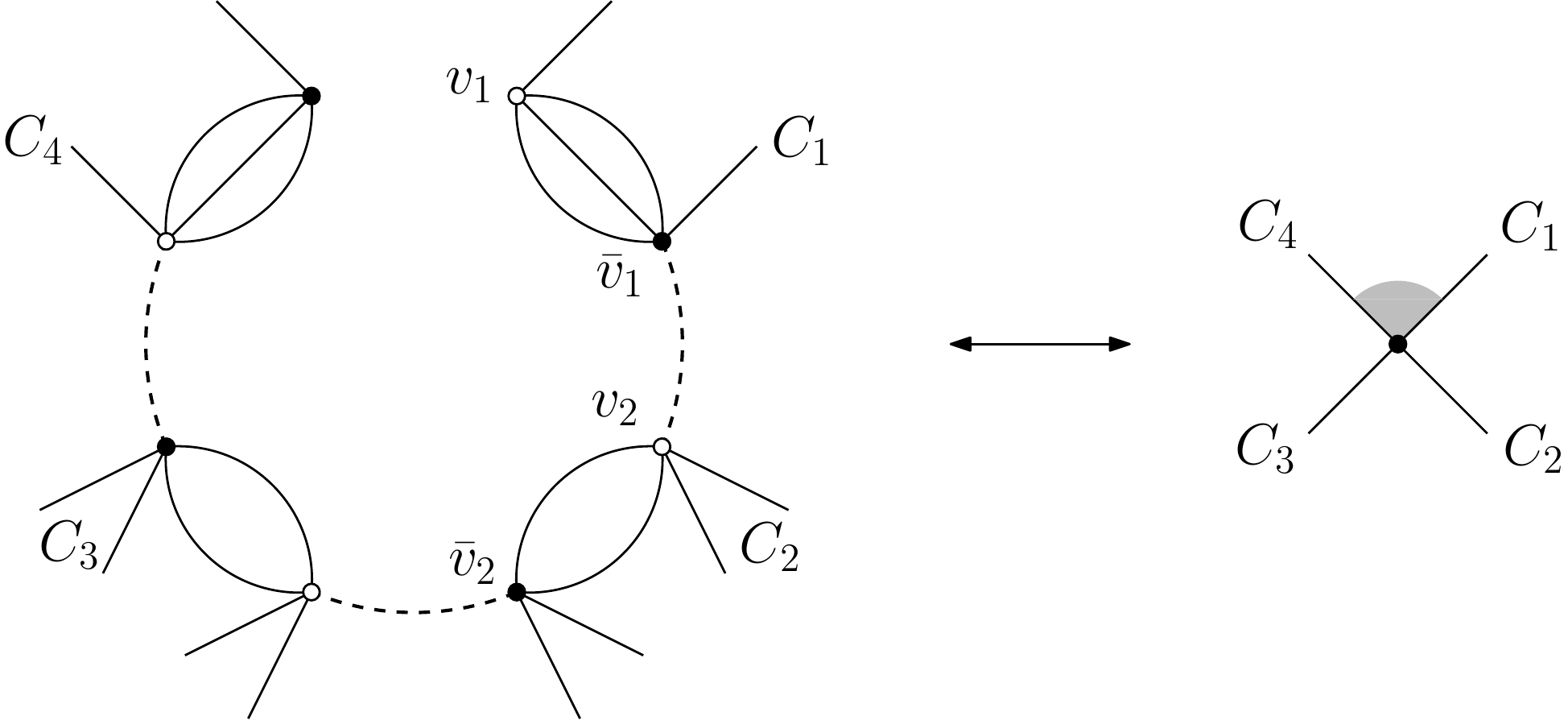} \end{array}
\end{equation}

Through the intermediate field representation, quartic bubbles connected along dashed lines are turned into combinatorial maps with a color set $C_e$ for every edge $e$ and possibly one marked corner per vertex. Denote $\mathbbm{M}_V$ the set of such maps with $V/2$ marked corners.

\paragraph*{Bijection with plane trees --} What we want to do now is characterize the set $\mathbbm{G}_V$ using the intermediate field representation. 

\begin{proposition} \label{thm:IntField}
Let $\mathbbm{T}_V \subset \mathbbm{M}_V$ be the set of plane trees with $V/2$ vertices and exactly one marked corner per vertex. Then there is a bijection $J_V\,: \mathbbm{G}_V \to \mathbbm{T}_V$, and we denote
\begin{equation}
J\ :\ \bigcup_V \mathbbm{G}_V \, \to\, \bigcup_V \mathbbm{T}_V
\end{equation}
such that $J$ restricted to $\mathbbm{G}_V$ is $J_V$. It is such that bubbles are mapped to edges, canonical pairs of vertices to half-edges and dashed lines to unmarked corners\footnote{We can check that there is the expected number of such corners. Indeed, the total number of corners is the number of half-edges, i.e. $2(V/2-1) = V-2$. Then one has to remove the marked corner at each vertex, yielding $V-2-V/2 = V/2 -2$.}.
\end{proposition}

\begin{proof}
This is proved by translating the recursion \eqref{BubbleInduction} with the intermediate field. First observe that the plane tree with a single edge and two vertices and marked corners is
\begin{equation}
\begin{array}{c} \includegraphics[scale=.4]{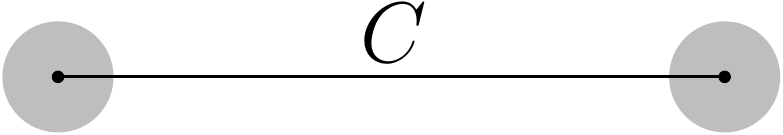} \end{array} = J\left(\begin{array}{c} \includegraphics[scale=.4]{4VertexBubble.pdf} \end{array}\right)
\end{equation}
Then, consider an object $\mathcal{T} \in J(\mathbbm{G}_V)$. It satisfies the following decomposition
\begin{equation}
\mathcal{T} = J(G) = J \left(\begin{array}{c} \includegraphics[scale=.4]{4VertexBubbleInduction.pdf} \end{array}\right) = \begin{array}{c} \includegraphics[scale=.4]{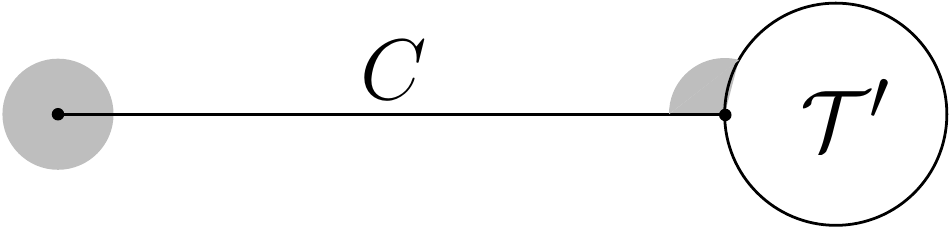} \end{array}.
\end{equation}
with $\mathcal{T}'\in J(\mathbbm{G}_{V-2})$. This defines the family of sets $\mathbbm{T}_V$ recursively.
\end{proof}

\subsection{The partition function as a matrix model}

We can state Theorem \ref{thm:MatrixModel} precisely.

\begin{theorem*}[\ref{thm:MatrixModel}]
Let $\mathcal{T}\in\mathbbm{T}_V$ such that $\partial J^{-1} \mathcal{T} = B$ is a GM bubble, with each edge $e\in\mathcal{T}$ decorated by a color set $C_e$. Associate to each half-edge $h$ a complex matrix $X_h$, such that if an edge $e=h_1 h_2$ is made of $h_1$ and $h_2$ then $X_{h_1} = X_e$ and $X_{h_2} = X^\dagger_e$. The matrix $X_e$ acts on the indices with colors in $C_e$,
\begin{equation}
X_e \in V_{C_e}^*\otimes V_{C_e} = \bigotimes_{c\in C_e} V^*_c\otimes V_c
\end{equation}
and it is extended to $\tilde{X}_e\in E_d^*\otimes E_d$ by a trivial action on $\widehat{C}_e$,
\begin{equation}
\tilde{X}_e = X_e \bigotimes_{c\in \widehat{C}_e} \mathbbm{1}_{V_c}.
\end{equation}
Set $\epsilon_h=1$ except for one half-edge where it is $-1$. Then
\begin{equation} \label{NewMatrixModel}
Z_N(t_B) = \int \prod_{e\in\mathcal{T}} dX_e dX_e^\dagger\ \exp -\sum_{e\in\mathcal{T}} \tr_{V_{C_e}} (X_e X^\dagger_e) - \tr_{E_d}\ln \Biggl(\mathbbm{1}_{E_d} - \sum_{v\in\mathcal{T}} \prod_{h_v}^{\substack{\text{counter-}\\\text{-clockwise}}} (N^{s_B} t_B)^{\frac{1}{V-2}} \epsilon_h\tilde{X}_{h_v} \Biggr) 
\end{equation}
where the sum in the logarithm is over the vertices $v\in\mathcal{T}$, and the product over the half-edges $h_v$ incident to $v$.
\end{theorem*}

The elements $X_{h_v}$ around a given vertex $v$ are non-commutative and the order of the product is the counter-clockwise order, say starting from the marked corner.

To prove this theorem, we will use the Hubbard-Stratonovich transformation to rewrite the partition function on $G\in\mathbbm{G}_V$ such that $B=\partial G$, then use the bijection of Proposition \ref{thm:IntField} and rewrite it again on $\mathcal{T}=J(G)$ this time, then use Hubbard-Stratonovich on each edge of $\mathcal{T}$ and finally integrate all tensors.

Let us package all Hubbard-Stratonovich transformations into two lemmas.

\paragraph*{Tensor assignments --} 
Let $B$ be a bubble with $V$ vertices and a graph $G\in\mathbbm{G}_V$ such that $\partial G = B$. We use the notation $v$ for any vertex of $G$ with no incident dashed line, and $\ell$ for a dashed line and $u$ for an arbitrary vertex of $G$. We decorate $G$ with the following tensor assignments.
\begin{itemize}
\item Assign a tensor $T_v$ to each white vertex with no incident dashed line, and $\bar{T}_v$ for black vertices with no incident dashed line. 
\item To each dashed line $\ell = 1, \dotsc, V/2-2$, associate the pair of tensors $(T_\ell, \bar{T}_\ell)$, $T_\ell$ to the white end and $\bar{T}_\ell$ to the black end.
\item Denote generally $T_u$ the tensor assigned on the vertex $u\in G$ (which can be a $T_v$ or a $T_\ell$) and $\bar{T}_{u}$ on a black vertex, according to the rules above.
\item Label the quartic bubbles $j=1, \dotsc, V/2-1$ and $\{u_j, \bar{u}_j\}$ and $\{u'_j, \bar{u}'_j\}$ the two canonical pairs of vertices of $Q_{C_j}$, such that $u_j$ and $u'_j$ are white and $\bar{u}_j = \pi_{Q_{C_j}}(u_j)$ and $\bar{u}'_j = \pi_{Q_{C_j}}(u'_j)$ are black.
\end{itemize}

\begin{lemma} \label{thm:Zint1}
Let $G\in\mathbbm{G}_V$ such that $\partial G=B$. Then, with the notations defined above,
\begin{equation}
\int \prod_{\ell=1}^{V/2-2} dT_\ell d\bar{T}_\ell\ \exp - \sum_{\ell=1}^{V/2-2} (T_\ell |T_\ell) + t^{\frac{2}{V-2}} \sum_{j=1}^{V/2-1} \epsilon_j Q_{C_j}(T_{u_j}, \bar{T}_{\bar{u}_j}; T_{u'_j}, \bar{T}_{\bar{u}'_j}) = \exp \left(\prod_{j=1}^{V/2-1} \epsilon_j\right)\, t\, B(\{T_v\})
\end{equation}
where every $\epsilon_j$ is a complex number.
\end{lemma}

\begin{proof}
We proceed by induction on the number of vertices of $B$. When $B$ is quartic, we simply have $G=B$, there is no dashed line in $G$ so no integral on the left hand side. On both sides only the term $j=1$ exists and the lemma is then trivial. Then assume that the theorem holds for all bubbles with $V-2$ vertices and fewer, and consider $B$ with $V$ vertices and $\partial G=B$.

Consider the decomposition \eqref{BubbleInduction} with the following notations
\begin{equation}
\begin{array}{c} \includegraphics[scale=.4]{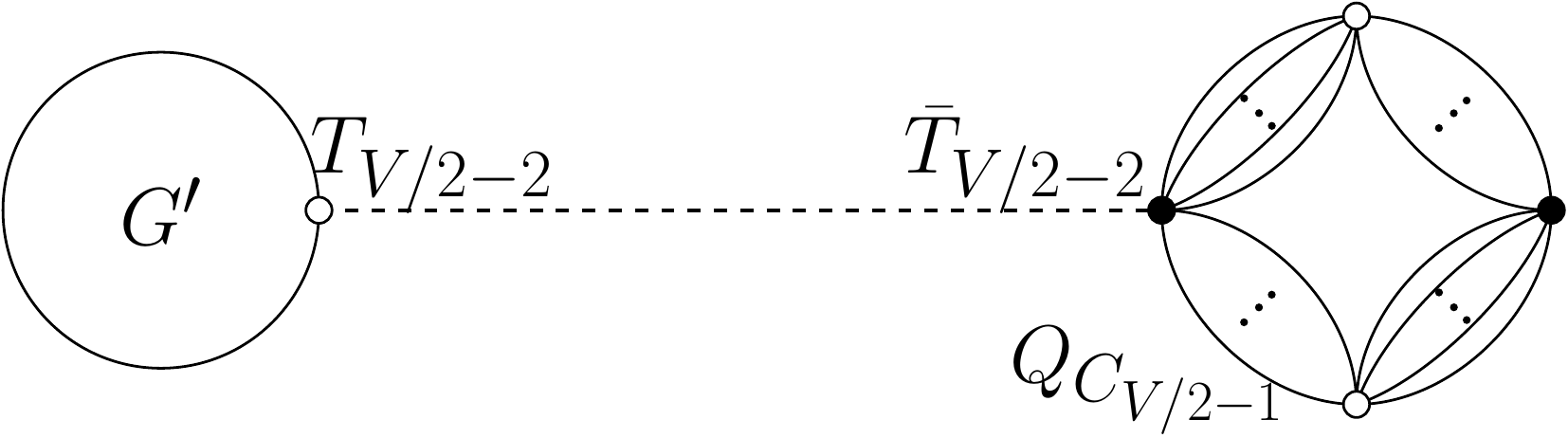} \end{array}
\end{equation}
i.e. the quartic bubble has label $j=V/2-1$ and the dashed line connecting $G'$ to it has label $\ell=V/2-2$. From the induction hypothesis, we first get
\begin{equation}
\int \prod_{\ell=1}^{V/2-3} dT_\ell d\bar{T}_\ell\ \exp - \sum_{\ell=1}^{V/2-3} (T_\ell |T_\ell) + t^{\frac{2}{V-2}} \sum_{j=1}^{V/2-2} \epsilon_j Q_{C_j}(T_{u_j}, \bar{T}_{\bar{u}_j}; T_{u'_j}, \bar{T}_{\bar{u}'_j}) = \exp \left(\prod_{j=1}^{V/2-3} \epsilon_j\right)\,t^{\frac{V-4}{V-2}}\, B'(\{T_v\})
\end{equation}
where $B' = \partial G'$ is a GM bubble with $V-2$ vertices.

It remains to basically find an analytical way to glue $B'$ to $Q_{C_{V/2-2}}$. This is the Hubbard-Stratonovich transformation: for all complex numbers $Z_1, Z_2$
\begin{equation} \label{HS}
\int dz d\bar{z}\ \exp \bigl(-z \bar{z} - z Z_1 + \bar{z} Z_2\bigr) = \exp - Z_1 Z_2.
\end{equation}
This directly implies that
\begin{multline} \label{HSTensor}
\int dT_{V/2-2} d\bar{T}_{V/2-2}\ \exp -(T_{V/2-2}|T_{V/2-2}) + t^{\frac{2}{V-2}}\,\epsilon_{V/2-1}\, Q_{C_{V/2-1}}(\{T_u\}) + t^{\frac{V-4}{V-2}}\,\left(\prod_{j=1}^{V/2-3} \epsilon_j\right)\, B'(\{T_u\}) \\
= \exp \left(\prod_{j=1}^{V/2-2} \epsilon_j\right)\,t\,B(T, \bar{T}).
\end{multline}
as desired.
\end{proof}

\begin{lemma} \label{thm:QuarticBubbleHS}
Consider a quartic bubble polynomial $Q_C(A_1, \bar{A}_2; A_3, \bar{A}_4)$, then
\begin{equation}
\exp \epsilon\, t\,Q_C(A_1, \bar{A}_2; A_3, \bar{A}_4) = \int dX\,dX^\dagger\ \exp \tr_{V_C}\Bigl(- XX^\dagger + \sqrt{t}\, (A_2|A_1)_{\widehat{C}}\,X + \epsilon\,\sqrt{t}\, (A_4|A_3)_{\widehat{C}}\,X^\dagger \Bigr)
\end{equation}
for $\epsilon=\pm 1$ and the integral is over matrices $X\in V^*_C\otimes V_C$.
\end{lemma}

The Hubbard-Stratonovich transformation cuts the quartic interaction in two. There are two ways to do so, as in Equation \eqref{CutQuartic}, which differ by the exchange $C\leftrightarrow \widehat{C}$. The best way is the one that will integrate as many degrees of freedom as possible, i.e. leaves us with traces of matrices of dimensions as small as possible. We thus cut along the degrees of freedom from the color sets $C$ (which we recall satisfy $|C|\leq d/2$).

\begin{proof}
Simply write the quartic polynomial as in \eqref{QuarticMatrix}, $Q_C(A_1, \bar{A}_2; A_3, \bar{A}_4) = \tr_{V_C} \Bigl((A_2|A_1)_{\widehat{C}} (A_4|A_3)_{\widehat{C}}\Bigr)$. Since it is a trace over $V_C$, we can then apply the Hubbard-Stratonovich transformation \eqref{HS} such that the intermediate field is a matrix $X\in V^*_C\otimes V_C$.
\end{proof}

\begin{proof}[Proof of Theorem \ref{thm:MatrixModel}]
We start by applying Lemma \ref{thm:Zint1} with $T_v = T$ for all white vertices with no dashed line and $\bar{T}_v=\bar{T}$ for all black vertices with no dashed line, $t = N^{s_B} t_B$ and $\epsilon_j=1$ except one which is $-1$, and integrate over $T, \bar{T}$. This leads to the following expression for the partition function
\begin{equation} \label{Zint1}
Z_N(t_B) = \int dT d\bar{T} \prod_{\ell=1}^{V/2-2} dT_\ell d\bar{T}_\ell\ \exp - (T|T) - \sum_{\ell=1}^{V/2-2} (T_\ell |T_\ell) + \bigl(N^{s_B} t_B\bigr)^{\frac{2}{V-2}} \sum_{j=1}^{V/2-1} \epsilon_j Q_{C_j}(T_{u_j}, \bar{T}_{\bar{u}_j}; T_{u'_j}, \bar{T}_{\bar{u}'_j}).
\end{equation}

We now switch to the plane tree representation of GM bubbles to rewrite \eqref{Zint1}. For a graph $G\in \mathbbm{G}_V$, the pairs of tensors $\{T_\ell, \bar{T}_\ell\}_{\ell=1, \dotsc, V/2-2}$ are associated to the dashed lines of $G$: $T_\ell$ to one half of $\ell$ and $\bar{T}_{\ell}$ to the other half. In $\mathcal{T} = J(G)$, a dashed line $\ell$ corresponds to an unmarked corner $\kappa=J(\ell)$ of the tree and a half dashed line is in $\mathcal{T}$ a pair of a corner and an edge. If $\kappa= J(\ell)$ is a corner incident to the edges $e,e'$ we denote $T_\kappa$ the tensor associated to, say, the pair $\{\kappa,e\}$ and $\bar{T}_\kappa$ the tensor associated to $\{\kappa,e'\}$. This is as follows
\begin{equation}
\begin{array}{c} \includegraphics[scale=.45]{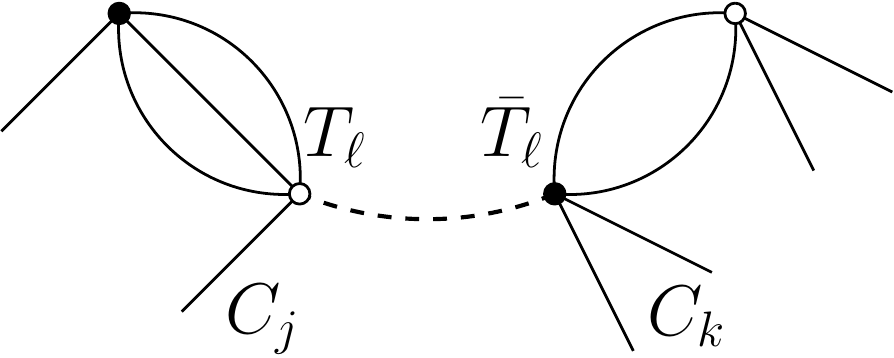} \end{array} \qquad \to \qquad \begin{array}{c} \includegraphics[scale=.45]{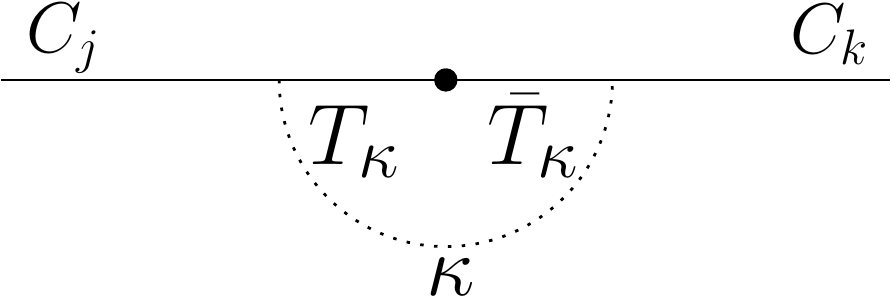} \end{array}
\end{equation}
where the dotted arch represents the corner $\kappa$. For convenience, one can use the counter-clockwise order around the vertex to orient the corner and decide that one encounters $T_\kappa$ first, then $\bar{T}_\kappa$.


In addition to the labels $\kappa=1, \dotsc, V/2-2$ of unmarked corners of $\mathcal{T}$, we label as a whole the marked corners by $\kappa=0$ and set $T_0=T$, $\bar{T}_0=\bar{T}$.


In $\mathcal{T}$, the quartic bubbles correspond to the edges. The quartic interactions are thus associated to the edges of $\mathcal{T}$. An edge $e$ is incident to four corners $\kappa_{e1}, \kappa_{e2}, \kappa_{e3}, \kappa_{e4}$ and the interaction is
\begin{equation}
\begin{array}{c} \includegraphics[scale=.5]{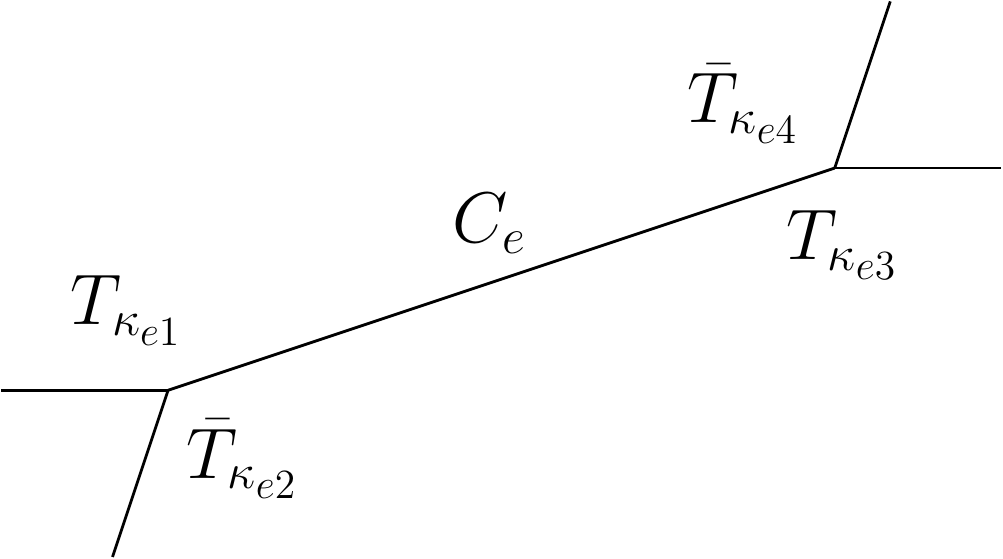} \end{array} \qquad \to \qquad 
I_e(\{T_{\kappa}, \bar{T}_{\kappa}\}) = \epsilon_e\ Q_{C_e}(T_{\kappa_{e1}}, \bar{T}_{\kappa_{e2}}; T_{\kappa_{e3}}, \bar{T}_{\kappa_{e4}})
\end{equation}
where $\epsilon_e=1$ except for one edge of $\mathcal{T}$.

The partition function \eqref{Zint1} thus becomes
\begin{equation} \label{Zint2}
Z_N(t_B) = \int \prod_{\kappa=0}^{V/2-2} dT_\kappa d\bar{T}_\kappa\ \exp - \sum_{\kappa=0}^{V/2-2} (T_\kappa |T_\kappa) + \bigl(N^{s_B} t_B\bigr)^{\frac{2}{V-2}} \sum_{e\in\mathcal{T}} I_e(\{T_{\kappa}, \bar{T}_{\kappa}\})
\end{equation}
We apply Lemma \ref{thm:QuarticBubbleHS} to each quartic bubble. For an edge $e\in \mathcal{T}$ with color set $C_e$, we obtain
\begin{multline}
e^{(N^{s_B} t_B)^{\frac{2}{V-2}} I_e(\{T_{\kappa}, \bar{T}_{\kappa}\})} \\
= \int dX_e dX^\dagger_e\ \exp \tr_{V_{C_e}} \Bigl(-X_e X^\dagger_e + (N^{s_B} t_B)^{\frac{1}{V-2}} (T_{\kappa_{e2}} | T_{\kappa_{e1}})_{\widehat{C}_e} X_e +  (N^{s_B} t_B)^{\frac{1}{V-2}} \epsilon_e\,X^\dagger_e (T_{\kappa_{e4}} | T_{\kappa_{e3}})_{\widehat{C}_e}\Bigr)
\end{multline}
with $X_e \in V^*_{C_e}\otimes V_{C_e}$ a complex matrix.

Clearly $X_e$ is naturally associated to one end of $e$ and $X^\dagger_e$ to the other. We can therefore associate a matrix $X_h$ to each half-edge $h$, with the constraint
\begin{equation}
X^\dagger_{h_1} = X_{h_2} \qquad \text{when $e = \{h_1, h_2\}$}.
\end{equation}
The edges can be oriented to distinguish between $X_e=X_{h_1}$ and $X_e^\dagger=X_{h_2}$ and the other way around, but there is no need to do so. Split $\epsilon_e$ onto the half-edges as $\epsilon_e = \epsilon_{h_1} \epsilon_{h_2}$. Since all $\epsilon_e$ are equal to 1 except one which is $-1$, we can also choose all $\epsilon_h$ equal to 1 except one which is $-1$.

As a result of the Hubbard-Stratonovich transformations, the action becomes quadratic in the tensors $\{T_u, \bar{T}_u\}$. We now focus on the corresponding quadratic form, so as to integrate the tensors and be left with a matrix model on the matrices $\{X_e, X^\dagger_e\}$. 
Let us denote $\kappa_h$ and $\kappa'_h$ the two corners incident to the half-edge $h$. Then
\begin{equation}
Z_N(t_B) = \int \prod_{e\in\mathcal{T}} dX_e dX^\dagger_e\ \prod_{\kappa=0}^{V/2-2} dT_\kappa d\bar{T}_\kappa\ e^{ -\sum_{\kappa=0}^{V/2-2} (T_\kappa|T_\kappa) -\sum_{e} \tr_{V_{C_e}} X_e X^\dagger_e + (N^{s_B} t_B)^{\frac{1}{V-2}} \sum_{h} \epsilon_h \tr_{V_{C_e}} (T_{\kappa'_h}|T_{\kappa_h}) X_h}
\end{equation}

Before we can write the quadratic form matricially, we need to extend the two-tensor contraction of \eqref{TwoTensorContraction} to the presence of operators. We recall that $E_d = \bigotimes_{c=1}^d V_c$. For $K\in E_d^*\otimes E_d$ we denote
\begin{equation}
(T_2 | K| T_1) = \sum (\bar{T}_2)_{a_1 \dotsb a_d} K_{ a_1 \dotsb a_d; b_1 \dotsb b_d} (T_1)_{b_1 \dotsb b_d} = \begin{array}{c} \includegraphics[scale=.4]{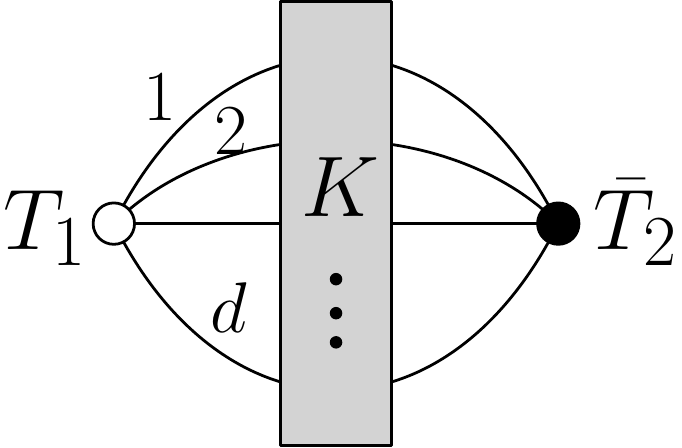} \end{array}
\end{equation}
The matrix $X_h$ is an element of $V^*_{C_e}\otimes V_{C_e}$ and it is canonically extended to $\tilde{X}_h\in E_d^*\otimes E_d$ by completing its action on the complementary colors $\widehat{C}_e$ to be the identity. Therefore $\tilde{X}_h$ can be used as an operator like $K$ above.

The quadratic form can thus be written
\begin{equation}
\sum_{\kappa, \kappa'=0}^{V/2-2} (T_\kappa |\bigl[K_{\mathcal{T}}\bigr]_{\kappa \kappa'} |T_{\kappa'}) = (T|\bigl[K_{\mathcal{T}}\bigr]_{00}|T) + \sum_{\kappa\geq 1} (T|\bigl[K_{\mathcal{T}}\bigr]_{0\kappa}|T_\kappa) + \sum_{\kappa\geq 1} (T_\kappa|\bigl[K_{\mathcal{T}}\bigr]_{\kappa 0}|T) + \sum_{\kappa, \kappa'\geq 1} (T_\kappa|\bigl[K_{\mathcal{T}}\bigr]_{\kappa \kappa'}| T_{\kappa'})
\end{equation}
where each component $\bigl[K_{\mathcal{T}}\bigr]_{\kappa \kappa'} \in E_d^*\otimes E_d$ is a linear operator on the space of tensors. Here the index $\kappa=0$ corresponds to marked corners (since $T_0=T$ only sits on marked corners) and $\kappa\geq 1$ to the $V/2-2$ unmarked corners.

The $\kappa\kappa'$-components, for unmarked corners $\kappa, \kappa'\geq 1$ and $\kappa\neq \kappa'$, are non-vanishing if and only if the two corners $\kappa$ and $\kappa'$ are adjacent in $\mathcal{T}$, i.e. if they meet along a half-edge $h$, and $\kappa$ is right before $\kappa'$ for the counter-clockwise order. Then
\begin{equation}
\bigl[K_{\mathcal{T}}\bigr]_{\kappa \kappa'} = \begin{cases} - (N^{s_B} t_B)^{\frac{1}{V-2}}\ \epsilon_h\, \tilde{X}_h \qquad & \text{if $\begin{array}{c}\includegraphics[scale=.4]{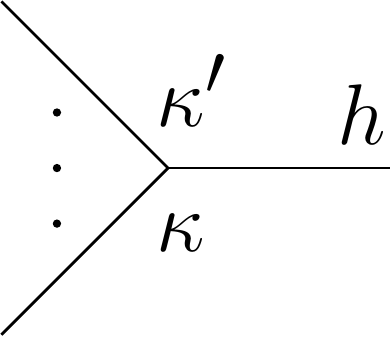}\end{array}$}\\
\mathbbm{1}_{E_d} & \text{if $\kappa=\kappa'$}\\
0 & \text{else}
\end{cases}
\end{equation}
If $\kappa$ or $\kappa'$ is a marked corner, i.e. $\kappa=0$ or $\kappa'=0$, the structure is similar
\begin{equation}
\bigl[K_{\mathcal{T}}\bigr]_{\kappa 0} = \begin{cases} - (N^{s_B} t_B)^{\frac{1}{V-2}}\ \epsilon_h\,\tilde{X}_h & \text{if $\begin{array}{c}\includegraphics[scale=.4]{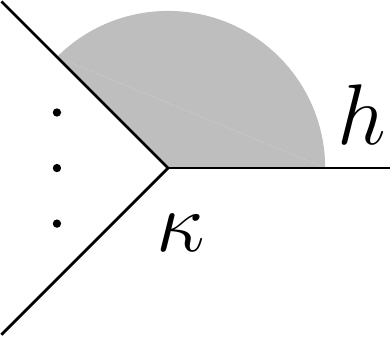}\end{array}$}\\
0 & \text{else}
\end{cases}
\hspace{2cm}
\bigl[K_{\mathcal{T}}\bigr]_{0 \kappa} = \begin{cases} -(N^{s_B} t_B)^{\frac{1}{V-2}}\ \epsilon_h\,\tilde{X}_h & \text{if $\begin{array}{c}\includegraphics[scale=.4]{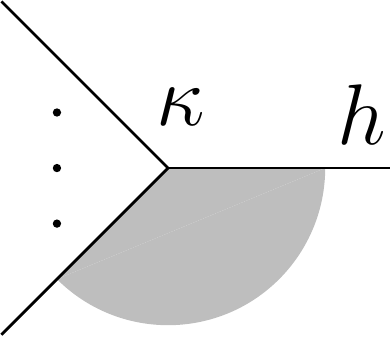}\end{array}$}\\
0 & \text{else}
\end{cases}
\end{equation}
Finally, the component $00$ picks up contributions from marked corners adjacent to themselves, in the sense that they share a half-edge. This incorporates the leaves of $\mathcal{T}$ (a leaf having a single corner, which has to be marked),
\begin{equation}
\bigl[K_{\mathcal{T}}\bigr]_{0 0} = \mathbbm{1}_{E_d} - (N^{s_B} t_B)^{\frac{1}{V-2}}\sum_{\text{leaves}} \epsilon_h\,\tilde{X}_h.
\end{equation}

It remains to evaluate the determinant of $K_{\mathcal{T}}$, as a function of the matrices $X_h$. It is mainly determined by the specific block structure of $K_{\mathcal{T}}$, so we can work out a slightly more general case given in Lemma \ref{thm:Determinant} below.

We apply this result to the quadratic form $K_{\mathcal{T}}$,
\begin{equation}
\det K_{\mathcal{T}} = \det \Bigl(\mathbbm{1}_{E_d} - \sum_{v\in\mathcal{T}} \prod_{h_v}^{\substack{\text{counter-}\\\text{-clockwise}}} (N^{s_B} t_B)^{\frac{1}{V-2}}\ \epsilon_h\tilde{X}_{h_v} \Bigr)
\end{equation}
Up to constants, this gives
\begin{equation}
Z_N(t_B) = \int \prod_{e\in \mathcal{T}} dX_e dX^\dagger_e\ \exp -\sum_{e\in\mathcal{T}} \tr_{V_{C_e}} (X_e X^\dagger_e) - \tr_{E_d}\ln \Bigl(\mathbbm{1}_{E_d} - \sum_{v\in\mathcal{T}} \prod_{h_v}^{\substack{\text{counter-}\\\text{-clockwise}}} (N^{s_B} t_B)^{\frac{1}{V-2}}\,\epsilon_h\,\tilde{X}_{h_v} \Bigr)
\end{equation}
for any tree $\mathcal{T}\in \mathbbm{T}_V$ corresponding to $B$.
\end{proof}

Here is the lemma which gives the determinant of quadratic forms of the same type as $K$ as a determinant over $E_d^*\otimes E_d$.

\begin{lemma} \label{thm:Determinant}
Let $P_{\mathcal{T}}$ be a matrix similar to $K_{\mathcal{T}}$ except that all $-(N^{s_B} t_B)^{\frac{1}{V-2}}\epsilon_h \tilde{X}_h$ are replaced with arbitrary $Y_h\in E_d^*\otimes E_d$. Moreover, $\bigl[P_{\mathcal{T}}\bigr]_{0 0} = Y_0 + \sum_{\text{leaves}} Y_h$ for $Y_0\in E_d^*\otimes E_d$. Then
\begin{equation}
\det P_{\mathcal{T}} = \det \Bigl(Y_0 + \sum_{v\in\mathcal{T}} \prod_{h_v}^{\text{counter-clockwise}} Y_{h_v} \Bigr)
\end{equation}
\end{lemma}

\begin{proof}
We proceed by induction. The theorem is true when $\mathcal{T}$ has a single edge with two marked corners. We then assume that it holds for all $\mathcal{T}'\in\mathbbm{T}_{V'}$ for $V'< V$.

The key is that an unmarked corner $\kappa$ is incident to only two edges, one on its left and one on its right, using the clockwise order. This means that on the row $\kappa$, the matrix $P_{\mathcal{T}}$ has a single non-diagonal, non-trivial element, and similarly on the column $\kappa$.

\begin{figure}
\captionsetup[subfigure]{singlelinecheck=false}
\centering
\subcaptionbox{\label{fig:Leaf}In $\mathcal{T}$, the half-edges $h_1, h_2$ connect a leaf to its parent vertex with a marked corner on its left.}{\includegraphics[scale=.5]{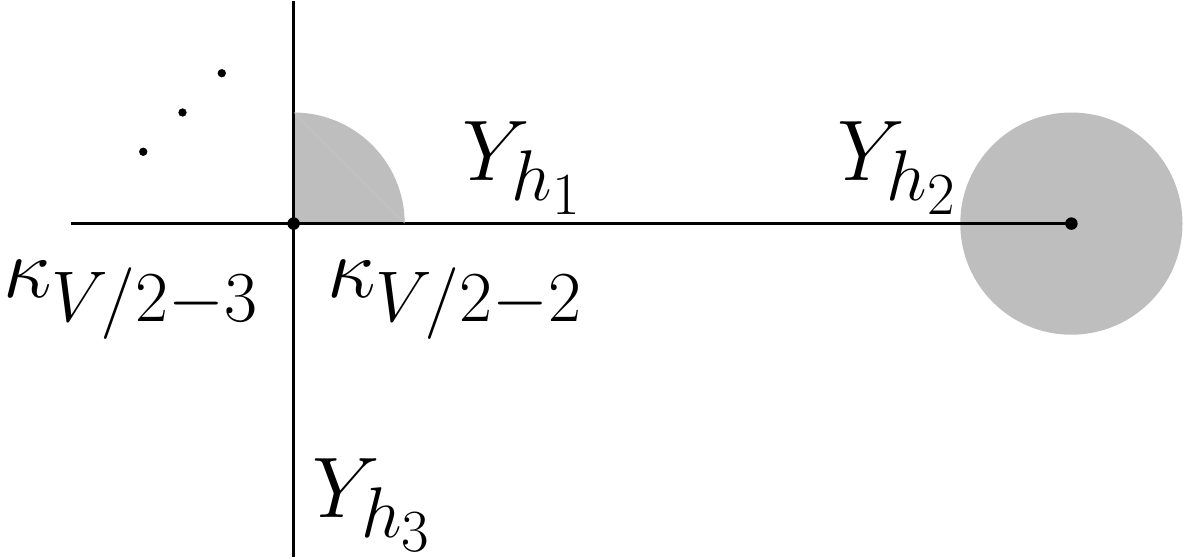}}
\hspace{2cm}
\subcaptionbox{\label{fig:NewTree}The new tree $\mathcal{T}'$ is obtained by removing the half-edges $h_1, h_2$ and re-assigning the matrix $Y_{h_3} Y_{h_1}$ to $h_3$.}{\includegraphics[scale=.5]{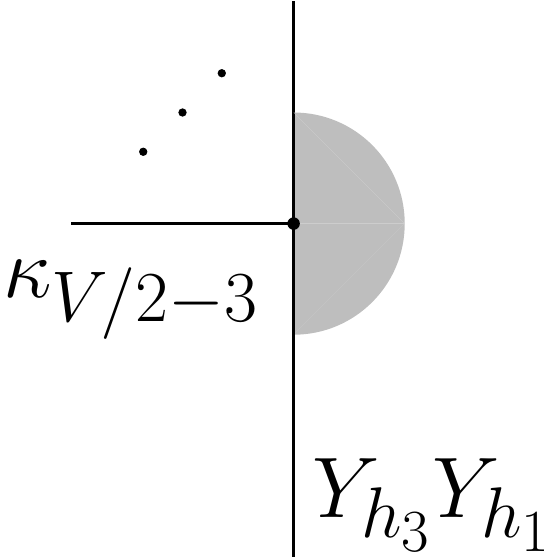}}
\caption{\label{fig:TreeInduction} The original tree $\mathcal{T}$ on the left and its transformation into $\mathcal{T}'$ on the right with one less vertex and one less edge.}
\end{figure}

Consider in $\mathcal{T}$ a marked corner, whose incident edge on the right (using the clockwise order) is connected to a leaf, as depicted in Figure \ref{fig:Leaf}. We have chosen, without loss of generality, to label the corners to the right of the edge $\kappa_{V/2-2}$ (the last label) and $\kappa_{V/2-3}$. On the row corresponding to $\kappa_{V/2-2}$, there is a single non-diagonal, non-trivial element,
\begin{equation}
\bigl[P_{\mathcal{T}}\bigr]_{V/2-2, \kappa'} = \delta_{\kappa', 0}\, Y_{h_1},
\end{equation}
corresponding to the adjacency with the marked corner and the half-edge $h_1$. Similarly there is a single non-diagonal, non-trivial element on the column $\kappa_{V/2-2}$, corresponding to the adjacency with the corner $\kappa_{V/2-3}$ and the half-edge $h_3$,
\begin{equation}
\bigl[P_{\mathcal{T}}\bigr]_{\kappa',V/2-2} = \delta_{\kappa', V/2-3}\, Y_{h_3}.
\end{equation}

Singling out the last column and last row which correspond to the corner $\kappa_{V/2-2}$, the following block structure is found
\begin{equation}
P_{\mathcal{T}} = \left( \begin{array}{ccc|c}
\multicolumn{3}{c}{\multirow{3}{*}{\scalebox{1}{$P'_{\mathcal{T}}$}}} & \\
 & & & C_{V/2-2} \\
 & & &  \\
 \hline
 & R_{V/2-2} & & \mathbbm{1}_{E_d}
\end{array}
\right)
\end{equation}
where we have introduced the column and row vectors
\begin{equation}
C_{V/2-2} = \begin{pmatrix} 0\\ \vdots \\ 0\\ Y_{h_3} \end{pmatrix} \qquad \text{and} \qquad R_{V/2-2} = \begin{pmatrix} Y_{h_1} & 0 & \dotsb &0 \end{pmatrix}.
\end{equation}
Therefore
\begin{equation}
\det P_{\mathcal{T}} =  \det \bigl(P'_{\mathcal{T}} - C_{V/2-2} R_{V/2-2}\bigr) 
\end{equation}
The matrix $P'_{\mathcal{T}} - C_{V/2-2} R_{V/2-2}$ is readily interpreted as a matrix $P_{\mathcal{T}'}$ of the same type for a new tree $\mathcal{T}'$ with one vertex less. Indeed, consider the tree obtained by removing the leaf and half-edges $h_1, h_2$ in Figure \ref{fig:TreeInduction}. The half-edge $h_3$ now receives $Y_{h_3} \to Y_{h_3} Y_{h_1}$. Moreover the matrix $Y_{h_2}$ is absorbed into $Y_0 \to Y_0 + Y_2$. The induction hypothesis can then be applied and directly gives
\begin{equation}
\det P_{\mathcal{T}} = \det \Bigl(Y_0 + \sum_{v\in\mathcal{T}} \prod_{h_v}^{\text{counter-clockwise}} Y_{h_v} \Bigr)
\end{equation}
Here $\sum_{v\in\mathcal{T}}$ corresponds to a sum over the vertices of $\mathcal{T}$ and $\{h_v\}$ are the half-edges incident to $v$.
\end{proof}


\paragraph*{Remark on the combinatorial encoding of data --} Already in \eqref{Zint1} which uses $G\in \mathbbm{G}_V$, all interactions are quartic. Therefore, already at this stage, the Hubbard-Stratonovich transformations could have been applied to each interaction. We think it is more interesting to perform them after the bijection to plane trees and use $\mathcal{T}\in\mathbbm{T}_V$, since all combinatorial data then refer to edges, half-edges, vertices and corners of $\mathcal{T}$.

\paragraph*{Reduction of the number of degrees of freedom --} The matrix $X_e\in \bigotimes_{c\in C_e} V^*_c\otimes V_c$ has 
\begin{equation}
N^{2|C_e|} \leq N^d
\end{equation}
degrees of freedom, which is crucially less than or equal to the number of degrees of freedom of $T$. For all $C_e$-bidipole insertions with $|C_e|<d/2$, the new degrees of freedom are fewer than in the original tensor formulation, meaning that the Hubbard-Stratonovich transformation on edge $e\in \mathcal{T}$ effectively integrates degrees of freedom.

Only for the insertions which are balanced in the number of colors, in the sense that $|C_e|=d/2$, for instance $C_e=\{1, 2, \dotsc, d/2\}$ in even dimensions, does this procedure not reduce the number of degrees of freedom. This is because in the corresponding quartic polynomial $Q_{C_e}$, a tensor splits its degrees of freedom into two equal sets of contractions with other tensors. There is therefore no obvious way to effectively integrate some degrees of freedom.

Had we exchanged $C\leftrightarrow \widehat{C}$, $X_e$ would have had $N^{2d - 2|C_e|}$ degrees of freedom, which is more than the original $N^d$ which we started with! This justifies the choice of splitting along $C_e$ to perform the Hubbard-Stratonovich transformations.

\subsection{Example}

Melonic cycles are formed by insertions of melonic dipoles always on edges of the same color. For melonic cycles used as interactions, the intermediate field of \cite{StuffedWalshMaps, PhDLionni} gives matrix models with matrices of size $N\times N$. This is the same with our new matrix models. Thus, our new matrix model does not really bring anything new in those cases.

A simple generalization of those melonic cycles are bubbles formed by recursive insertions of $C$-bidipoles for the same set $C$. The intermediate field of \cite{StuffedWalshMaps} then gives matrices of size $N^{|C|}\times N^{|C|}$ and again our new matrix model cannot improve on this.

However, our new matrix model is relevant whenever a GM bubble is formed by insertions of bidipoles with at least two distinct sets $C_1 \neq C_2$. The simplest case is at $d=3$ with $C_1 = \{1\}$ and $C_2 = \{2\}$, with $B$ and $G$ such that $\partial G = B$ as follows
\begin{equation}
B = \begin{array}{c} \includegraphics[scale=.4]{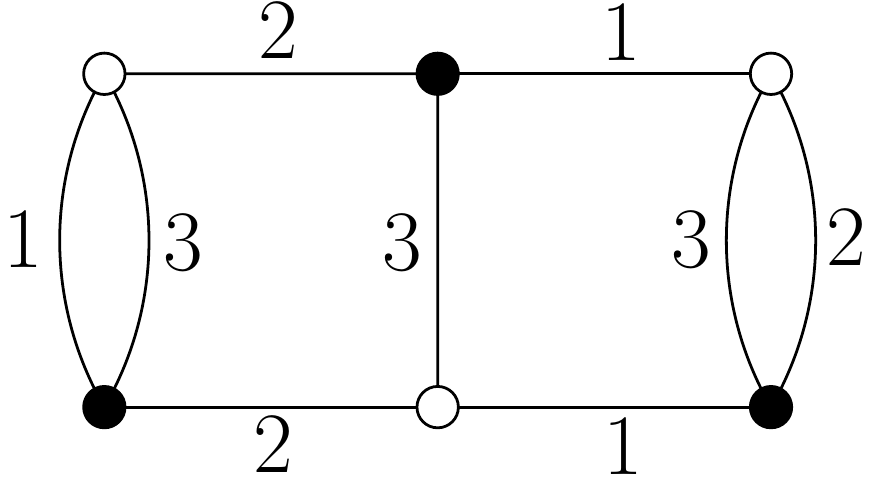} \end{array} \qquad
G = \begin{array}{c} \includegraphics[scale=.4]{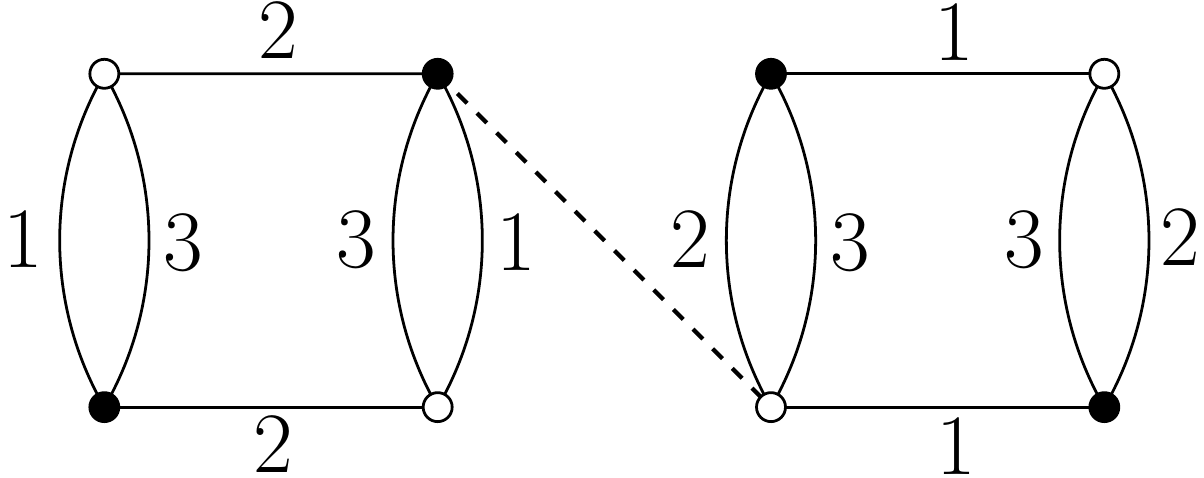} \end{array}
\end{equation}
We apply the map $J$ to get a tree representation,
\begin{equation}
\mathcal{T} = \begin{array}{c} \includegraphics[scale=.55]{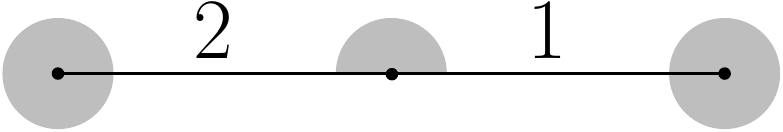} \end{array}
\end{equation}

Since it is a melonic bubble, we know from \cite{Uncoloring} that $s_B = -2(d-1)$. Theorem \ref{thm:MatrixModel} requires introducing two complex matrices $X_1\in V_1^*\otimes V_1$ and $X_2\in V_2^*\otimes V_2$, thus both of size $N\times N$. It then
gives
\begin{multline}
Z_N(t) = \int dX_1 dX_1^\dagger dX_2 dX_2^\dagger\ \exp -\tr_{V_1} (X_1 X_1^\dagger) - \tr_{V_2} (X_2 X_2^\dagger) \\
\times\ \exp - N \tr_{V_1\otimes V_2} \ln \Bigl(\mathbbm{1}\otimes \mathbbm{1} - N^{-(d-1)/2} t^{1/4}\bigl( X_1\otimes \mathbbm{1} + \mathbbm{1} \otimes X_2\bigr) + N^{-(d-1)} t^{1/2} X_1^\dagger \otimes X_2^\dagger\Bigr)
\end{multline}
The factor $N$ in front of the interaction comes from the trivial trace over $V_3$ which factorizes since there are no non-trivial operators acting on the third color. By expanding the logarithm, we can rewrite the interaction as a sum over all finite words $w$ over the alphabet $\{a, b, c\}$. More precisely, let $W$ be the set of finite words over the alphabet $\{a, b, c\}$ which have the same number $|w|$ of $a$, $b$, $c$. Then
\begin{multline}
\tr_{V_1\otimes V_2} \ln \Bigl(\mathbbm{1}\otimes \mathbbm{1} - N^{-(d-1)/2} t^{1/4}\bigl( X_1\otimes \mathbbm{1} + \mathbbm{1} \otimes X_2\bigr) + N^{-(d-1)} t^{1/2} X_1^\dagger \otimes X_2^\dagger\Bigr)\\
= \sum_{w\in W} \frac{1}{3|w|} (-t)^{|w|}\ N^{-2(d-1)|w|}\ \tr_{V_1\otimes V_2} (X_w)
\end{multline}
where $X_w$ is obtained by mapping $a\mapsto X_1\otimes \mathbbm{1}$, $b\mapsto \mathbbm{1}\otimes X_2$ and $c\mapsto X_1^\dagger\otimes X_2^\dagger$.

\subsection{Relation with another intermediate field model}

Another intermediate field model for random tensor models was presented in \cite{StuffedWalshMaps, PhDLionni}. The two obvious and crucial differences with the present intermediate field are: $i)$ it applies to arbitrary invariant interactions while ours only applies to GM interactions, $ii)$ it however generally increases the number of degrees of freedom, in sharp contrast with our present proposal.

In the case of GM interactions, both \cite{StuffedWalshMaps} and the present model apply equally. Since our model decreases the number of degrees of freedom, we expect to be able to derive it from the intermediate field of \cite{StuffedWalshMaps} by integrating some variables. We show how do that. First we must give some details about the model of \cite{StuffedWalshMaps}.

We start with a tree $\mathcal{T}$ such that $\partial J^{-1} \mathcal{T} = B$ as we have done here. Each vertex $v\in\mathcal{T}$ has a marked corner which we now equip with an additional half-edge and a matrix $\phi_v$. If $C_{e,v}$ denotes the color set of the edge $e$ of $\mathcal{T}$ incident to $v$, then denote $C_v = \bigcup_e C_{e,v}$ the set of all colors meeting at $v$. Then
\begin{equation}
\phi_v \in\ \bigotimes_{c\in C_v} V^*_c \otimes V_c.
\end{equation}
The dimension of this space is only bounded by $N^{2d}$ so it can increase the number of degrees of freedom of the tensor model. For instance in $d=3$, where GM interactions are just ordinary melonic interactions, it is sufficient to have two colors meeting at $v$ to get $\phi_v$ with $N^4$ degrees of freedom, compared to the $N^3$ of $T$ we started with. In our model we are guaranteed to have matrices $X, X^\dagger$ with $N^2$ degrees of freedom.

The intermediate field model of \cite{StuffedWalshMaps} gives a recipe to associate to $\mathcal{T}$ a function $\mathcal{I}_{\mathcal{T}}$ of the matrices $\{\phi_v\}_{v\in\mathcal{T}}$. It is obtained by considering the tensor product of all $\phi_v$ and taking partial traces as follows. Consider the forest $\mathcal{T}_c$ of color $c$ obtained by removing all edges whose color sets do not contain $c$. For each connected component, one follows the face and deduce from it a partial trace over the indices of color $c$ of all the matrices $\phi_v$ of that connected component. Doing so for each connected component and all colors we obtain the funtion $\mathcal{I}_{\mathcal{T}}(\{\phi_v\})$.

To glue those interactions together, a sort of hyper-propagator is required, i.e. hyperedges instead of edges, obtained by means of a logarithm. The action is
\begin{equation}
\sum_v \tr \phi_v \phi_v^\dagger + \tr \ln \Bigl(\mathbbm{1} - \frac{1}{N^{d-1}} \sum_v \tilde{\phi}_v^\dagger\Bigr) + N^s\, t\ \mathcal{I}_{\mathcal{T}}(\{\phi_v\}).
\end{equation}
Here $\tilde{\phi}_v^\dagger$ is the embedding of $\phi_v^\dagger$ into $E_d^*\otimes E_d$ by adding the identity on $V_c$ for $c\not\in C_v$. To get a unit coefficient in front of the interaction, we rescale $\phi_v$ and $\phi_v^\dagger$ and get
\begin{equation}
S_N(\{\phi_v, \phi_v^\dagger\},t) = \sum_v \tr \phi_v \phi_v^\dagger + \tr \ln \Bigl(\mathbbm{1} - N^{-d+1+2s/V} t^{V/2} \sum_v \tilde{\phi}_v^\dagger\Bigr) + \mathcal{I}_{\mathcal{T}}(\{\phi_v\}).
\end{equation}
The partition function is
\begin{equation}
Z_N(t) = \int \prod_v d\phi_v d\phi_v^\dagger\ \exp -S_N(\{\phi_v, \phi_v^\dagger\},t)
\end{equation}

We now show how to transform directly the above expression into the new formulation of Theorem \ref{thm:MatrixModel}. The first thing to do is to introduce the matrices of $X_e$ which appear in Theorem \ref{thm:MatrixModel}. Here it is done by ``cutting'' $\mathcal{I}_{\mathcal{T}}$ along the edges of $\mathcal{T}$. Consider an edge $e\subset \mathcal{T}$. It separates $\mathcal{T}$ into two trees $\mathcal{T}_1, \mathcal{T}_2$. Both trees also have marked half-edges $h_1, h_2$ which are the two halves of $e = \{h_1, h_2\}$. Then the Hubbard-Stratonovich transformation assigns them matrices $X_e = X_{h_1}$, and $X_{h_2}  = X_{e}^\dagger$, and splits the interaction into two parts,
\begin{equation}
\exp - \mathcal{I}_{\mathcal{T}}(\{\phi_v\}) = \int dX_{e} dX_{e}^\dagger\ \exp -\tr X_{e} X_{e}^\dagger - \mathcal{I}_{\mathcal{T}_1} (\{\phi_v\};X_{h_1}) + \mathcal{I}_{\mathcal{T}_2} (\{\phi_v\};X_{h_2})
\end{equation}
where $\mathcal{I}_{\mathcal{T}_i}(\{\phi_v\},X_{h_i})$ is calculated like $\mathcal{I}_{\mathcal{T}_i}$ with an additional insertion of $X_{h_i}$ on the marked half-edge $h_i$. This is the same procedure as in \eqref{HSTensor}. By applying it repeatedly, one splits $\mathcal{T}$ into a set of disjoint $V/2$ vertices
\begin{equation}
Z_N(t) = \int \prod_v d\phi_v d\phi_v^\dagger\ \prod_e dX_e dX_e^\dagger\ \exp -\sum_v \tr \phi_v \phi_v^\dagger - \tr \ln \Bigl(\mathbbm{1} - N^{-d+1+2s/V} t^{V/2} \sum_v \tilde{\phi}_v^\dagger\Bigr) + \sum_{v\in\mathcal{T}} \tr \phi_v \prod_{h_v}^{\substack{\text{counter-}\\ \text{-clockwise}}} \epsilon_{h_v} \tilde{X}_{h_v}
\end{equation}
where all $\epsilon_h=1$ except one which is $-1$.

We now use a matrix version of
\begin{equation}
\int dzd\bar{z}\ f(\bar{z})\ \exp-z\bar{z} + az  = f(a),
\end{equation}
which is itself just another version of Hubbard-Stratonovich \eqref{HS}. We apply it to the integrals over $\phi_v$, and $\phi_v^\dagger$ for each vertex, and $a$ is the counter-clockwise product of the $\epsilon_h \tilde{X}_h$ around each vertex. This leads directly to Theorem \ref{thm:MatrixModel}.

\subsection{Saddle points for the new matrix models}

Theorem \ref{thm:Gaussian} is based on a direct combinatorial analysis of the Feynman graphs. Since the result asserts that the models are Gaussian, we may expect to reproduce the equation on the covariance \eqref{LargeN2point} at large $N$ using a saddle point analysis of the new matrix model \eqref{NewMatrixModel}.

We consider the same situation as in Theorem \ref{thm:Gaussian}, translated into the plane tree representation of $B$. The tree $\mathcal{T}$ satifies $B = \partial J^{-1}\mathcal{T}$ with $B$ totally unbalanced. It gives $|C_e|<d/2$ for all edges $e$ of $\mathcal{T}$. In general, the matrices involved in the new matrix model \eqref{NewMatrixModel} do not commute, $[X_e,X_e']\neq 0$ whenever $C_e\cap C_{e'}\neq \emptyset$. Therefore, using eigenvalues will leave angular integrals which are difficult to perform.

Instead, we make some ansatz based on Theorem \ref{thm:Gaussian}. From a matrix model which becomes Gaussian at large $N$, we expect the Vandermonde determinant to be subdominant so that all eigenvalues fall into the potential well instead of spreading around it. This claim is supported by the direct analysis of the quartic melonic case in \cite{Quartic-Nguyen-Dartois-Eynard}. We make the ansatz
\begin{equation} \label{MatrixAnsatz}
X_e = (N^{s}t)^{-1/(V-2)}\,x_e\,\mathbbm{1}_{V_{C_e}},\qquad X^\dagger_e = (N^{s}t)^{-1/(V-2)}\,x^*_e\,\mathbbm{1}_{V_{C_e}}
\end{equation}
Recall that $X_e$ is associated naturally to a half-edge of $e$, say $h_1$, and $X_e^\dagger$ to the other half, say $h_2$. We thus use the notation $x_e = x_{h_1}$ and $x_e^* = x_{h_2}$ too. Here the prefactor $(N^{s}t)^{-1/(V-2)}$ is introduced to get rid of the explicit dependence in $N$ and $t$ in the logarithm in \eqref{NewMatrixModel}. However, keep in mind that we will need $x_e, x_e^*$ to also scale with $N$ so that all terms of the potential scale the same at large $N$.

It is clear that taking the variation of the potential with respect to $X_h$ and then plugging in the ansatz \eqref{MatrixAnsatz} is equivalent to plugging in the ansatz directly into the potential and taking the variation with respect to $x_h$. Therefore we focus on the latter approach and denote the potential as a function of $\{x_h\}$ as
\begin{equation}
V_{\mathcal{T},N}(\{x_h\}) = t^{-2/(V-2)} \sum_e N^{|C_e| - \frac{2s}{V-2}} x_e\, x_e^* + N^d \ln \Bigl(1 - \sum_{v\in\mathcal{T}} \prod_{h_v} \epsilon_{h_v} x_{h_v}\Bigr)
\end{equation}

\paragraph*{The tree $\mathcal{T}_h$ --} To analyze the potential $V_{\mathcal{T},N}(\{x_h\})$ we introduce for each half-edge $h$ the tree $\mathcal{T}_h\subset \mathcal{T}$ which is the subtree ``opposite'' to $h$ together with the edge containing $h$. More formally, partition $\mathcal{T}$ as the edge $e=\{hh'\}$, the subtree $\mathcal{T}_h^0$ incident to $h'$ and $\mathcal{T}^0_{h'}$ incident to $h$. Then $\mathcal{T}_h = e\cup \mathcal{T}_h^0$,
\begin{equation}
\begin{array}{c} \includegraphics[scale=.4]{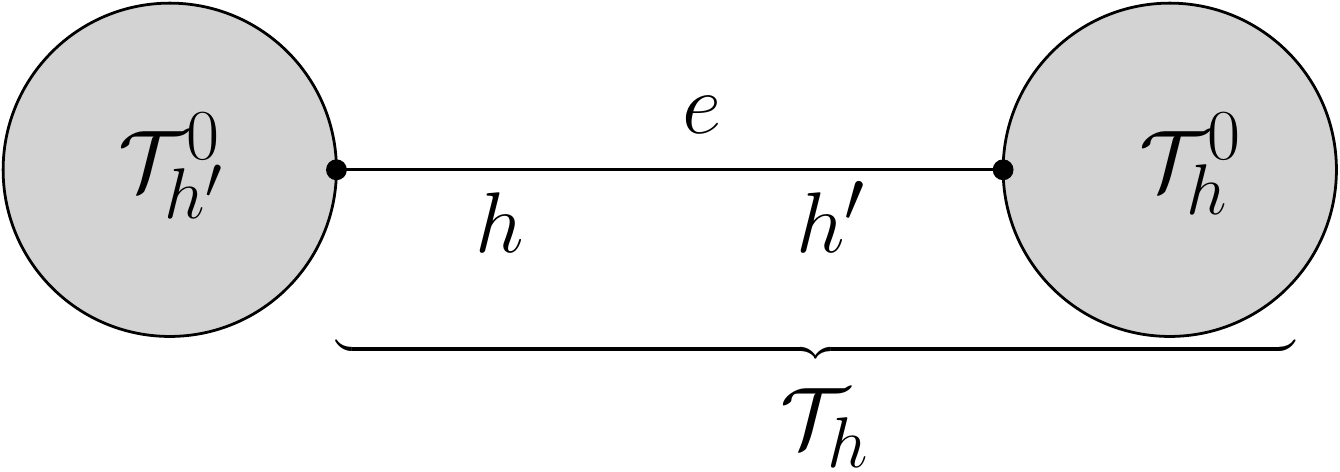} \end{array}
\end{equation}

\begin{proposition}
There is a unique set of rescaling coefficients $\{\eta_h\}$ such that all terms of $V(\{N^{\eta_h} x_h\})$ scale the same, i.e. 
\begin{equation} \label{ScalingConstraint}
V_{\mathcal{T}, N}(\{x_h\}) = N^d\, \hat{V}_{\mathcal{T}}(\{N^{-\eta_h}\,x_h\}) \qquad \text{with}\quad \hat{V}_{\mathcal{T}}(\{y_h\}) = t^{-2/(V-2)}\sum_{e\subset \mathcal{T}} y_e y_e^* + \ln \Bigl(1 - \sum_{v\in\mathcal{T}} \prod_{h_v} \epsilon_{h_v} y_{h_v}\Bigr)
\end{equation}
being independent of $N$.
%
Denoting $E(\mathcal{T}_h)$ the number of edges of $\mathcal{T}_h$, it is given by
\begin{equation} \label{Rescaling}
\eta_h = \Bigl(d + \frac{2s}{V-2}\Bigr) E(\mathcal{T}_h) - \sum_{e\subset \mathcal{T}_h} |C_e|
\end{equation}
\end{proposition}

\begin{proof}
Introduce $y_h = N^{-\eta_h} x_h$ and consider the Taylor expansion of \eqref{ScalingConstraint} around the point where $x_h=0$ for all half-edges. The term corresponding to an edge, i.e. $x_e x_e^*$, gives the constraint
\begin{equation} \label{EdgeConstraint}
\forall e\subset\mathcal{T} \qquad N^{|C_e| - \frac{2s}{V-2}}\, x_e\, x_e^* = N^d\, y_e y_e^* 
\end{equation}
Similarly, the term corresponding to a vertex gives
\begin{equation} \label{VertexConstraint}
\forall v\in\mathcal{T} \qquad \prod_{h_v} x_{h_v} = \prod_{h_v} y_{h_v}.
\end{equation}

We first prove that if a rescaling satisfying all constraints \eqref{EdgeConstraint}, \eqref{VertexConstraint} exists, then it is unique. Indeed, start from the half-edge $h$ incident to a leaf of $\mathcal{T}$. There the constraint \eqref{VertexConstraint} gives $y_h=x_h$, then \eqref{EdgeConstraint} enforces $x_{h'}=N^{d + \frac{2s}{V-2} - |C_e|} y_{h'}$ for $e=\{hh'\}$.

Then consider a vertex incident to edges $e_1, \dotsc, e_q$ connected to leaves except one half-edge $h$. The variables $y_{e_i}, y^*_{e_i}$ are all determined uniquely and the constraint \eqref{VertexConstraint} then determines $y_h$. Repeating those steps, we see that a trivial induction proves uniqueness, provided existence is ensured.

To prove existence, it is enough to show that \eqref{Rescaling} is a solution to \eqref{EdgeConstraint} and \eqref{VertexConstraint}. For \eqref{EdgeConstraint}, one needs to observe that summing over all edges of $\mathcal{T}_h$ and $\mathcal{T}_{h'}$ is like summing over all edges of $\mathcal{T}$ and counting $e$ twice, so that
\begin{equation}
\eta_h + \eta_{h'} = \Bigl(d + \frac{2s}{V-2}\Bigr) \bigl(E(\mathcal{T})+1\bigr) - \sum_{e'\subset \mathcal{T}} |C_{e'}| - |C_e| 
\end{equation}
This greatly simplifies with the value of $s$ given in Theorem \ref{thm:Scaling},
\begin{equation}
s = -\frac{d(V-2)}{2} + \sum_{e'\subset \mathcal{T}} |C_e|
\end{equation}
and $E(\mathcal{T}) = (V-2)/2$ making the relation between the number of vertices of the bubble and the number of edges of the tree, so that
\begin{equation}
\eta_h + \eta_{h'} = d + \frac{2s}{V-2} - |C_e|,
\end{equation}
which gives \eqref{EdgeConstraint}. The equation \eqref{VertexConstraint} is verified similarly: if $v$ is a vertex with incident half-edges $h_1, \dotsc, h_{d(v)}$, then observe that $\bigcup_{i=1}^{d(v)} \mathcal{T}_{h_i} = \mathcal{T}$, so that with the above value of $s$,
\begin{equation}
\sum_{i=1}^{d(v)} \eta_{h_i} = \Bigl(d + \frac{2s}{V-2}\Bigr) E(\mathcal{T}) - \sum_{e\subset \mathcal{T}} |C_e| = 0.
\end{equation}
\end{proof}

\begin{proposition} \label{thm:SaddlePoint}
The new matrix model admits the saddle point
\begin{equation}
X_h = t^{-1/(V-2)}\,N^{\eta_h - \frac{s}{V-2}}\, y_h(t)\,\mathbbm{1}_{V_{C_h}} \qquad \text{with} \quad
y_h(t) = \frac{1}{\epsilon_h} \biggl(\prod_{e\subset \mathcal{T}_h} \epsilon_e\biggr)\ \Bigl(t^{\frac{2}{V-2}}\, W_{\mathcal{T}}(t)\Bigr)^{E(\mathcal{T}_h)},
\end{equation}
where $W_{\mathcal{T}}(t)$ satisfies the equation of the large $N$, 2-point function
\begin{equation}
W_{\mathcal{T}}(t) = 1 - t\,W_{\mathcal{T}}(t)^{V/2}.
\end{equation}
\end{proposition}

\begin{proof}
From the previous proposition, the potential takes the form $V_{\mathcal{T},N}(\{x_h\}) = N^d\,\hat{V}_{\mathcal{T}}(\{y_h\})$. It comes in particular that there are no Vandermonde contributions at large $N$, since the Vandermonde for the matrices $X_e, X_e^\dagger$ (on non-coinciding eigenvalues) would scale like $\sum_{1\leq i< j\leq N^{|C_e|}} \ln |x_{e,i} - x_{e,j}| \sim N^{2|C_e|}$ in the exponential and $|C_e|<d/2$ for all edges.

Let $e = \{h_1 h_2\}$ be an edge in $\mathcal{T}$ made of the half-edges $h_1, h_2$, incident to $v_1$ and $v_2$ respectively,
\begin{equation}
\begin{array}{c} \includegraphics[scale=.5]{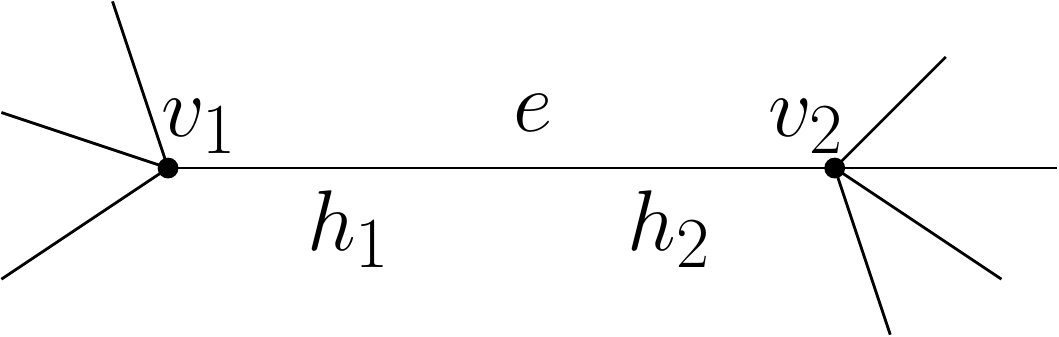} \end{array}
\end{equation}
The saddle point equation with respect to $h_2$ is $\partial \hat{V}_{\mathcal{T}}(\{y_h\})/\partial y_{h_2} = 0$ and gives
\begin{equation}
y_{h_1} = \epsilon_{h_2} t^{\frac{2}{V-2}} \hat{W}_{\mathcal{T}}(\{y_h\}) \prod_{h_{v_2} \neq h_2} \epsilon_{h_{v_2}} y_{h_{v_2}}
\end{equation}
where the product is over all half-edges incident to $v_2$ except $h_2$ and the function $\hat{W}_{\mathcal{T}}(\{y_h\})$ is
\begin{equation}
\hat{W}_{\mathcal{T}}(\{y_h\}) = \frac{1}{1-\sum_{v\in\mathcal{T}} \prod_{h_v} \epsilon_{h_v} y_{h_v}}.
\end{equation}
If $v_2$ is a leaf, the empty product is 1. A simple induction then shows that $y_{h_1}$ can be written as a product over the edges of $\mathcal{T}_{h_1}$,
\begin{equation}
y_{h_1} = \frac{1}{\epsilon_{h_1}} \prod_{e\subset \mathcal{T}_{h_1}} \Bigl( \epsilon_e\,t^{\frac{2}{V-2}} \hat{W}_{\mathcal{T}}(\{y_h\})\Bigr) = \frac{\prod_{e\subset \mathcal{T}_{h_1}} \epsilon_e}{\epsilon_{h_1}} \Bigl(t^{\frac{2}{V-2}} \hat{W}_{\mathcal{T}}(\{y_h\})\Bigr)^{E(\mathcal{T}_{h_1})}
\end{equation}

Then we need to determine the value of $\hat{W}_{\mathcal{T}}(\{y_h\})$ at the saddle point. We use the above equation for all half-edges $h_v$ incident to a given vertex $v$,
\begin{equation}
\prod_{h_v} \epsilon_{h_v} y_{h_v} = \prod_{h_v} \Bigl(\prod_{e\subset \mathcal{T}_{h_v}} \epsilon_{e}\Bigr) \ \Bigl(t^{\frac{2}{V-2}} \hat{W}_{\mathcal{T}}(\{y_h\})\Bigr)^{\sum_{h_v} E(\mathcal{T}_{h_1})} = \Bigl(\prod_{e\subset \mathcal{T}}\epsilon_e \Bigr) \Bigl(t^{\frac{2}{V-2}} \hat{W}_{\mathcal{T}}(\{y_h\})\Bigr)^{E(\mathcal{T})} = - t\,\hat{W}_{\mathcal{T}}(\{y_h\})^{\frac{V-2}{2}}
\end{equation}
where again we have used $\bigcup_{h_v} \mathcal{T}_{h_v} = \mathcal{T}$ and $E(\mathcal{T}) = (V-2)/2$, and the fact that the product of the $\epsilon_e$ is $-1$. This leads to 
\begin{equation}
\hat{W}_{\mathcal{T}}(\{y_h\}) = 1 - t\ \hat{W}_{\mathcal{T}}(\{y_h\})^{V/2}
\end{equation}
meaning that $\hat{W}_{\mathcal{T}}(\{y_h\}) = W_{\mathcal{T}}(t)$ at the saddle point is determined by $t$ only..
\end{proof}

\section*{Acknowledgements}

This research was supported by the ANR MetAConc project ANR-15-CE40-0014.


\end{document}